\documentclass{amsart}

\usepackage{adjustbox}
\usepackage{amsmath}	
\usepackage{amssymb}
\usepackage{graphicx}
\usepackage{mathrsfs}	
\usepackage{multicol}
\usepackage{soul}	
\usepackage{color}	
\usepackage[titletoc,title]{appendix}	
\usepackage{eufrak}
\usepackage{hyperref}	
\usepackage{subfig}

\usepackage[backend=bibtex,style=numeric,citestyle=numeric,maxbibnames=99]{biblatex}
\renewbibmacro{in:}{\ifentrytype{article}{}{\printtext{\bibstring{in}\intitlepunct}}}
 
\newcommand{\de}{\ensuremath{\partial}}
\newcommand{\dee}{\ensuremath{\textrm{d}}}
\newcommand{\pd}[2]{\ensuremath{ \frac{\de #1}{\de #2}}}

\newcommand{\fdf}[1]{\ensuremath{ \frac{\dee}{\dee #1}}}

\newcommand{\inty}[4]{\ensuremath{ \int_{#1}^{#2} \! #3 \, \dee#4 }}

\newcommand{\field}[1]{\mathbb{#1}}
\newcommand{\ip}[2]{\ensuremath{ \left< #1 | #2 \right> } }
\newcommand{\ipy}[2]{\ensuremath{ \left< #1 | #2 \right>_{L^2_y(\field{R}^d)} } }
\newcommand{\ipz}[2]{\ensuremath{ \left< #1 | #2 \right>_{L^2_z(\Omega)}}}

\newtheorem{definition}{Definition}
\newtheorem{assumption}{Assumption}
\newtheorem{remark}{Remark}
\newtheorem{theorem}{Theorem}
\newtheorem{lemma}{Lemma}
\newtheorem{corollary}{Corollary}
\newtheorem{proposition}{Proposition}

\let\oldhat\hat
\renewcommand{\hat}[1]{\oldhat{\boldsymbol{#1}}}
\bibliography{bibliography.bib}

\graphicspath{{images/}}

\numberwithin{lemma}{section}
\numberwithin{proposition}{section}
\numberwithin{equation}{section}
\numberwithin{theorem}{section}
\numberwithin{corollary}{section}
\numberwithin{remark}{section}
\numberwithin{definition}{section}
\numberwithin{assumption}{section}

\allowdisplaybreaks

\title[Effective particle-field dynamics and Berry curvature]{Wavepackets in inhomogeneous periodic media: effective particle-field dynamics and Berry curvature}

\author{Alexander Watson, Jianfeng Lu, and Michael I. Weinstein}

\begin{document}

\captionsetup[subfigure]{labelformat=empty}

\maketitle

\begin{abstract}
We consider a model of an electron in a crystal moving under the influence of an external electric field: Schr\"{o}dinger's equation with a potential which is the sum of a periodic function and a general smooth function. We identify two dimensionless parameters: (re-scaled) Planck's constant and the ratio of the lattice spacing to the scale of variation of the external potential. We consider the special case where both parameters are equal and denote this parameter $\epsilon$. In the limit $\epsilon \downarrow 0$, we prove the existence of solutions known as semiclassical wavepackets which are asymptotic up to `Ehrenfest time' $t \sim \ln 1/\epsilon$. To leading order, the center of mass and average quasi-momentum of these solutions evolve along trajectories generated by the classical Hamiltonian given by the sum of the Bloch band energy and the external potential. We then derive all corrections to the evolution of these observables proportional to $\epsilon$. The corrections depend on the gauge-invariant Berry curvature of the Bloch band, and a coupling to the evolution of the wave-packet envelope which satisfies Schr\"{o}dinger's equation with a time-dependent harmonic oscillator Hamiltonian. This infinite dimensional coupled `particle-field' system may be derived from an `extended' $\epsilon$-dependent Hamiltonian.
 It is known that such coupling of observables (discrete particle-like degrees of freedom) to the wave-envelope (continuum field-like degrees of freedom) can have a significant impact on the overall dynamics. \end{abstract}

\newpage

\tableofcontents

\newpage

\subsection{Notation and conventions} \label{notation_and_conventions}
\begin{itemize}
\item Where necessary to avoid ambiguity we will use index notation, making the standard convention that repeated indices are summed over from $1$ to $d$. Thus, in the expression: 
\begin{equation} \label{eq:long_term}
	\de_{p_\alpha} \de_{p_\beta} f(p) ( - i \de_{y_\alpha}) (- i \de_{y_\beta}) g(y), 
\end{equation}
it is understood that we are summing over $\alpha \in \{1,...,d\}, \beta \in \{1,...,d\}$. 
\item Where there is no danger of confusion we will use the standard conventions: 
\begin{equation}
	v_\beta w_\beta = v \cdot w, v_\beta v_\beta = v \cdot v
\end{equation}
\item $\Delta_x = \nabla_x^2$ denotes the $d-$dimensional Laplacian
\item $D^2_x$ denotes the $d-$dimensional Hessian matrix with respect to $x$
\item We will adopt multi-index notation where appropriate so that:
\begin{equation}
	\sum_{|\alpha| = l} | \de_x^\alpha f(x) | \in L^\infty(\field{R}^d) 
\end{equation}  
means all derivatives of order $l$ of $f(x)$ are uniformly bounded.
\item It will be useful to introduce the energy spaces for every $l \in \field{N}$:
\begin{equation} \label{eq:sigma_l_spaces}
	\Sigma^l(\field{R}^d) := \left\{ f \in L^2(\field{R}^d) : \| f \|_{\Sigma^l} := \sum_{|\alpha| + |\beta| \leq l} \| y^\alpha (- i \de_y)^\beta f(y) \|_{L^2_y} < \infty, \right\}  
\end{equation}
\item The space of Schwartz functions $\mathcal{S}(\field{R}^d)$ is the space of functions defined as:
\begin{equation}
	\mathcal{S}(\field{R}^d) := \cap_{l \in \field{N}} \Sigma^l(\field{R}^d). 
\end{equation}
\item We will refer throughout to the space of $L^2$-integrable functions which are periodic on the lattice $\Lambda$:
\begin{equation}
	L^2_{per} := \left\{ f \in L^2_{loc}(\field{R}^d) : \text{ for all } z \in \field{R}^d, v \in \Lambda, f(z+v) = f(z) \right\}.
\end{equation}
\item We will write a fundamental period cell in $\field{R}^d$ of the lattice $\Lambda$ as $\Omega$.
\item We will make use of the Sobolev norms on a fundamental period cell $\Omega$ for integers $s \geq 0$:
\begin{equation}
	\| f(z) \|_{H^s_z} := \sum_{|j| \leq s} \| (\de_z)^j f(z) \|_{L^2_z}.
\end{equation}
\item It will also be useful to introduce the `shifted' Sobolev norms, for arbitrary $p \in \field{R}^d$:
\begin{equation}
	\| f(z) \|_{H^s_{z,p}} := \sum_{|j| \leq s} \| (p - i \de_z)^j f(z) \|_{L^2_z}
\end{equation}
\item Define the dual lattice to $\Lambda$:
\begin{equation}
	\Lambda^* := \left\{ b \in \field{R}^d : \exists v \in \Lambda : b \cdot v = 2 \pi n, n \in \field{Z} \right\}
\end{equation}
\item We will refer to a fundamental cell in $\field{R}^d_k$ of the dual lattice $\Lambda^*$ as the Brillouin zone, or $\mathcal{B}$
\item We make the standard convention for the $L^2$-inner product:
\begin{equation}
	\ip{ f }{ g }_{L^2(\mathcal{D})} := \inty{\mathcal{D}}{}{ \overline{f(x)} g(x) }{x} 
\end{equation}
\item We will make the conventions:
\begin{equation}
\begin{split}
	&f^\epsilon(x,t) = O(\epsilon^K e^{c t}) \iff \exists c > 0, C > 0, \text{ independent of $t, \epsilon$ such that } \| f^\epsilon(x,t) \|_{L^2_x} \leq C \epsilon^K e^{c t} 	\\
	&g^\epsilon(t) = O(\epsilon^K e^{c t}) \iff \exists c > 0, C > 0, \text{ independent of $t, \epsilon$ such that } | g^\epsilon(t) | \leq C \epsilon^K e^{c t} 	\\
\end{split}
\end{equation} 
\item Let $A$ be a complex matrix. Then we will write $A^T$ for its transpose, $\overline{A}$ for its complex conjugate, and $\text{Tr} A$ for its trace. Using index notation:
\begin{equation}
	(A_{\alpha \beta})^T := A_{\beta \alpha}, \quad (\overline{A})_{\alpha \beta} := \overline{ A_{\alpha \beta} }, \quad \text{Tr} A := A_{\alpha \alpha}
\end{equation}
\item The Kronecker delta $\delta_{\alpha \beta}$ is defined:
\begin{equation} \label{eq:kronecker_delta}
	\delta_{\alpha \beta} = \begin{cases} +1 &\text{when } \alpha = \beta \\ 0 &\text{when } \alpha \neq \beta \end{cases} 
\end{equation}
\item In dimension $d = 3$, the Levi-Civita symbol $\varepsilon_{\alpha \beta \gamma}$ is defined:
\begin{equation} \label{eq:levi_civita_symbol}
	\varepsilon_{\alpha \beta \gamma} := \begin{cases} +1 &\text{when }(\alpha, \beta, \gamma) \in \{ (1,2,3), (2,3,1), (3,1,2) \} \\ -1 &\text{when } (\alpha, \beta, \gamma) \in \{ (1,3,2), (2,3,1), (3,2,1) \}  \\ 0 &\text{when } \alpha = \beta, \beta = \gamma, \text{ or } \gamma = \alpha \end{cases}  	
\end{equation} 
and satisfies the identities:
\begin{equation} 
\begin{split}
	&\varepsilon_{\alpha \beta \gamma} = \varepsilon_{\beta \gamma \alpha} = \varepsilon_{\gamma \alpha \beta}	\\
	&\varepsilon_{\phi \alpha \beta} \varepsilon_{\phi \gamma \phi} = \delta_{\alpha \gamma} \delta_{\beta \phi} - \delta_{\beta \gamma} \delta_{\alpha \phi}.
\end{split}
\end{equation}
The cross product of 3-vectors $v, w$ may then be written:
\begin{equation} \label{eq:cross_product}
	( v \times w )_{\alpha} = \epsilon_{\alpha \beta \gamma} v_\beta w_\gamma
\end{equation}
\end{itemize}

\newpage

\section{Introduction}
\noindent In this work we study the non-dimensionalized time-dependent Schr\"{o}dinger equation for $\psi^\epsilon (x,t): \field{R}^d \times [0,\infty) \rightarrow \field{C}$:
\begin{equation} \label{eq:original_equation}
\begin{split}
	&i \epsilon \de_t \psi^\epsilon = - \frac{1}{2} \epsilon^2 \Delta_x \psi^\epsilon + V\left(\frac{x}{\epsilon}\right) \psi^\epsilon + W(x) \psi^\epsilon \\
	&\psi^\epsilon(x,0) = \psi^\epsilon_0(x).
\end{split}
\end{equation} 
Here, $\epsilon$ is a positive real parameter which we assume to be small: $\epsilon \ll 1$. We assume throughout that the function $V$ is smooth and periodic with respect to a $d$-dimensional lattice $\Lambda$ so that:  
\begin{equation}
	V(z + v) = V(z) \text{ for all $v \in \Lambda, z \in \field{R}^d$},
\end{equation}
and that $W$ is smooth with all derivatives uniformly bounded. Equation (\ref{eq:original_equation}) is a well-studied model in condensed matter physics of the dynamics of an electron in a crystal under the independent-particle approximation \cite{ashcroft_mermin}, whose periodic effective potential due to the atomic nuclei is specified by $V$, under the influence of a `slowly varying' external electric field generated by $W$.

In this work we rigorously derive a family of explicit asymptotic solutions of (\ref{eq:original_equation}) known as \emph{semiclassical wavepackets}. We then derive the equations of motion of the center of mass and average quasi-momentum of these solutions, including corrections proportional to $\epsilon$. 

At order $\epsilon^0$, the mean position and momentum of the semi-classical wavepacket evolve along the classical trajectories associated with the `Bloch band' Hamiltonian $\mathcal{H}_n := E_n(p)+W(q)$, where $p\mapsto E_n(p)$ is the dispersion relation associated with the $n^{th}$ spectral (Bloch) band of the periodic Schr\"{o}dinger operator $-\frac12\Delta_z+V(z)$.
The order $\epsilon$ corrections to the leading order equations of motion depend on the gauge-invariant Berry curvature of the Bloch band and the wavepacket envelope. Through order $\epsilon$, the system governing appropriately defined mean position $\mathcal{Q}^\epsilon(t)$, mean momentum $\mathcal{P}^\epsilon(t)$ and wave-amplitude profile $a^\epsilon(y,t)$ is a closed system of Hamiltonian type (Theorem \ref{prop:center_of_mass_proposition}). When $d = 3$, this system takes the form:  
\begin{equation} \label{system}
\begin{split}
	&\begin{matrix} \dot{{\mathcal{Q}^\epsilon}}(t) \\ \dot{\mathcal{P}^\epsilon}(t) \end{matrix} \begin{matrix} = \\ = \end{matrix} \underbrace{\begin{matrix} \nabla_{\mathcal{P}^\epsilon} \mathcal{H}_n(\mathcal{Q}^\epsilon(t),{\mathcal{P}^\epsilon}(t)) \\ - \nabla_{\mathcal{Q}^\epsilon} \mathcal{H}_n(\mathcal{Q}^\epsilon(t),{\mathcal{P}^\epsilon}(t)) \end{matrix}}_{\substack{\text{Dynamics generated by } \\ \text{`Bloch band' Hamiltonian} \\ \mathcal{H}_n :=  E_n(\mathcal{P}^\epsilon) + W({\mathcal{Q}^\epsilon})}} \begin{matrix} + \\ + \end{matrix} \begin{matrix} \epsilon \left\{ \vphantom{\frac{1}{2}} \right. \\ \epsilon \left\{\vphantom{\frac{1}{2}} \right. \end{matrix} \underbrace{\begin{matrix} - \dot{\mathcal{P}^\epsilon}(t) \times \nabla_{\mathcal{P}^\epsilon} \times \mathcal{A}_n(\mathcal{P}^\epsilon(t)) \\ \vphantom{\frac{1}{2}} \end{matrix}}_{\substack{\text{Anomalous velocity due to } \\ \text{Berry curvature } }} \underbrace{\begin{matrix} + C_1[a^\epsilon](t) \\ + C_2[a^\epsilon](t) \end{matrix}}_{\substack{\text{`Particle-field'} \\ \text{coupling to} \\ \text{envelope $a^\epsilon$}}} \begin{matrix} \left.\vphantom{\frac{1}{2}} \right\} \\ \left. \vphantom{\frac{1}{2}} \right\} \end{matrix},	\\
	&i \de_t a^\epsilon = \mathscr{H}^\epsilon_n(t) a^\epsilon; \quad \mathscr{H}^\epsilon_n(t) := \underbrace{- \frac{1}{2} \nabla_y \cdot D^2_{\mathcal{P}^\epsilon} E_n(\mathcal{P}^\epsilon(t)) \nabla_y + \frac{1}{2} y \cdot D^2_{\mathcal{Q}^\epsilon} W(\mathcal{Q}^\epsilon(t)) y}_{\substack{\text{Quantum harmonic oscillator Hamiltonian} \\ \text{with parametric forcing through $\mathcal{Q}^\epsilon(t),\mathcal{P}^\epsilon(t)$}}}.
\end{split}
\end{equation}
Here, $D^2_{\mathcal{P}^\epsilon} E_n, D^2_{\mathcal{Q}^\epsilon} W$ denote Hessian matrices and $\mathcal{A}_n$ is the Berry connection (\ref{eq:berry_connection}). For the explicit forms of $C_1[a], C_2[a]$ and the generalization of (\ref{system}) to arbitrary dimensions $d \geq 1$, see (\ref{eq:eom_for_center_of_mass}). The derivation of the form of the anomalous velocity displayed in (\ref{system}) is given in Remark \ref{rem:anomalous_velocity_cross_products}. 

The `particle-field' dynamical system \eqref{system} appears to be new, and contains terms which are not accounted for in the works of Niu et al. \cite{xiao_chang_niu}. The system reduces, in the case of Gaussian initial data and zero periodic background $V = 0$, to that presented in Proposition 4.4 of \cite{ohsawa} (see also \cite{ohsawa_leok}). 

The asymptotic solutions and effective Hamiltonian system \eqref{system} provide an approximate description of the dynamics of the full PDE (\ref{eq:original_equation}) up to `Ehrenfest time' $t \sim \ln 1/\epsilon$, known to be the general limit of applicability of wavepacket, or coherent state, approximations \cite{schubert_vallejos_toscano}. The validity of the approximation relies on an extension of the result of Carles and Sparber \cite{carles_sparber} (Theorem \ref{th:periodic_background_nonseparable_isolated_band_theorem})

Our methods are applicable when the wavepacket is spectrally localized in a Bloch band which has crossings (degeneracies), as long as the distance in phase space between the average quasi-momentum of the wavepacket and any crossing is uniformly bounded below independent of $\epsilon$ (see Assumption \ref{isolated_band_assumption_nonseparable_case}). We do not attempt a description of wavepacket dynamics when this distance $\downarrow 0$ (propagation through a band crossing), or at an avoided crossing where the separation between bands is proportional to $\epsilon$. We believe that both of these cases may be studied by adapting the work of Hagedorn and Joye \cite{hagedorn}, \cite{hagedorn_joye} on wavepacket dynamics in the Born-Oppenheimer approximation of molecular dynamics to the model (\ref{eq:original_equation}). 

Our methods are also applicable, with some modifications (see Section \ref{sec:discussion_of_results}), to potentials with the general two-scale form $U\left(\frac{x}{\epsilon},x\right)$ where $U$ is periodic in its first argument: 
\begin{equation}
	U(z + v,x) = U(z,x) \text{ for all } z, x \in \field{R}^d, v \in \Lambda
\end{equation}
and $U(z,x)$ is `nonseparable', i.e., cannot be written as the sum of a periodic potential $V(z)$ and an `external' potential $W(x)$. For details, see \cite{thesis}. For ease of presentation we consider in this work only the `separable' case (\ref{eq:original_equation}).  

The semiclassical wavepacket ansatz was introduced by Heller \cite{heller} and Hagedorn \cite{hagedorn_coherent_states_1} to study the uniform background case ($V = 0$) of (\ref{eq:original_equation}). See also related work on Gaussian beams \cite{ralston}. Hagedorn then extended this theory to the case where the potential $W(x)$ is replaced by an $x$-dependent operator in his study of the Born-Oppenheimer approximation of molecular dynamics \cite{hagedorn}. Semiclassical wavepacket solutions of (\ref{eq:original_equation}) in the periodic background case ($V \neq 0$) were then constructed by Carles and Sparber \cite{carles_sparber}. 

The anomalous velocity term in \eqref{system} was first derived by Karplus and Luttinger \cite{1954KarplusLuttinger}. For a derivation in terms of Berry curvature of the Bloch band, see Chang and Niu \cite{chang_niu}. It was then derived rigorously by Panati, Spohn, and Teufel \cite{panati_spohn_teufel_1} (see also \cite{e_lu_yang}\cite{xiao_chang_niu}). This term is responsible for the `intrinsic contribution' to the anomalous Hall effect which occurs in solids with broken time-reversal symmetry (see \cite{nagaosa_sinova_onoda_macdonald_ong} and references therein). The anomalous velocity due to Berry curvature is better known in optics as the spin Hall effect of light and was experimentally observed in \cite{bliokh_niv_kleiner_hasman}.

\subsection{Dimensional analysis, derivation of equation (\ref{eq:original_equation})} 
In this section we derive the non-dimensionalized equation (\ref{eq:original_equation}) starting with the Schr\"{o}dinger equation in physical units: 
\begin{equation} \label{eq:equation_in_physical_units}
	i \hbar \de_t \psi = - \frac{ \hbar^2 }{ 2 m } \Delta_x \psi + V(x) \psi + W(x) \psi 
\end{equation}
where $\hbar$ is the reduced Planck constant, and $m$ is the mass of an electron. This analysis is based on those given in \cite{e_lu_yang} \cite{bechouche_mauser_poupaud}. Define $l$ as the lattice constant, and let $\tau$ denote the quantum time scale:
\begin{equation}
	\tau = \frac{ m l^2 }{ \hbar }. 
\end{equation}
Let $L, T$ denote macroscopic length and time-scales. We assume that the periodic potential $V$ acts on the `fast quantum scale' and the $W$ acts on the `slow macroscopic scale':
\begin{equation}
	V(x) = \frac{m l^2}{\tau^2} \tilde{V}\left(\frac{x}{l}\right), W(x) = \frac{ m L^2 }{ T^2 } \tilde{W}\left(\frac{x}{L}\right).
\end{equation} 
After re-scaling $x, t$ by the macroscopic length and time-scales:
\begin{equation}
 	\tilde{x} := \frac{x}{L}, \tilde{t} := \frac{t}{T}, \tilde{\psi}( \tilde{x},\tilde{t} ) := \psi(x,t),
\end{equation}
(\ref{eq:equation_in_physical_units}) becomes:
\begin{equation} \label{eq:rescaled_equation}
	\frac{ i \hbar }{ T } \de_{\tilde{t}} \tilde{\psi} = - \frac{\hbar^2}{2 m L^2} \Delta_{\tilde{x}} \tilde{\psi} + \frac{ m l^2 }{ \tau^2 } \tilde{V}\left(\frac{L \tilde{x}}{l}\right) \psi + \frac{m L^2 }{T^2}\tilde{W}\left(\tilde{x}\right) \psi.
\end{equation}
We now identify two dimensionless parameters. Let $h$ denote a scaled Planck's constant, and $\epsilon$ the ratio of the lattice constant to the macroscopic scale:
\begin{equation}
	h := \frac{ \hbar T }{ m L^2 }, \epsilon := \frac{ l }{ L }.
\end{equation}
Writing (\ref{eq:rescaled_equation}) in terms of $h, \epsilon$ and dropping the tildes we arrive at:
\begin{equation}
	i h \de_t \psi^{h,\epsilon} = - \frac{ h^2 }{ 2 } \Delta_{x} \psi^{h,\epsilon}  + \frac{h^2}{\epsilon^2} V\left(\frac{x}{\epsilon}\right) \psi^{h,\epsilon} + W(x) \psi^{h,\epsilon} 
\end{equation}
where we have written $\psi^{h,\epsilon}(x,t)$ to emphasize the dependence of the solution on both parameters. We obtain the problem depending only on $\epsilon$ (\ref{eq:original_equation}) by setting $h = \epsilon$. Therefore, the limit $\epsilon \downarrow 0$ in (\ref{eq:original_equation}) corresponds to sending to zero the ratio of the lattice spacing $l$ to the scale of inhomogeneity $L$ and Planck's constant (appropriately re-scaled) to zero at the same rate. 

\subsection{Statement of results} \label{sec:statement_of_results}
In order to state our results we require some background on the spectral theory of the Schr\"{o}dinger operator:
\begin{equation} \label{eq:periodic_operator}
	H := - \frac{1}{2} \Delta_z + V(z)
\end{equation}
where $V$ is periodic with respect to a $d$-dimensional lattice $\Lambda$ \cite{kuchment} \cite{reed_simon_4}. Let $\Lambda^*$ denote the dual lattice to $\Lambda$, and define the first Brillouin zone $\mathcal{B}$ to be a fundamental period cell. Consider the family of self-adjoint eigenvalue problems parameterized by $p \in \mathcal{B}$:
\begin{equation} \label{eq:reduced_eigenvalue_problem}
\begin{split}
	&H(p) \chi(z;p) = E(p) \chi(z;p)	\\	
	&\chi(z + v;p) = \chi(z;p) \text{ for all } z \in \field{R}^d, v \in \Lambda	\\
	&H(p) := \frac{1}{2} ( p - i \nabla_z )^2 + V(z).
\end{split}
\end{equation}
For fixed $p$, known as the quasi-momentum, the spectrum of the operator (\ref{eq:reduced_eigenvalue_problem}) is real and discrete and the eigenvalues can be ordered with multiplicity: 
\begin{equation}
	E_1(p) \leq E_2(p) \leq ... \leq E_n(p) \leq ...
\end{equation}
For fixed $p$, the associated normalized eigenfunctions $\chi_n(z;p)$ are a basis of the space:
\begin{equation}
	L^2_{per} := \left\{ f \in L^2_{loc} : f(z + v) = f(z) \text{ for all } v \in \Lambda, z \in \field{R}^d \right\}
\end{equation} 
Varying $p$ over the Brillouin zone, the maps $p \mapsto E_n(p)$ are known as the spectral band functions. Their graphs are called the dispersion surfaces of $H$. The set of all dispersion surfaces as $p$ varies over $\mathcal{B}$ is called the band structure of $H$ (\ref{eq:periodic_operator}). Any function in $L^2(\field{R}^d)$ can be written as a superposition of Bloch waves:
\begin{equation}
	\left\{ \Phi_n(z;p) = e^{i p z} \chi_n(z;p) : n \in \field{N}, p \in \mathcal{B} \right\};
\end{equation}
see (\ref{eq:linear_comb_of_bloch_waves}). Moreover, the $L^2$-spectrum of the operator (\ref{eq:periodic_operator}) is the union of the real intervals swept out by the spectral band functions $E_n(p)$:
\begin{equation}
	\sigma(H)_{L^2(\field{R}^d)} = \cup_{n \in \field{R}} \left\{ E_n(p) : p \in \mathcal{B} \right\}.
\end{equation}
The map $p \mapsto E_n(p)$ extends to a map on $\field{R}^d$ which is periodic with respect to the reciprocal lattice $\Lambda^*$:
\begin{equation} \label{eq:extension_of_bloch_functions}
	\text{for any $b \in \Lambda^*$, } E_n(p + b) = E_n(p).
\end{equation}
If the eigenvalue $E_n(p)$ is simple, then: (up to a constant phase shift) $\chi_n(z;p+b) = e^{- i b \cdot z} \chi_n(z;p)$. A more detailed account of the Floquet-Bloch theory which we require, in particular results on the regularity of the maps $p \rightarrow E_n(p), \chi_n(z;p)$, can be found in Section \ref{sec:floquet_bloch_theory}. 

We will make the following assumptions throughout:
\begin{assumption}[Uniformly isolated band assumption]\label{isolated_band_assumption_nonseparable_case}
Let $E_n(p)$ be an eigenvalue band function of the periodic Schr\"{o}dinger operator (\ref{eq:periodic_operator}). Assume that $(q_0, p_0) \in \field{R}^d \times \field{R}^d$ are such that the flow generated by the classical Hamiltonian $\mathcal{H}_n(q,p) := E_n(p) + W(q)$:
\begin{equation} \label{eq:classical_system}
\begin{split}
	&\dot{q}(t) = \nabla_p E_n(p(t)) 	\\
	&\dot{p}(t) = - \nabla_q W(q(t))	\\
	&q(0), p(0) = q_0, p_0
\end{split}
\end{equation} 
has a unique smooth solution $(q(t),p(t)) \in \field{R}^d \times \field{R}^d, \forall t \geq 0$, and that there exists a constant $M > 0$ such that:
\begin{equation} \label{eq:band_remains_isolated}
	\inf_{m \neq n} |E_m(p(t)) - E_n(p(t))| \geq M \text{ for all } t \geq 0. 
\end{equation} 
That is, the $n$th spectral band is uniformly isolated along the classical trajectory $(q(t),p(t))$. 
\end{assumption}
\begin{assumption} \label{W_assumption} $\sum_{|\alpha| = 1,2,3,4} | \de_x^\alpha W(x) | \in L^\infty(\field{R}^d)$
\end{assumption}
\subsection{Dynamics of semiclassical wavepackets in a periodic background}
Our first result is an extension of Theorem 1.7 of Carles and Sparber \cite{carles_sparber}: 
\begin{theorem} \label{th:periodic_background_nonseparable_isolated_band_theorem}
Let Assumptions \ref{isolated_band_assumption_nonseparable_case} and \ref{W_assumption} hold. Let $a_0(y), b_0(y) \in \mathcal{S}(\field{R}^d)$. Let $S(t)$ denote the classical action along the path $(q(t),p(t))$:
\begin{equation} \label{eq:action_integral}
	S(t) = \inty{0}{t}{ p(t') \cdot \nabla_{p} E_n(p(t')) - E_n(p(t')) - W(q(t')) }{t'}.
\end{equation} 
Let $a(y,t)$ satisfy:
\begin{equation} 
\begin{split} \label{eq:envelope_equation} 
	&i \de_t a(y,t) = \mathscr{H}(t) a(y,t) 	\\
	&a(y,0) = a_0(y),
\end{split}
\end{equation}
where:
\begin{equation} \label{eq:def_of_H}
	\mathscr{H}(t) := - \frac{1}{2} \nabla_y \cdot D^2_p E_n(p(t)) \nabla_y + \frac{1}{2} y \cdot D^2_q W(q(t)) y. 
\end{equation}
And let $b(y,t)$ satisfy:
\begin{equation}
\begin{split} \label{eq:first_order_envelope_equation} 
	&i \de_t b(y,t) = \mathscr{H}(t) b(y,t) + \mathscr{I}(t) a(y,t) \\
	&b(y,0) = b_0(y),
\end{split}
\end{equation} 
where $\mathscr{H}(t)$ is as in (\ref{eq:def_of_H}) and: 
\begin{equation} \label{eq:def_of_I}
\begin{split}
	\mathscr{I}(t) := &- \frac{1}{6} \nabla_p \left[ \nabla_y \cdot D^2_p E_n(p(t)) \nabla_y \right] \cdot (- i \nabla_y) + \frac{1}{6} \nabla_q \left[ y \cdot D^2_q W(q(t)) y \right] \cdot y 	\\
	&+ \nabla_p \left[ \nabla_q W(q(t)) \cdot \mathcal{A}_n(p(t)) \right] \cdot (- i \nabla_y) + \nabla_q \left[ \nabla_q W(q(t)) \cdot \mathcal{A}_n(p(t)) \right] \cdot y. 	
\end{split}
\end{equation}
Here:
\begin{equation} \label{eq:berry_connection}
	\mathcal{A}_{n}(p(t)) := i \ip{\chi_n(\cdot;p(t)))}{\nabla_{p}\chi_n(\cdot;p(t))}_{L^2(\Omega)}
\end{equation}
is the $n$-th band Berry connection. Let $\phi_B(t)$ be the Berry phase associated with transport of $\chi_n$ along the path $p(t) \in \mathcal{B}$ given by:
\begin{equation} \label{eq:phase_equation}
	\phi_B(t) = \inty{0}{t}{ \dot{p}(t') \cdot \mathcal{A}_{n}(p(t'))  }{t'} = \inty{p_0}{p(t)}{ \mathcal{A}_n(p) \cdot}{p}.
\end{equation} 
Then, there exists a constant $\epsilon_0 > 0$ such that for all $0 < \epsilon \leq \epsilon_0$ the following holds. Let $\psi^\epsilon(x,t)$ be the unique solution of the initial value problem (\ref{eq:original_equation}) with `Bloch wavepacket' initial data:
\begin{equation} \label{eq:equation_and_initial_data}
\begin{split}
	&i \epsilon \de_t \psi^\epsilon = - \frac{1}{2} \epsilon^2 \Delta_x \psi^\epsilon + V\left(\frac{x}{\epsilon}\right) \psi^\epsilon + W(x) \psi^\epsilon \\
	&\psi^\epsilon(x,0) = \epsilon^{-d/4} e^{i p_{0} \cdot (x - q_0)/ \epsilon} \left\{ a_0 \left( \frac{x - q_0}{\epsilon^{1/2}} \right) \chi_n\left(\frac{x}{\epsilon};p_0\right) \right. \\
	&\left. + \epsilon^{1/2} \left[ (- i \nabla_{y}) a_0\left(\frac{x - q_0}{\epsilon^{1/2}}\right) \cdot \nabla_{p} \chi_n\left(\frac{x}{\epsilon};p_0\right) + b_0\left(\frac{x - q_0}{\epsilon^{1/2}}\right) \chi_n\left(\frac{x}{\epsilon};p_0\right) \right] \right\}.
\end{split}
\end{equation} 
Then, for all $t \geq 0$ the solution evolves as a modulated `Bloch wavepacket' plus a corrector $\eta^\epsilon(x,t)$:
\begin{equation} \label{eq:nonseparable_asymptotic_solution}
\begin{split}
	&\psi^\epsilon(x,t) = \epsilon^{-d/4} e^{i S(t) / \epsilon} e^{i p(t) \cdot (x - q(t)) / \epsilon} e^{i \phi_B(t)} \left\{ a\left(\frac{x - q(t)}{\epsilon^{1/2}},t\right) \chi_n\left(\frac{x}{\epsilon};p(t)\right) \right.  \\
	&\left. + \epsilon^{1/2} \left[ (- i \nabla_{y}) a\left(\frac{x - q(t)}{\epsilon^{1/2}},t\right) \cdot \nabla_{p} \chi_n\left(\frac{x}{\epsilon};p(t)\right) + b\left(\frac{x - q(t)}{\epsilon^{1/2}},t\right) \chi_n\left(\frac{x}{\epsilon};p(t)\right) \right] \right\} \\
	&+ \eta^\epsilon(x,t) 
\end{split}
\end{equation}
where the corrector $\eta^\epsilon$ satisfies the following estimate: 
\begin{equation} \label{eq:bound_on_eta}
	\| \eta^\epsilon(\cdot,t) \|_{L^2(\field{R}^d)} \leq C \epsilon e^{c t}.
\end{equation}
Here, $c > 0, C > 0$ are constants independent of $\epsilon, t$. It follows that:
\begin{equation}
	\sup_{t \in [0,\tilde{C} \ln 1/\epsilon]} \| \eta^\epsilon(\cdot,t) \|_{L^2(\field{R}^d)} = o(\epsilon^{1/2}) 
\end{equation}
where $\tilde{C}$ is any constant satisfying $\tilde{C} < \frac{1}{2 c}$. 
\end{theorem}
\begin{remark}
We include the pre-factor $\epsilon^{-d/4}$ throughout so that the $L^2_x(\field{R}^d)$ norm of $\psi^\epsilon(x,t)$ is of order $1$ as $\epsilon \downarrow 0$. 
\end{remark}
\begin{remark}
We have improved the error bound of Carles and Sparber \cite{carles_sparber} from $C \epsilon^{1/2} e^{c t}$ to $C \epsilon e^{c t}$ by including correction terms in the asymptotic expansion proportional to $\epsilon^{1/2}$. Note that we must also assume that the initial data is well-prepared up to terms proportional to $\epsilon^{1/2}$ (\ref{eq:equation_and_initial_data}). By keeping more terms in the expansion we may produce approximations where the corrector can be bounded by $C \epsilon^{N/2} e^{c t}$ for any positive integer $N$. The only changes in the proof are that we include corrections to the initial data proportional to $\epsilon^{N/2}$, and that Assumption \ref{W_assumption} is replaced by $\sum_{|\alpha| = 1,2,...,N + 2} |\de_x^\alpha W(x)| \in L^\infty(\field{R}^d)$. 
\end{remark}
\begin{remark}
Keeping terms proportional to $\epsilon^{1/2}$ in the expansion will allow us to calculate corrections to the dynamics of physical observables proportional to $\epsilon$; see Theorem \ref{prop:center_of_mass_proposition} and Section \ref{sec:center_of_mass_calculation}. 
\end{remark}
\begin{remark}
The time-scale of validity of the approximation (\ref{eq:nonseparable_asymptotic_solution}), $t \sim \ln 1/\epsilon$, is known as `Ehrenfest time', and is the general limit of validity of wavepacket, or coherent state, approximations; see \cite{robert}\cite{schubert_vallejos_toscano}. Note that including higher order terms (proportional to powers of $\epsilon^{1/2}$) in the approximation does not extend the time-scale of validity.
\end{remark}
There exists a family of time-dependent Gaussian explicit solutions of the envelope equation (\ref{eq:envelope_equation}). Consider (\ref{eq:envelope_equation}) with initial data:
 \begin{equation} \label{eq:gaussian_initial_data}
	a_0(y) = \frac{N}{ [\mathrm{det} A_0]^{1/2} } \exp \left( i \frac{1}{2} y \cdot B_0 A^{-1}_0 y \right)\ .
\end{equation} 
Here, $N \in \field{C}$ is an arbitrary non-zero constant, and $A_0, B_0$ are $d \times d$ complex matrices satisfying:
\begin{equation} \label{eq:conditions_on_A_B_again}
\begin{split}
	&A_0^T B_0 - B_0^T A_0 = 0	\\
	&\overline{A_0^T} B_0 - \overline{B^T_0} A_0 = 2 i I.
\end{split}
\end{equation}
\begin{remark}\label{matrices}
The conditions (\ref{eq:conditions_on_A_B_again}) imply: 
\begin{enumerate}
\item The matrices $B_0, A_0$ are invertible
\item The matrix $B_0 A_0^{-1}$ is complex symmetric: $( B_0 A_0^{-1} )^T = B_0 A_0^{-1}$
\item The imaginary part of the matrix $B_0 A_0^{-1}$ is symmetric, positive definite, and satisfies: $\text{Im }B_0 A_0^{-1} = (A_0 \overline{A_0^T})^{-1}$
\end{enumerate}
and are equivalent to the condition that the matrix:
\begin{equation}
	Y := \begin{pmatrix} \text{Re } A_0 & \text{Im } A_0 \\ \text{Re } B_0 & \text{Im } B_0 \end{pmatrix} \text{ is symplectic: } Y^T J Y = J \text{ where } J := \begin{pmatrix} 0 & - I \\ I & 0 \end{pmatrix}.
\end{equation}
The proofs of (1)-(3) are given in \cite{hagedorn_raising_and_lowering}\cite{hagedorn_coherent_states_1}\cite{faou_gradinaru_lubich}. 
\end{remark}
Note that it follows from assertion (3) of Remark \ref{matrices} that $a_0(y)$ (\ref{eq:gaussian_initial_data}) satisfies $| a_0(y) | \leq C e^{- c |y|^2}$ for constants $C > 0, c > 0$. We have then that: 
\begin{proposition} [Gaussian wavepackets] \label{prop:gaussian_wavepackets}
The initial value problem (\ref{eq:envelope_equation}) with initial data $a_0(y)$ given by (\ref{eq:gaussian_initial_data}) has the unique solution for all $t \ge 0$: 
\begin{equation} \label{eq:a(y,t)_periodic}
	a(y,t) = \frac{N}{ [ \mathrm{det} A(t) ]^{1/2} } \exp \left( i \frac{1}{2} y \cdot B(t) A^{-1}(t) y \right).
\end{equation}
Here, the complex matrices $A(t), B(t)$ satisfy: 
\begin{equation} \label{eq:equations_for_A_B_again}
\begin{split}
	&\dot{A}(t) = D^2_p E_n(p(t)) B(t), \quad \dot{B}(t) = - D^2_q W(q(t)) A(t), \\
	&A(0) = A_0, B(0) = B_0.
\end{split}
\end{equation}
Moreover, for all $t \ge 0$, the matrices $A(t), B(t)$ satisfy (\ref{eq:conditions_on_A_B_again}) with $A_0$ replaced by $A(t)$ and $B_0$ replaced by $B(t)$. Thus (see Remark \ref{matrices}), $|a(y,t)| \leq C(t) e^{- c(t) |y|^2}$ where $C(t) > 0, c(t) > 0$ for all $t \geq 0$. More generally, we may construct a basis of $L^2(\field{R}^d)$ of solutions of the envelope equation (\ref{eq:envelope_equation}), consisting of products of Gaussians with Hermite polynomials, known as the `semiclassical wavepacket' basis \cite{hagedorn_raising_and_lowering}\cite{hagedorn_coherent_states_1}\cite{faou_gradinaru_lubich}.
\end{proposition}
\begin{remark}
Our convention for the complex matrices $A, B$ follows that introduced in \cite{faou_gradinaru_lubich}, with $A, B$ standing for $Q, P$ in \cite{faou_gradinaru_lubich} respectively. Note that our convention is not to be confused with that introduced in \cite{hagedorn_coherent_states_1}; our choice of $B$ corresponds to $i B$ in \cite{hagedorn_coherent_states_1}. 
\end{remark}
\subsection{Dynamics of observables associated with the asymptotic solution}
We now deduce consequences for the physical observables associated with the solution $\psi^\epsilon(x,t)$ of (\ref{eq:equation_and_initial_data}), using the asymptotic form (\ref{eq:nonseparable_asymptotic_solution}). Denote the solution of (\ref{eq:equation_and_initial_data}) through order $\epsilon^{1/2}$ by:  
\begin{equation} \label{eq:tilde_psi}
\begin{split}
	&\tilde{\psi}^\epsilon(y,z,t) := \epsilon^{-d/4} e^{i S(t) / \epsilon} e^{ i p(t) \cdot y / \epsilon^{1/2} } e^{i \phi_B(t)} \left\{ a(y,t) \chi_n(z;p(t)) \vphantom{\epsilon^{1/2}} \right. \\
	&\left. + \epsilon^{1/2} \left[ \vphantom{\epsilon^{1/2}}  (- i \nabla_{y}) a(y,t) \cdot \nabla_{p} \chi_n(z;p(t)) + b(y,t) \chi_n(z;p(t)) \right] \right\}.
\end{split}
\end{equation}
Thus: 
\begin{equation}
	\psi^\epsilon(x,t) = \left. \tilde{\psi}^\epsilon(y,z,t) \right|_{z = \frac{x}{\epsilon},y = \frac{x - q(t)}{\epsilon^{1/2}}} + \eta^\epsilon(x,t). 
\end{equation}
where $\eta^\epsilon$ is the corrector which satisfies the bound (\ref{eq:bound_on_eta}). Define the physical observables: 
\begin{equation} \label{eq:bloch_observables}
\begin{split}
	&\mathcal{Q}^\epsilon(t) := \frac{1}{\mathcal{N}^\epsilon(t)} \inty{\field{R}^d}{}{ x \left| \tilde{\psi}^\epsilon(y,z,t) \right|^2_{z = \frac{x}{\epsilon},y = \frac{x - q(t)}{\epsilon^{1/2}}} }{x}  	\\
	&\mathcal{P}^\epsilon(t) := \frac{1}{\mathcal{N}^\epsilon(t)} \inty{\field{R}^d}{}{ \overline{ \tilde{\psi}^\epsilon(y,z,t) } \left( - i \epsilon^{1/2} \nabla_y \right) \left. \tilde{\psi}^\epsilon(y,z,t) \right|_{z = \frac{x}{\epsilon}, y = \frac{x - q(t)}{\epsilon^{1/2}}} }{x}.	\\
\end{split}
\end{equation}
where $\mathcal{N}^\epsilon(t)$ is the normalization factor:
\begin{equation}
	\mathcal{N}^\epsilon(t) = \inty{\field{R}^d}{}{ \left| \tilde{\psi}^\epsilon(y,z,t) \right|^2_{z = \frac{x}{\epsilon},y = \frac{x - q(t)}{\epsilon^{1/2}}} }{x}.
\end{equation}
We will refer to $\mathcal{Q}^\epsilon(t), \mathcal{P}^\epsilon(t)$ as the center of mass and average quasi-momentum of the wavepacket. We will see (Theorem \ref{prop:center_of_mass_proposition}): $\mathcal{Q}^\epsilon(t) = q(t) + o(1), \mathcal{P}^\epsilon(t) = p(t) + o(1)$ up to `Ehrenfest time' $t \sim \ln 1/\epsilon$.
\begin{remark} \label{rem:remark_on_momentum_observable}
In the uniform background case $V = 0$, solutions of (\ref{eq:reduced_eigenvalue_problem}) are independent of $z$: $\chi_n(z;p) = 1$ for all $p \in \mathcal{B}$. The asymptotic solution (\ref{eq:tilde_psi}) obtained in this case is therefore independent of $z$, and our definition of $\mathcal{P}^\epsilon(t)$ reduces to the usually defined momentum observable:
\begin{equation}
	\inty{\field{R}^d}{}{ \overline{\psi^\epsilon(x,t)} ( - i \epsilon \nabla_x ) \psi^\epsilon(x,t) }{x}.
\end{equation}
In the periodic background case $V \neq 0$, $\mathcal{P}^\epsilon$ (\ref{eq:bloch_observables}) corresponds to the quasi-momentum and may be measured in experiments \cite{damascelli}.
\end{remark}
\noindent Let $a^\epsilon(y,t)$ satisfy the equation of a quantum harmonic oscillator, with parametric forcing defined by $(\mathcal{Q}^\epsilon(t),\mathcal{P}^\epsilon(t))$:
\begin{equation} \label{eq:a_epsilon_equation}
\begin{split}
	&i \de_t a^\epsilon = \mathscr{H}^\epsilon(t) a^\epsilon; \quad \mathscr{H}^\epsilon(t) := - \frac{1}{2} \nabla_y \cdot D^2_{\mathcal{P}^\epsilon} E_n(\mathcal{P}^\epsilon(t)) \nabla_y + \frac{1}{2} y \cdot D^2_{\mathcal{Q}^\epsilon} W(\mathcal{Q}^\epsilon(t)) y	\\
	&a^\epsilon(y,0) = a_0(y).
\end{split}
\end{equation}
Note that we have replaced dependence on $(q(t),p(t))$ in equation (\ref{eq:envelope_equation}) with dependence on $\mathcal{Q}^\epsilon(t), \mathcal{P}^\epsilon(t)$. 

For simplicity of presentation of the following theorem we assume that:
\begin{equation} \label{eq:well_prepared_envelope}
\begin{split}
	&\ipy{a_0(y)}{y a_0(y)} = \ipy{a_0(y)}{(- i \nabla_y) a_0(y)} = 0 \\
	&\|a_0(y)\|_{L^2_y(\field{R}^d)} = 1.	
\end{split}
\end{equation}
The result holds for general $a_0(y) \in \mathcal{S}(\field{R}^d)$; see Section \ref{sec:center_of_mass_calculation}. 
\begin{theorem} \label{prop:center_of_mass_proposition}
Let $\tilde{\psi}^\epsilon(y,z,t)$ denote the asymptotic solution (\ref{eq:tilde_psi}) including corrections proportional to $\epsilon^{1/2}$. Let $\mathcal{Q}^\epsilon(t), \mathcal{P}^\epsilon(t)$ denote the observables (\ref{eq:bloch_observables}). Then, there exists an $\epsilon_0 > 0$ such that for all $0 < \epsilon \leq \epsilon_0$, and for all $t \in [0,\tilde{C}' \ln 1/\epsilon]$ where $\tilde{C}' > 0$ is a constant independent of $t, \epsilon$:
\begin{enumerate}
\item $\mathcal{Q}^\epsilon(t), \mathcal{P}^\epsilon(t)$ satisfy: 
\begin{equation} \label{eq:observables_expanded}
\begin{split}
	&\mathcal{Q}^\epsilon(t) = {q}(t) + \epsilon \left[ \ipy{b(y,t)}{y a(y,t)} + \ipy{a(y,t)}{y b(y,t)} \right] + \epsilon \mathcal{A}_{n}({p}(t)) + o(\epsilon) 	\\
	&\mathcal{P}^\epsilon(t) = {p}(t) + \epsilon \left[ \ipy{b(y,t)}{(- i \nabla_{y}) a(y,t)} + \ipy{a(y,t)}{(- i \nabla_{y}) b(y,t)} \right] + o(\epsilon) 
\end{split}
\end{equation} 
where $\mathcal{A}_{n}(p)$ is the $n$th band Berry connection (\ref{eq:berry_connection}), and $a(y,t), b(y,t)$ satisfy (\ref{eq:envelope_equation}) and (\ref{eq:first_order_envelope_equation}) respectively. 
\item Let $a^\epsilon(y,t)$ satisfy (\ref{eq:a_epsilon_equation}). Then: 
\begin{equation} \label{eq:eom_for_center_of_mass}
\begin{split}
	&\dot{\mathcal{Q}}^\epsilon_\alpha(t) = \de_{\mathcal{P}_\alpha^\epsilon} E_n(\mathcal{P}^\epsilon(t)) - \epsilon \dot{\mathcal{P}}^\epsilon_\beta(t) \mathcal{F}_{n, \alpha \beta}(\mathcal{P}^\epsilon(t)) 	\\ 
	&+ \epsilon \frac{1}{2} \de_{\mathcal{P}_\alpha^\epsilon} \ip{\nabla_y a^\epsilon(y,t)}{ \cdot D^2_{\mathcal{P}^\epsilon} E_n({\mathcal{P}}^\epsilon(t)) \nabla_y a^\epsilon(y,t)}_{L^2_y(\field{R}^d)} + o(\epsilon)  \\
	&\dot{\mathcal{P}}_\alpha^\epsilon(t) = - \de_{\mathcal{Q}_\alpha^\epsilon} W(\mathcal{Q}^\epsilon(t)) 	\\
	&- \epsilon \frac{1}{2} \de_{\mathcal{Q}_\alpha^\epsilon} \ipy{y a^\epsilon(y,t)}{ \cdot D^2_{\mathcal{Q}^\epsilon} W(\mathcal{Q}^\epsilon(t)) y a^\epsilon(y,t) } + o(\epsilon).
\end{split}
\end{equation}
Here, $\mathcal{F}_{n,\alpha \beta}(\mathcal{P}^\epsilon(t))$ denotes the Berry curvature of the $n$th band: 
\begin{equation} \label{eq:berry_curvature}
	\mathcal{F}_{n,\alpha \beta}(\mathcal{P}^\epsilon) := \de_{\mathcal{P}_\alpha^\epsilon} \mathcal{A}_{n,\beta}(\mathcal{P}^\epsilon) - \de_{\mathcal{P}^\epsilon_\beta} \mathcal{A}_{n,\alpha} (\mathcal{P}^\epsilon)
\end{equation}
where $\mathcal{A}_n(\mathcal{P}^\epsilon)$ is the $n$th band Berry connection (\ref{eq:berry_connection}). When $d = 3$ the anomalous velocity $- \dot{\mathcal{P}}^\epsilon_\beta(t) \mathcal{F}_{n, \alpha \beta}(\mathcal{P}^\epsilon(t))$ may be re-written using the cross product as in (\ref{system}); see Remark \ref{rem:anomalous_velocity_cross_products}. 
\item After dropping the terms of $o(\epsilon)$ in (\ref{eq:eom_for_center_of_mass}), equations (\ref{eq:eom_for_center_of_mass}), (\ref{eq:a_epsilon_equation}) form a closed, coupled `particle-field' system for $\mathcal{Q}^\epsilon(t), \mathcal{P}^\epsilon(t), a^\epsilon(y,t)$. 
\item Let 
\begin{equation} \label{eq:change_to_canonicals}
\begin{split}
	&\mathscr{Q}^\epsilon(t) := \mathcal{Q}^\epsilon(t) - \epsilon \mathcal{A}_n(\mathcal{P}^\epsilon(t))	\\
	&\mathscr{P}^\epsilon(t) := \mathcal{P}^\epsilon(t).
\end{split}
\end{equation}
Let $\mathfrak{a}^\epsilon(y,t)$ denote the solution of (\ref{eq:a_epsilon_equation}) with co-efficients evaluated at $\mathscr{Q}^\epsilon(t), \mathscr{P}^\epsilon(t)$ rather than $\mathcal{Q}^\epsilon(t), \mathcal{P}^\epsilon(t)$. 

Then, after dropping terms of $o(\epsilon)$, $\mathscr{Q}^\epsilon(t), \mathscr{P}^\epsilon(t), \mathfrak{a}^\epsilon(y,t)$ satisfy a closed, coupled `particle-field' system which is expressible as an $\epsilon$-dependent Hamiltonian system:
\begin{equation} \label{eq:periodic_hamiltonian_system}
\begin{split}
	&\dot{\mathscr{Q}}^\epsilon = \nabla_{\mathscr{P}^\epsilon} \mathcal{H}^\epsilon, \dot{\mathscr{P}}^\epsilon = - \nabla_{\mathscr{Q}^\epsilon} \mathcal{H}^\epsilon	\\
	&i \de_t \mathfrak{a}^\epsilon = \frac{ \delta \mathcal{H}^\epsilon }{ \delta \overline{\mathfrak{a}^\epsilon}  }
\end{split}
\end{equation}
with Hamiltonian:
\begin{equation} \label{eq:periodic_hamiltonian}
\begin{split}
	&\mathcal{H}^\epsilon(\mathscr{P}^\epsilon,\mathscr{Q}^\epsilon,\overline{\mathfrak{a}^\epsilon},\mathfrak{a}^\epsilon) := E_n(\mathscr{P}^\epsilon) + W(\mathscr{Q}^\epsilon) + \epsilon \nabla_{\mathscr{Q}^\epsilon} W(\mathscr{Q}^\epsilon) \cdot \mathcal{A}_n(\mathscr{P}^\epsilon)	\\
	&+ \epsilon \frac{1}{2} \ipy{\nabla_{y} \mathfrak{a}^\epsilon}{ \cdot D^2_{\mathscr{P}^\epsilon} E_n(\mathscr{P}^\epsilon) \nabla_{y} \mathfrak{a}^\epsilon} + \epsilon \frac{1}{2} \ipy{y \mathfrak{a}^\epsilon}{ \cdot D^2_{\mathscr{Q}^\epsilon} W({\mathscr{Q}}^\epsilon) y \mathfrak{a}^\epsilon}   
\end{split}
\end{equation}
\end{enumerate}
\end{theorem}
\begin{remark} \label{rem:anomalous_velocity_cross_products}
In three spatial dimensions ($d = 3$) the anomalous velocity may be re-written using the cross product:
\begin{align}
	& &- \dot{\mathcal{P}}^\epsilon_{\beta}(t) \mathcal{F}_{n,\alpha\beta}(\mathcal{P}^\epsilon(t)) &= - \dot{\mathcal{P}}^\epsilon_\beta(t) \left( \de_{\mathcal{P}_\alpha^\epsilon} \mathcal{A}_{n,\beta}(\mathcal{P}^\epsilon(t)) - \de_{\mathcal{P}^\epsilon_\beta} \mathcal{A}_{n,\alpha}(\mathcal{P}^\epsilon(t)) \right) \nonumber \\
	& & &= - ( \delta_{\alpha \gamma} \delta_{\beta \phi} - \delta_{\alpha \phi} \delta_{\beta \gamma} ) \dot{\mathcal{P}}^\epsilon_\beta(t) \de_{\mathcal{P}^\epsilon_\gamma} \mathcal{A}_{n,\phi}(\mathcal{P}^\epsilon(t)) \nonumber \\
	& & &= - \varepsilon_{\eta \alpha \beta} \varepsilon_{\eta \gamma \phi} \dot{\mathcal{P}}^\epsilon_\beta(t) \de_{\mathcal{P}^\epsilon_\gamma} \mathcal{A}_{n,\phi}(\mathcal{P}^\epsilon(t))  \\ 
	& & &= - \left( \dot{\mathcal{P}}^\epsilon(t) \times \nabla_{\mathcal{P}^\epsilon} \times \mathcal{A}_n(\mathcal{P}^\epsilon(t)) \right)_\alpha. \nonumber
\end{align}
Here, $\varepsilon$ and $\delta$ are the Levi-Civita and Kronecker delta symbols respectively and each equality follows from well-known properties of these symbols; see Section \ref{notation_and_conventions} (\ref{eq:kronecker_delta})-(\ref{eq:cross_product}). In this case the curl of the Berry connection: $\nabla_{\mathcal{P}^\epsilon} \times \mathcal{A}_n(\mathcal{P}^\epsilon)$ is often referred to as the Berry curvature, see for example \cite{e_lu_yang}. 
\end{remark}
\begin{remark} \label{rem:change_to_canonicals}
Equations (\ref{eq:eom_for_center_of_mass}) agree with those derived elsewhere (for example (3.5)-(3.6) of \cite{xiao_chang_niu}) up to the terms which depend on the wavepacket envelope $a^\epsilon$. The change of variables (\ref{eq:change_to_canonicals}) was introduced in \cite{e_lu_yang} to transform between the Hamiltonian system for the characteristics of a `corrected' eikonal ((4.9)-(4.10) in that work) and a gauge-invariant system ((4.11)-(4.12) in that work).  
\end{remark}
\begin{samepage}
\begin{corollary} \label{cor:gaussian_case}
{\ }
\begin{enumerate}
\item Choose initial data $a_0(y)$ of the form \eqref{eq:gaussian_initial_data}-\eqref{eq:conditions_on_A_B_again} with $N = \pi^{-d/4}$ so that $\| a_0(y) \|_{L^2_y} = 1$. 
 Then, $a^\epsilon(y,t)$, the solution of the initial value problem \eqref{eq:a_epsilon_equation}, is given by:
\begin{equation} \label{eq:a_epsilon_gaussian}
	a^\epsilon(y,t) = \frac{N}{ [ \mathrm{det} A^\epsilon(t) ]^{1/2} } \exp \left( i \frac{1}{2} y \cdot B^\epsilon(t) {A^\epsilon}^{-1}(t) y \right)\ ,
\end{equation}
where $A^\epsilon(t), B^\epsilon(t)$ satisfy:
\begin{equation} \label{eq:equations_for_A_B}
\begin{split}
	&\dot{A}^\epsilon(t) = D^2_{\mathcal{P}^\epsilon} E_n(\mathcal{P}^\epsilon(t)) B^\epsilon(t), \quad \dot{B}^\epsilon(t) = - D^2_{\mathcal{Q}^\epsilon} W(\mathcal{Q}^\epsilon(t)) A^\epsilon(t), \\
	&A^\epsilon(0) = A_0, B^\epsilon(0) = B_0.
\end{split}
\end{equation}
\item After the change of variables (\ref{eq:change_to_canonicals}), the full coupled system governing $(\mathscr{Q}^\epsilon,\mathscr{P}^\epsilon,A^\epsilon(t), B^\epsilon(t))$ governed by (\ref{eq:eom_for_center_of_mass})\ (with $o(\epsilon)$ terms dropped) and (\ref{eq:equations_for_A_B}), is expressible as a Hamiltonian system: 
\begin{equation} \label{eq:gaussian_eom}
\begin{split}
	&\dot{\mathscr{Q}}^\epsilon = \nabla_{\mathscr{P}^\epsilon} \mathcal{H}^\epsilon, \dot{\mathscr{P}}^\epsilon = - \nabla_{\mathscr{Q}^\epsilon} \mathcal{H}^\epsilon	\\
	&\dot{A}^\epsilon(t) = 4 \pd{\mathcal{H}^\epsilon}{\overline{B^\epsilon}}, \dot{B}^\epsilon(t) = - 4 \pd{\mathcal{H}^\epsilon}{\overline{A^\epsilon}}
\end{split}
\end{equation}
with Hamiltonian:
\begin{equation} \label{eq:periodic_hamiltonian_gaussian}
\begin{split}
	&\mathcal{H}^\epsilon(\mathscr{P}^\epsilon,\mathscr{Q}^\epsilon,\overline{A^\epsilon},A^\epsilon,\overline{B^\epsilon},B^\epsilon) := E_n(\mathscr{P}^\epsilon) + W(\mathscr{Q}^\epsilon) + \epsilon \nabla_{\mathscr{Q}^\epsilon} W(\mathscr{Q}^\epsilon) \cdot \mathcal{A}_{n}(\mathscr{P}^\epsilon)	\\
	&+ \epsilon \frac{1}{4} \text{Tr} \left[ ( \overline{ B^{\epsilon} } )^T D^2_{\mathcal{P}^\epsilon} E_n(\mathscr{P}^\epsilon) B^\epsilon \right] + \epsilon \frac{1}{4} \text{Tr} \left[ ( \overline{ A^{\epsilon} } )^T D^2_{\mathscr{Q}^\epsilon} W(\mathscr{Q}^\epsilon) A^\epsilon \right].
\end{split}
\end{equation}
\end{enumerate}
\end{corollary}
\end{samepage}
\begin{remark}
In the special case where the periodic background potential $V = 0$ the Bloch band dispersion function $E_n(\mathscr{P}^\epsilon)$ reduces to the `free' dispersion relation $\frac{1}{2} (\mathscr{P}^\epsilon)^2$ and the Hamiltonian (\ref{eq:periodic_hamiltonian_gaussian}) takes on the simple form: 
\begin{equation} \label{eq:V_0_Hamiltonian}
\begin{split}
	&\mathcal{H}^\epsilon(\mathscr{P}^\epsilon,\mathscr{Q}^\epsilon,\overline{A^\epsilon},A^\epsilon,\overline{B^\epsilon},B^\epsilon) := \frac{1}{2} (\mathscr{P}^\epsilon)^2 + W(\mathscr{Q}^\epsilon) 	\\
	&+ \epsilon \frac{1}{4} \text{Tr} \left[ ( \overline{ A^{\epsilon} } )^T D^2_{\mathscr{Q}^\epsilon} W(\mathscr{Q}^\epsilon) A^\epsilon \right] + \epsilon \frac{1}{4} \text{Tr} \left[ ( \overline{ B^\epsilon} )^T B^\epsilon \right]. 
\end{split}
\end{equation}
The system (\ref{eq:gaussian_eom}), with Hamiltonian $\mathcal{H}^\epsilon$ given by (\ref{eq:V_0_Hamiltonian}) has been derived by other methods: see Proposition 4.4 and equations (32b)-(32c) of \cite{ohsawa}. It was shown furthermore in \cite{ohsawa_leok} that corrections to the dynamics of $\mathscr{Q}^\epsilon, \mathscr{P}^\epsilon$ proportional to $\epsilon$ due to `field-particle' coupling to the $A^\epsilon, B^\epsilon$ system can lead to qualitatively different dynamical behavior. In particular, the coupling may destabilize periodic orbits of the unperturbed ($\epsilon = 0$) system; see Section 9 of \cite{ohsawa_leok}. 
\end{remark}

\subsection{Discussion of results, relation to previous work} \label{sec:discussion_of_results}

The $\epsilon \downarrow 0$ limit of (\ref{eq:original_equation}) has been studied by other methods. For example, by space-adiabatic perturbation theory \cite{panati_spohn_teufel_1}\cite{panati_spohn_teufel_2}\cite{teufel_book}\cite{2013StiepanTeufel}, and by studying the propagation of Wigner functions associated to the solution of (\ref{eq:original_equation}) \cite{markowich_mauser_poupaud}\cite{bechouche_mauser_poupaud}\cite{carles_fermaniankammerer_mauser_stimming}. The Wigner function approach is notable in that it has been used to study the propagation of wavepacket solutions of (\ref{eq:original_equation}) through band crossings \cite{lasser_swart_teufel}\cite{fermaniankammerer_gerard_lasser}. It was shown in \cite{e_lu_yang} that the anomalous velocity due to Berry curvature can be derived by a multiscale WKB-like ansatz by studying the characteristic equations of a corrected eikonal equation. The Hamiltonian structure of equations (\ref{eq:eom_for_center_of_mass}) without field-particle coupling terms was studied in \cite{2006DuvalHorvathHorvathyMartinaStichel} 

The effective system (\ref{eq:eom_for_center_of_mass}), in particular the `particle-field' coupling that we derive, is original to this work. Such coupled `particle-field' models arise naturally in many settings where a coherent structure interacts with a linear or nonlinear wave-field; see, for example \cite{wayne_weinstein} and references therein.  

The results detailed in Section \ref{sec:statement_of_results} generalize to the case where the potential has the more general form $U\left(\frac{x}{\epsilon},x\right)$ where $U$ is periodic in its first argument: 
\begin{equation}
	U(z + v,x) = U(z,x) \text{ for all } z, x \in \field{R}^d, v \in \Lambda.
\end{equation}
If $U(z,x)$ is not expressible as the sum of a periodic potential $V(z)$ and a smooth potential $W(x)$ we will say that $U$ is `non-separable'. In this case we must work with an $x$-dependent Bloch eigenvalue problem:
\begin{equation} \label{eq:reduced_eigenvalue_problem_again}
\begin{split}
	&H(p,x) \chi_n(z;p,x) = E_n(p,x) \chi_n(z;p,x)	\\	
	&\chi_n(z + v;p,x) = \chi_n(z;p,x) \text{ for all } z, x \in \field{R}^d, v \in \Lambda	\\
	&H(p,x) := \frac{1}{2} ( p - i \nabla_z )^2 + U(z,x).
\end{split}
\end{equation}
For details, see \cite{thesis}. Related problems were considered in \cite{2002JohnsonBienstmanSkorobogatiyIbanescuLidorikisJoannopoulos}\cite{2016Schnitzer}. An interesting example of a potential of this type is that of a domain wall modulated honeycomb lattice potential, which was shown to support `topologically protected' edge states in \cite{fefferman_leethorp_weinstein}.
\bigskip

\noindent{\bf Acknowledgements:} The authors wish to thank George Hagedorn, Christof Sparber, and Tomoki Ohsawa for stimulating discussions. This research was supported in part by National Science Foundation Grants  DMS-1412560 (AW \& MIW) and DMS-1454939 (JL),
and Simons Foundation Math + X Investigator Award \#376319 (MIW).

\section{Summary of relevant Floquet-Bloch theory} \label{sec:floquet_bloch_theory}
In this section we recall the spectral theory of the operator:
\begin{equation} \label{eq:original_operator}
	H := - \frac{1}{2} \Delta_z + V(z).
\end{equation}
where $V$ is periodic with respect to the lattice $\Lambda$ \cite{kuchment}\cite{reed_simon_4}. For $p \in \field{R}^d$, define the spaces of $p-$pseudo-periodic $L^2$ functions as follows:
\begin{equation}
	L^2_p := \left\{ f \in L^2_{loc} : f(z+v) = e^{i p \cdot v} f(z) \text{ for all } z \in \field{R}^d, v \in \Lambda \right\}. 
\end{equation}
Let $\Lambda^*$ denote the lattice dual to $\Lambda$:
\begin{equation}
	\Lambda^* := \left\{ b \in \field{R}^d : \exists v \in \Lambda : v \cdot b = 2 \pi n, n \in \field{Z} \right\}
\end{equation}
since the $p-$pseudo-periodic boundary condition is invariant under $p \rightarrow p + b$ where $b \in \Lambda^*$, the dual lattice to $\Lambda$, it is natural to restrict to a fundamental cell, $\mathcal{B}$. 

We now consider the family of eigenvalue problems depending on the parameter $p \in \mathcal{B}$:
\begin{equation} \label{eq:bloch_eigenvalue_problem}
\begin{split}
	&H \Phi(z;p) = E(p) \Phi(z;p)	\\
	&\Phi(z + v;p) = e^{i p \cdot v} \Phi(z;p) \text{ for all } z \in \field{R}^d, v \in \Lambda	\\
\end{split}
\end{equation}
We can also define the space $L^2_{loc}$ functions which are periodic with respect to the lattice: 
\begin{equation}
	L^2_{per} := \left\{ f(z) \in L^2_{loc}(\field{R}^d) : \forall v \in \Lambda, f(z+v) = f(z) \right\}.
\end{equation}
Then solving the eigenvalue problem (\ref{eq:bloch_eigenvalue_problem}) is equivalent via $\Phi(z;p) = e^{i p \cdot z} \chi(z;p)$ to solving the family of eigenvalue problems: 
\begin{equation} \label{eq:eigenvalue_problem}
\begin{split}
	&H(p) \chi(z;p) = E(p) \chi(z;p), \\
	&\chi(z + v;p) = \chi(z;p) \text{ for all } z \in \field{R}^d, v \in \Lambda	\\
	&H(p) = \frac{1}{2} \left( p - i \nabla_{z} \right)^2 + V(z) 
\end{split}	
\end{equation}

For fixed $p$, the operator $H(p)$ with periodic boundary conditions is self-adjoint and has compact resolvent. So, for each $n \in \field{N}$, there exists an eigenpair $E_n(p), \chi_n(z;p)$. The eigenvalues are real and can be ordered with multiplicity: 
\begin{equation}
	E_1(p) \leq E_2(p) \leq ... \leq E_{n-1}(p) \leq E_n(p) \leq E_{n+1}(p) \leq ...
\end{equation}  
and the set of normalized eigenfunctions $\left\{ \chi_n(z;p) : n \in \field{N} \right\}$ is complete in $L^2_{per}$. The set of Floquet-Bloch waves $\left\{ \Phi_n(z;p) = e^{i p z} \chi_n(z;p) : n \in \field{N}, p \in \mathcal{B} \right\}$ are complete in $L^2(\field{R}^d)$: 
\begin{equation} \label{eq:linear_comb_of_bloch_waves}
	g \in L^2(\field{R}^d) \implies g(x) = \sum_{n \geq 1} \inty{\field{B}}{}{\tilde{g}_n(p) \Phi_n(x;p)}{p}, \text{ where } \tilde{g}_n(p) := \ip{\Phi_n(\cdot;p)}{g(\cdot)}_{L^2(\field{R}^d)}
\end{equation}
where the sum converges in $L^2$. The $L^2(\field{R}^d)$ spectrum of the operator (\ref{eq:original_operator}) is obtained by taking the union of the closed real intervals swept out as $p$ varies over the Brillouin zone $\mathcal{B}$:
\begin{equation}
	\sigma(H)_{L^2(\field{R}^d)} = \cup_{n \in \field{R}} \overline{ \left\{ E_n(p) : p \in \mathcal{B} \right\} }
\end{equation} 

Our results require sufficient regularity of the maps: 
\begin{align} \label{eigenvalue_band_maps}
	&E_n:\mathcal{B} \rightarrow \field{R}, &	&p \mapsto E_n(p)	\\
	&\chi_n:\mathcal{B} \rightarrow L^2_{per}, &	&p \mapsto \chi_n(z;p) 
\end{align}
\begin{definition} \label{def:isolated_band}
We will call an eigenvalue band $E_n(p)$ of the problem (\ref{eq:eigenvalue_problem}) \emph{isolated} at a point $p \in \mathcal{B}$ if: 
\begin{equation} \label{eq:isolated_band_at_a_point}
	\inf_{m \neq n} | E_m(p) - E_n(p) | > 0.
\end{equation} 
\end{definition}
We have in this case:
\begin{theorem}[Smoothness of isolated bands]\label{th:analyticity_of_isolated_bands}
Let $E_n(p), \chi_n(z;p)$ satisfy the eigenvalue problem (\ref{eq:eigenvalue_problem}). Let the band $E_n(p)$ be isolated at a point $p_0$ in the sense of Definition \ref{def:isolated_band}. Then the maps (\ref{eigenvalue_band_maps}) are smooth in a neighborhood of the point $p_0$. 
\end{theorem}
When bands are not isolated we have the following situation:
\begin{definition} \label{def:band_crossing} Let $E_m(p), E_n(p) : m, n \in \field{N}, m \neq n$ be eigenvalue bands of the eigenvalue problem (\ref{eq:eigenvalue_problem}). If $p^* \in \mathcal{B}$ is such that:
\begin{equation}
	E_m(p^*) = E_n(p^*),
\end{equation}
we will say that the bands $E_n(p)$ and $E_m(p)$ have a \emph{band crossing} at $p^*$.
\end{definition}
In a neighborhood of a crossing, the band functions $E_n(p), E_m(p)$ are  only Lipschitz continuous, and the eigenfunction maps $p \mapsto \chi_n(z;p), \chi_m(z;p)$ may be discontinuous \cite{kuchment}. This loss of regularity occurs at conical degeneracies, which appear, for example, in the band structure of honeycomb lattice potentials \cite{fefferman_weinstein_diracpoints} \cite{fefferman_weinstein}, and in the dispersion surfaces of plane waves for homogeneous anisotropic media \cite{berry_jeffrey}. An in depth study of conical crossings which appear in the study of the Born-Oppenheimer approximation of molecular dynamics was given by \cite{hagedorn}.  

It will be convenient to extend the maps $p \mapsto E_n(p), \chi_n(z;p)$ to maps on all of $\field{R}^d$. Let $p \in \mathcal{B}$, and let $b \in \Lambda^*$ denote a reciprocal lattice vector. Then we have that:
\begin{equation}
\begin{split}
	&H(p + b) \left( e^{- i b \cdot z} \chi_n(z;p) \right) = e^{- i b \cdot z} H(p) \chi_n(z;p) 	\\
	&\quad = e^{- i b \cdot z} E_n(p) \chi_n(z;p) = E_n(p) \left( e^{- i b \cdot z} \chi_n(z;p) \right)	\\
	&\text{for all $v \in \Lambda$, } e^{- i b \cdot (z + v)} \chi_n(z + v;p) 	\\
	&\quad = e^{- i b \cdot v} e^{- i b \cdot z} \chi_n(z;p) = e^{- i b \cdot z} \chi_n(z;p),
\end{split}
\end{equation}
so that if $\chi_n(z;p)$ satisfies (\ref{eq:eigenvalue_problem}) with eigenvalue $E_n(p)$, then $ e^{- i b \cdot z} \chi_n(z;p)$ satisfies (\ref{eq:eigenvalue_problem}) with $p$ replaced by $p + b$, with the same eigenvalue. It then follows that the map $p \mapsto E_n(p)$ extends to a periodic function with respect to the reciprocal lattice $\Lambda^*$. If the eigenvalue $E_n(p)$ is simple, then: (up to a constant phase shift) $\chi_n(z;p + b) = e^{- i b \cdot z} \chi_n(z;p)$. 

\section{Proof of Theorem \ref{th:periodic_background_nonseparable_isolated_band_theorem} by multiscale analysis} \label{sec:proof_of_first_theorem}
\subsection{Derivation of asymptotic solution (\ref{eq:nonseparable_asymptotic_solution}) via multiscale expansion} \label{isolated_band_theory_nonseparable_case}
Following Hagedorn \cite{hagedorn_coherent_states_1}, \cite{hagedorn}, and Carles and Sparber \cite{carles_sparber}, we seek a solution of (\ref{eq:original_equation}) of the form:
\begin{equation} \label{eq:basic_ansatz}
	\psi^\epsilon(x,t) = \left. \epsilon^{-d/4} e^{i S(t) / \epsilon} e^{ i p(t) \cdot y / \epsilon^{1/2} } f^\epsilon(y,z,t) \right|_{z = \frac{x}{\epsilon}, y = \frac{x - q(t)}{\epsilon^{1/2}}} + \eta^\epsilon(x,t). 
\end{equation}
Substituting (\ref{eq:basic_ansatz}) into (\ref{eq:original_equation}) gives an inhomogeneous time-dependent Schr\"{o}dinger equation for $\eta^\epsilon(x,t)$, with a source term $r^\epsilon(x,t)$ which depends on $S(t), q(t), p(t)$, and $f^\epsilon(y,z,t)$: 
\begin{equation} \label{eq:equation_for_eta}
\begin{split}
	&i \epsilon \de_t \eta^\epsilon(x,t) =\left[ - \frac{\epsilon^2}{2} \Delta_x + V\left(\frac{x}{\epsilon}\right) + W(x) \right] \eta^\epsilon(x,t) + r^\epsilon[S,q,p,f^\epsilon](x,t)	\\
	&\eta^\epsilon(x,0) = \eta^\epsilon_0[S,q,p,f^\epsilon](x) = \psi^\epsilon(x,0) - \left. \epsilon^{-d/4} e^{i S(0) / \epsilon} e^{ i p(0) \cdot y / \epsilon^{1/2} } f^\epsilon(z,y,0) \right|_{z = \frac{x}{\epsilon}, y = \frac{x - q(0)}{\epsilon^{1/2}}} 
\end{split}
\end{equation}
The idea behind the proof of Theorem \ref{th:periodic_background_nonseparable_isolated_band_theorem} is to choose the functions $S(t), q(t), p(t)$, and $f^\epsilon(y,z,t)$ so that:
\begin{equation}
\begin{split}
	&r^\epsilon[S,q,p,f^\epsilon](x,t) = O(\epsilon^2) 	\\
	&\eta_0^\epsilon[S,q,p,f^\epsilon](x) = O(\epsilon)
\end{split}
\end{equation}
We will derive $S(t), q(t), p(t), f^\epsilon(y,z,t)$ by a systematic formal analysis. This is the content of Sections \ref{sec:leading_order_terms}, \ref{sec:first_order_terms}, \ref{sec:second_and_third_order_terms}. Proving rigorous bounds on the residual will be the content of Section \ref{sec:estimation_of_residual}. The bound (\ref{eq:bound_on_eta}) on $\eta^\epsilon(x,t)$ will then follow from applying the standard a priori $L^2$ bound for solutions of the time-dependent inhomogeneous Schr\"{o}dinger equation. 

Before starting on the formal asymptotic analysis, we note some exact manipulations which will ease calculations below. The residual $r^\epsilon(x,t)$ has the explicit form:
\begin{equation}
\begin{split}
	&r^\epsilon(x,t) = \epsilon^{-d/4} e^{i S(t) / \epsilon} e^{i p(t) \cdot y / \epsilon^{1/2}} \left\{ \epsilon \left[ \frac{1}{2} (- i \nabla_{y})^2 - i \de_t \right] \right. \\
	&+ \epsilon^{1/2} \left[ \vphantom{\frac{1}{2}} (p(t) - i \nabla_z) \cdot ( - i \nabla_{y}) - \dot{q}(t) \cdot (- i \nabla_{y}) + \dot{p}(t) \cdot y \right] \\
	&\left. \left.+ \left[ \dot{S}(t) - \dot{q}(t) p(t) + \frac{1}{2}(p(t) - i \nabla_z)^2 + V(z) + W(q(t) + \epsilon^{1/2} y) \right] \right\} f^\epsilon(y,z,t) \right|_{z = \frac{x}{\epsilon}, y = \frac{x - q(t)}{\epsilon^{1/2}}}	\\
\end{split}
\end{equation}
Since $W$ is assumed smooth, we can replace $W(q(t) + \epsilon^{1/2}y)$ by its Taylor series expansion in $\epsilon^{1/2} y$:
\begin{equation} \label{eq:taylor_expansion_of_W}
\begin{split}
	&W(q(t) + \epsilon^{1/2}y) = W(q(t)) + \epsilon^{1/2} \nabla_{q} W(q(t)) \cdot y + \epsilon \frac{1}{2} \de_{q_\alpha} \de_{q_\beta} W(q(t)) y_\alpha y_\beta   \\	
	&+ \epsilon^{3/2} \frac{1}{6} \de_{q_\alpha}\de_{q_\beta} \de_{q_\gamma} W(q(t)) y_\alpha y_\beta y_\gamma + \epsilon^2 \inty{0}{1}{ \frac{(\tau - 1)^4}{4!} \de_{q_\alpha} \de_{q_\beta} \de_{q_\gamma} \de_{q_\delta} W(q(t) + \tau \epsilon^{1/2} y) }{\tau} y_\alpha y_\beta y_\gamma y_\delta.
\end{split}
\end{equation}
We expand $f^\epsilon(y,z,t)$ as a formal power series:
\begin{equation} \label{eq:expanded_f}
	f^\epsilon(y,z,t) = f^0(y,z,t) + \epsilon^{1/2} f^1(y,z,t) + ... 
\end{equation}
and assume that for all $j \in \{0,1,2,...\}$ the $f^j(y,z,t)$ are periodic with respect to the lattice in $z$ and have sufficient smoothness and decay in $y$:
\begin{equation} \label{eq:assumptions_on_f_j}
\begin{split}
	&\text{ for all } v \in \Lambda, f^j(y,z + v,t) = f^j(y,z,t)	\\
	&f^j(y,z,t) \in \Sigma_y^{R - j}(\field{R}^d).
\end{split}
\end{equation}
The $\Sigma^l$-spaces are defined in (\ref{eq:sigma_l_spaces}). $R > 0$ is a fixed positive integer which we will take as large as required. Recall the notation:
\begin{equation} \label{eq:notation_for_operator}
	H(p) := \frac{1}{2} (p - i \nabla_z)^2 + V(z).
\end{equation}
Substituting (\ref{eq:taylor_expansion_of_W}) and (\ref{eq:expanded_f}) then gives:
\begin{equation} \label{eq:everything_expanded}
\begin{split}
	&r^\epsilon(x,t) = \epsilon^{-d/4} e^{i S(t) / \epsilon} e^{i p(t) \cdot y / \epsilon^{1/2}} \left\{ \epsilon^{2} \left[ \inty{0}{1}{ \frac{(\tau - 1)^4}{4!} \de_{q_\alpha} \de_{q_\beta} \de_{q_\gamma} \de_{q_\delta} W(q(t) + \tau \epsilon^{1/2} y) }{\tau} y_\alpha y_\beta y_\gamma y_\delta \right] \right. \\
	&+ \epsilon^{3/2} \left[ \frac{1}{6} \de_{q_\alpha} \de_{q_\beta} \de_{q_\gamma} W(q(t)) y_\alpha y_\beta y_\gamma \right] + \epsilon \left[ \frac{1}{2} (- i \nabla_{y})^2 + \frac{1}{2} \de_{q_\alpha} \de_{q_\beta} W(q(t)) y_\alpha y_\beta - i \de_t \right] \\
	&+ \epsilon^{1/2} \left[ \vphantom{\frac{1}{2}} (p(t) - i \nabla_z) \cdot ( - i \nabla_{y}) + \nabla_{q} W(q(t)) \cdot y - \dot{q}(t) \cdot (- i \nabla_{y}) + \dot{p}(t) \cdot y \right] \\
	&\left. \left.+ \left[ \dot{S}(t) - \dot{q}(t) \cdot p(t) + H(p(t)) \right] \right\} \left\{ f^0(y,z,t) + \epsilon^{1/2} f^1(y,z,t) + ... \right\} \right|_{z = \frac{x}{\epsilon}, y = \frac{x - q(t)}{\epsilon^{1/2}}}	\\
\end{split}
\end{equation}
In order to prove Theorem \ref{th:periodic_background_nonseparable_isolated_band_theorem} it will be sufficient to choose the $f^j(y,z,t), j \in \{0,...,3\}$ so that terms of orders $\epsilon^{j/2}, j \in \{0,...,3\}$ vanish. With this choice of $f^j, j \in \{0,...,3\}$ we will then prove rigorously in Section \ref{sec:estimation_of_residual} that $r^\epsilon(x,t)$ can be bounded by $C \epsilon^2 e^{c t}$ for constants $c > 0, C > 0$ independent of $\epsilon, t$. There will then be no loss of accuracy in the approximation by taking $f^j(y,z,t) = 0, j \geq 4$.

\subsubsection{Analysis of leading order terms} \label{sec:leading_order_terms}
Recall that we assume each $f^j, j \in \{0,1,2,...\}$ to be periodic with respect to the lattice $\Lambda$ in $z$ (\ref{eq:assumptions_on_f_j}). Collecting terms of order $1$ in (\ref{eq:everything_expanded}) and setting equal to zero therefore gives the following self-adjoint elliptic eigenvalue problem in $z$:
\begin{equation} \label{eq:equation_for_f_0}
\begin{split}
	&H(p(t)) f^0(y,z,t) = \left[ - \dot{S}(t) + \dot{q}(t) \cdot p(t) \right] f^0(y,z,t)	\\
	&\text{for all } v \in \Lambda, f^0(y,z + v,t) = f^0(y,z,t) 	\\
	&f^0(y,z,t) \in \Sigma^R_y(\field{R}^d)
\end{split}
\end{equation}
Under Assumption \ref{isolated_band_assumption_nonseparable_case}, $E_n(p(t))$ is a simple eigenvalue with eigenfunction $\chi_n(z;p(t))$ for all $t \geq 0$. Projecting equation (\ref{eq:equation_for_f_0}) onto the subspace of:
\begin{equation}
	L^2_{per} := \left\{ f \in L^2_{loc}(\field{R}^d) : \forall v \in \Lambda, f(z+v) = f(z) \right\}.
\end{equation}
spanned by $\chi_n(z;p(t))$ implies:
\begin{equation} \label{eq:dot_S}
	\dot{S}(t) = \dot{q}(t) \cdot p(t) - E_n(p(t)) 	
\end{equation}
which, after matching with the initial data (\ref{eq:equation_and_initial_data}), implies (\ref{eq:action_integral}). Equation (\ref{eq:equation_for_f_0}) then becomes:
\begin{equation}
\begin{split}
	&\left[ H(p(t)) - E_n(p(t)) \right] f^0(y,z,t) = 0		\\
	&\text{for all } v \in \Lambda, f^0(y,z + v,t) = f^0(y,z,t)	\\
	&f^0(y,z,t) \in \Sigma^R_y(\field{R}^d)
\end{split}
\end{equation}
which has the general solution: 
\begin{equation} \label{eq:leading_order_solution}
	f^0(y,z,t) = a^0(y,t) \chi_n(z;p(t)).
\end{equation}
where $a^0(y,t)$ is an arbitrary function in $\Sigma^R_y(\field{R}^d)$, to be fixed at higher order in the expansion. 

\subsubsection{Analysis of order $\epsilon^{1/2}$ terms} \label{sec:first_order_terms}
Collecting terms of order $\epsilon^{1/2}$ in (\ref{eq:everything_expanded}), substituting the form of $\dot{S}(t)$ (\ref{eq:dot_S}), and setting equal to zero gives the following inhomogeneous self-adjoint elliptic equation in $z$ for $f^1(y,z,t)$:
\begin{equation} \label{eq:equation_for_f_1}
\begin{split}
	&\left[ H(p(t)) - E_n(p(t)) \right] f^1(y,z,t) = \xi^1(y,z,t) 	\\
	&\text{for all } z \in \Lambda, f^1(y,z + v,t) = f^1(y,z,t); \; f^1(y,z,t) \in \Sigma^{R-1}_y(\field{R}^d)	\\
	&\xi^1 := - \left[ (p(t) - i \nabla_z) \cdot ( - i \nabla_{y}) + \nabla_{q} W(q(t)) \cdot y - \dot{q}(t) \cdot (- i \nabla_{y}) + \dot{p}(t) \cdot y \right] f^0(y,z,t).
\end{split}
\end{equation}
Before solving (\ref{eq:equation_for_f_1}) we remark on our general strategy for solving equations of this type.  
\begin{remark} \label{rem:general_strategy}
Collecting terms of orders $\epsilon^{j/2}$ for each $j \in \{1,2,...\}$ and setting equal to zero, we obtain inhomogeneous self-adjoint elliptic equations of the form:
\begin{equation} \label{eq:equation_for_f_j}
\begin{split}
	&\left[ H(p(t)) - E_n(p(t)) \right] f^j(y,z,t) = \xi^j[f^0,f^1,...,f^{j-1}](y,z,t) 	\\
	&\text{for all } z \in \Lambda, f^j(y,z + v,t) = f^j(y,z,t); \; f^j(y,z,t) \in \Sigma^{R-j}_y(\field{R}^d)
\end{split}
\end{equation}
Our strategy for solving (\ref{eq:equation_for_f_j}) will be the same for each $j$. Under Assumption \ref{isolated_band_assumption_nonseparable_case}, the eigenvalue $E_n(p(t))$ is simple with eigenfunction $\chi_n(z;p(t))$ for all $t \geq 0$. By the Fredholm alternative, equation (\ref{eq:equation_for_f_j}) is solvable if and only if:
\begin{equation} \label{eq:first_solvability_condition}
	\text{for all } t \geq 0, \ipz{\chi_n(z;p(t))}{\xi^j(y,z,t)} = 0. 
\end{equation}
We will first use identities derived in Appendix \ref{ch:useful_identities} from the eigenvalue equation $[H(p) - E_n(p)] \chi_n(z;p) = 0$ to write $\xi^j(y,z,t)$ as a sum:
\begin{equation}
	\xi^j(y,z,t) = \tilde{\xi}^j(y,z,t) + \left[ H(p(t)) - E_n(p(t)) \right] u^j(y,z,t)
\end{equation}
Note that by self-adjointness of $H(p(t)) - E_n(p(t))$, condition (\ref{eq:first_solvability_condition}) is equivalent to the same condition with $\xi^j(y,z,t)$ replaced by $\tilde{\xi}^j(y,z,t)$:
\begin{equation} \label{eq:general_solvability_condition}
	\text{for all } t \geq 0, \ipz{\chi_n(z;p(t))}{\tilde{\xi}^j(y,z,t)} = 0. 
\end{equation} 
For $f \in L^2_{per}$, define: 
\begin{equation} \label{eq:def_of_P_perp}
	P^\perp_n(p) f(z) := f(z) - \ipz{\chi_n(z;p)}{f(z)} \chi_n(z;p)
\end{equation}
to be the projection onto the orthogonal complement of the subspace of $L^2_{per}$ spanned by $\chi_n(z;p(t))$. Then, assuming (\ref{eq:general_solvability_condition}) is satisfied, the general solution of (\ref{eq:equation_for_f_j}) is:
\begin{equation} \label{eq:general_solution_f_j}
	f^j(y,z,t) = a^j(y,t) \chi_n(z;p(t)) + u^j(y,z,t) + \left[ H(p(t)) - E_n(p(t)) \right]^{-1} P^\perp_n(p(t)) \tilde{\xi}^j(y,z,t).
\end{equation}
Note that we have again made use of Assumption \ref{isolated_band_assumption_nonseparable_case} to ensure that the operator $[H(p(t)) - E_n(p(t))]^{-1} P^\perp_n(p(t)) : L^2_{per} \rightarrow L^2_{per}$ is bounded for all $t \geq 0$. When $j = 1$, condition (\ref{eq:general_solvability_condition}) may be enforced by choosing $\dot{q}(t), \dot{p}(t)$ to satisfy (\ref{eq:classical_system}). For $j \geq 2$, enforcing the constraint (\ref{eq:general_solvability_condition}) leads to evolution equations for $a^{j-2}(y,t)$. 
\end{remark}
We will give the proof of the following Lemma at the end of this section:
\begin{lemma} \label{lem:xi_1}
$\xi^1(y,z,t)$, defined in (\ref{eq:equation_for_f_1}), satisfies:
\begin{equation}
	\xi^1(y,z,t) = \tilde{\xi}^1(y,z,t) + \left[ H(p(t)) - E_n(p(t)) \right] u^1(y,z,t) 
\end{equation}
where:
\begin{equation} \label{eq:tilde_xi_and_u}
\begin{split}
	&\tilde{\xi}^1(y,z,t) := 	\\
	&- \left[ \left( \nabla_{p} E_n(p(t)) - \dot{q}(t) \right) \cdot (- i \nabla_{y}) a^0(y,t) + \left( \nabla_{q} W(q(t)) + \dot{p}(t) \right) \cdot y a^0(y,t) \right] \chi_n(z;p(t)) 	\\
	&u^1(y,z,t) := (- i \nabla_{y}) a^0(y,t) \cdot \nabla_{p} \chi_n(z;p(t))  		
\end{split}
\end{equation}
\end{lemma}
The solvability condition of (\ref{eq:equation_for_f_1}), given by (\ref{eq:general_solvability_condition}) with $j = 1$ on $\tilde{\xi}^1(y,z,t)$ (\ref{eq:tilde_xi_and_u}) is then equivalent to: 
\begin{equation}
	\left( \nabla_{p} E_n(p(t)) - \dot{q}(t) \right) \cdot (- i \nabla_{y}) a^0(y,t) + \left( \nabla_{q} W(q(t)) + \dot{p}(t) \right) \cdot y a^0(y,t) = 0
\end{equation}
which we can satisfy by choosing $(q(t),p(t))$ to evolve as the Hamiltonian flow of the $n$th Bloch band Hamiltonian $\mathcal{H}_n(q,p) = E_n(p) + W(q)$: 
\begin{equation} \label{eq:dot_p_dot_q}
	\dot{q}(t) = \nabla_p E_n(p(t)), \dot{p}(t) = - \nabla_q W(q(t)).
\end{equation}
Taking $q(0), p(0) = q_0, p_0$ to match with the initial data (\ref{eq:equation_and_initial_data}) implies (\ref{eq:classical_system}). 

The general solution of (\ref{eq:equation_for_f_1}) is given by taking $j = 1$ in (\ref{eq:general_solution_f_j}), where $u^1, \tilde{\xi}^1$ are given by (\ref{eq:tilde_xi_and_u}). With the choice (\ref{eq:dot_p_dot_q}) for $\dot{q}(t), \dot{p}(t)$ we have that $\tilde{\xi}^1 = 0$ for all $t \geq 0$ so that the general solution reduces to:  
\begin{equation} \label{eq:first_order_solution}
	f^1(y,z,t) = a^1(y,t) \chi_n(z;p(t)) + (- i \nabla_{y}) a^0(y,t) \cdot \nabla_{p} \chi_n(z;p(t)) 
\end{equation}
where $a^1(y,t)$ is an arbitrary function in $\Sigma^{R-1}_y(\field{R}^d)$ to be fixed at higher order in the expansion. Note that since $a^0(y,t) \in \Sigma^R_y(\field{R}^d)$, this ensures that $f^1(y,z,t) \in \Sigma^{R-1}_y(\field{R}^d)$ as required.  
\begin{proof} [Proof of Lemma \ref{lem:xi_1}]
By Assumption \ref{isolated_band_assumption_nonseparable_case}, $E_n(p)$ is smooth in a neighborhood of $p(t)$. By adding and subtracting $\nabla_p E_n(p(t)) \cdot (- i \nabla_y) f^0(y,z,t)$, $\xi^1(y,z,t)$ is equal to:
\begin{equation} \label{eq:xi_again}
\begin{split}
	&\xi^1(y,z,t) = - \left[ \left( (p(t) - i \nabla_z) - \nabla_p E_n(p(t)) \right) \cdot (- i \nabla_y) \right] f^0(y,z,t) 	\\
	&- \left[ \left( \nabla_p E_n(p(t)) - \dot{q}(t) \right) \cdot (- i \nabla_y) \right] f^0(y,z,t) - \left[ \left( \nabla_q W(q(t)) + \dot{p}(t) \right) \cdot y \right] f^0(y,z,t)
\end{split}
\end{equation}
Substituting the explicit form of $f^0(y,z,t)$ (\ref{eq:leading_order_solution}) into (\ref{eq:xi_again}) we have: 
\begin{equation} \label{eq:right_hand_side_terms_again}
\begin{split}
	&\xi^1(y,z,t) = - (- i \nabla_{y}) a^0(y,t) \cdot \left[ (p(t) - i \nabla_z) - \nabla_{p} E_n(p(t)) \right] \chi_n(z;p(t))	\\
	&- (- i \nabla_{y}) a^0(y,t) \cdot \left[ \nabla_{p} E_n(p(t)) - \dot{q}(t) \right] \chi_n(z;p(t)) - y a^0(y,t) \cdot \left[ \nabla_{q} W(q(t)) + \dot{p}(t) \right] \chi_n(z;p(t)).
\end{split}
\end{equation}
(\ref{eq:tilde_xi_and_u}) then follows immediately from identity (\ref{eq:first_derivative_identities}). 
\end{proof}

\subsubsection{Analysis of order $\epsilon$ and $\epsilon^{3/2}$ terms (summary)} \label{sec:second_and_third_order_terms}
It is possible to continue the procedure outlined in Remark \ref{rem:general_strategy} to any order in $\epsilon^{1/2}$. In Appendices \ref{derivation_of_leading_order_envelope_equation} and \ref{derivation_of_first_order_envelope_equation} we show the details of how to continue the procedure in order to cancel terms in the expansion of orders $\epsilon$ and $\epsilon^{3/2}$. In particular, we derive the evolution equations of the amplitudes $a^0(y,t), a^1(y,t)$ and show that: 
\begin{equation}
	a^0(y,t) = a(y,t) e^{i \phi_B(t)}, a^1(y,t) = b(y,t) e^{i \phi_B(t)}
\end{equation}
where $a(y,t), b(y,t), \phi_B(t)$ satisfy equations (\ref{eq:envelope_equation}), (\ref{eq:first_order_envelope_equation}), and (\ref{eq:phase_equation}) respectively. 

\subsection{Proof of estimate (\ref{eq:bound_on_eta}) for the corrector $\eta$} \label{sec:estimation_of_residual}
Let:
\begin{equation}
	f_3^\epsilon(y,z,t) := f^0(y,z,t) + \epsilon^{1/2} f^1(y,z,t) + \epsilon f^2(y,z,t) + \epsilon^{3/2} f^3(y,z,t)
\end{equation}
Where the $f^0, f^1, f^2, f^3$ are given by (\ref{eq:leading_order_solution}), (\ref{eq:first_order_solution}), (\ref{eq:second_order_solution}), (\ref{eq:f_3}) respectively, and define: 
\begin{equation}
	\psi^\epsilon_3(x,t) := \left. \epsilon^{-d/4} e^{i S(t) / \epsilon} e^{i p(t) \cdot y / \epsilon^{1/2} } f^\epsilon_3(y,z,t) \right|_{z = \frac{x}{\epsilon}, y = \frac{x - q(t)}{\epsilon^{1/2}} }
\end{equation}
Let $\psi^\epsilon(x,t)$ denote the exact solution of the initial value problem (\ref{eq:equation_and_initial_data}). From the manipulations of the previous Section, we have that $\eta_3^\epsilon(x,t) := \psi^\epsilon(x,t) - \psi^\epsilon_3(x,t)$ satisfies:
\begin{equation} \label{eq:equation_for_eta_3}
\begin{split}
	&i \epsilon \de_t \eta_3^\epsilon(x,t) =\left[ - \frac{\epsilon^2}{2} \Delta_x + V\left(\frac{x}{\epsilon}\right) + W(x) \right] \eta_3^\epsilon(x,t) + r_3^\epsilon(x,t)	\\
	&\eta_3^\epsilon(x,0) = \eta_{3,0}^\epsilon(x) 
\end{split}
\end{equation} 
where $r_3^\epsilon(x,t)$ is given by:
\begin{equation}  \label{eq:expression_for_r}
\begin{split}
	&r_3^\epsilon(x,t) = \epsilon^{-d/4} e^{i S(t) / \epsilon} e^{i p(t) \cdot y / \epsilon^{1/2}} \left\{\vphantom{\frac{1}{2}} \epsilon^2 \inty{0}{1}{ \frac{(\tau - 1)^4}{4!} \de_{q_\alpha} \de_{q_\beta} \de_{q_\gamma} \de_{q_\delta} W(q(t) + \tau \epsilon^{1/2} y) }{\tau} y_\alpha y_\beta y_\gamma y_\delta  \right. \\
	&\cdot \left( f^0(y,z,t) + \epsilon^{1/2} f^1(y,z,t) + \epsilon f^2(y,z,t) + \epsilon^{3/2} f^3(y,z,t) \right) 	\\
	&+ \epsilon^{2} \left[ \frac{1}{6} \de_{q_\alpha} \de_{q_\beta} \de_{q_\gamma} W(q(t)) y_\alpha y_\beta y_\gamma \right] \left( f^1(y,z,t) + \epsilon^{1/2} f^2(y,z,t) + \epsilon f^3(y,z,t) \right) 	\\
	&+ \epsilon^{2} \left[ - i \de_t + \frac{1}{2} (- i \nabla_{y})^2 + \frac{1}{2} \de_{q_\alpha} \de_{q_\beta} W(q(t)) y_\alpha y_\beta \right] \left( f^2(y,z,t) + \epsilon^{1/2} f^3(y,z,t) \right)  	\\
	&\left. \left. + \epsilon^2 \left[ \vphantom{\frac{1}{2}} \left( \left( p(t) - i \nabla_z \right) - \nabla_{p} E_n(p(t))\right) \cdot (- i \nabla_{y}) \right] f^3(y,z,t) \right\} \right|_{z = \frac{x}{\epsilon},y = \frac{x - q(t)}{\epsilon^{1/2}}}
\end{split}
\end{equation}
And $\eta^\epsilon_{3,0}(x)$ is given by:
\begin{equation} \label{eq:expression_for_eta}
	\eta^\epsilon_{3,0}(x) = - \left. \epsilon^{-d/4} e^{i p_{0} \cdot y / \epsilon^{1/2}} \left\{ \epsilon \left[ f^2(z,y,0) + \epsilon^{1/2} f^3(z,y,0) \right] \right\} \right|_{z = \frac{x}{\epsilon}, y = \frac{x - q_0}{\epsilon^{1/2}}} 
\end{equation}
Since the $f^j(y,z,t), j \in \{0,...,3\}$ are periodic with respect to the lattice $\Lambda$, we will follow Carles and Sparber \cite{carles_sparber} and bound the above expressions in the uniform norm in $z$ and the $L^2$ norm in $y$:
\begin{equation} \label{eq:r_uniform_and_L_2_norms}
\begin{split}
	&\| r_3^\epsilon(x,t) \|_{L^2_x} \leq \\
	&= \left\| \epsilon^2 \inty{0}{1}{ \frac{(\tau - 1)^4}{4!} \de_{q_\alpha} \de_{q_\beta} \de_{q_\gamma} \de_{q_\delta} W(q(t) + \tau \epsilon^{1/2} y) }{\tau} y_\alpha y_\beta y_\gamma y_\delta  \right. \\
	&\cdot \left( f^0(y,z,t) + \epsilon^{1/2} f^1(y,z,t) + \epsilon f^2(y,z,t) + \epsilon^{3/2} f^3(y,z,t) \right) 	\\
	&+ \epsilon^{2} \left[ \frac{1}{6} \de_{q_\alpha} \de_{q_\beta} \de_{q_\gamma} W(q(t)) y_\alpha y_\beta y_\gamma \right] \left( f^1(y,z,t) + \epsilon^{1/2} f^2(y,z,t) + \epsilon f^3(y,z,t) \right) 	\\
	&+ \epsilon^{2} \left[ - i \de_t + \frac{1}{2} (- i \nabla_y)^2 + \frac{1}{2} \de_{q_\alpha} \de_{q_\beta} W(q(t)) y_\alpha y_\beta \right] \left( f^2(y,z,t) + \epsilon^{1/2} f^3(y,z,t) \right)  	\\
	&\left. + \epsilon^{2} \left[ \vphantom{\frac{1}{2}} \left( \left( p(t) - i \nabla_z \right) - \nabla_{p} E_n(p(t))\right) \cdot (- i \nabla_{y}) \right] f^3(y,z,t) \right\|_{L^\infty_z, L^2_y} 	\\
	&\| \eta_3^\epsilon(x,0) \|_{L^2_x} \leq \left\| \epsilon \left[ f^2(z,y,0) + \epsilon^{1/2} f^3(z,y,0) \right] \right\|_{L^\infty_z, L^2_y}	\\ 
\end{split}
\end{equation}
where we have used the fact that:
\begin{equation}
	\left\| \epsilon^{-d/4} f\left(\frac{x - q(t)}{\epsilon^{1/2}}\right) \right\|_{L^2_x} = \| f(y) \|_{L^2_y}. 
\end{equation}
We show how to bound the first term in (\ref{eq:r_uniform_and_L_2_norms}). Bounding the other terms is similar, although care must be taken in bounding terms in $L^\infty_z$, see Appendix \ref{app:uniform_bounds_on_z_dependence}. Let:
\begin{equation}
	\mathcal{I}(t) := \epsilon^2 \left\| \inty{0}{1}{ \frac{(\tau - 1)^4}{4!} \de_{q_\alpha} \de_{q_\beta} \de_{q_\gamma} \de_{q_\delta} W(q(t) + \tau \epsilon^{1/2} y) }{\tau} y_\alpha y_\beta y_\gamma y_\delta f^0(y,z,t) \right\|_{L^\infty_z, L^2_y}		\\
\end{equation} 
where $f^0(y,z,t) = a^0(y,t) \chi_n(z;p(t))$ (\ref{eq:leading_order_solution}). By Assumption \ref{W_assumption}, $\sum_{|\alpha| = 4} | \de_x^\alpha W(x) | \in L^\infty(\field{R}^d)$:
\begin{equation} \label{eq:term_to_be_bounded}
	\mathcal{I}(t) \leq \epsilon^2 \frac{1}{4!} \left\| \sum_{|\alpha| = 4} | \de_x^\alpha W(x) | \right\|_{L^\infty(\field{R}^d)} \sum_{|\alpha| = 4} \left\| y^\alpha a(y,t) \chi_n(z;p(t)) \right\|_{L^\infty_z, L^2_y}		\\
\end{equation} 
Recall Assumption \ref{isolated_band_assumption_nonseparable_case}. Define: 
\begin{equation} \label{eq:S_n}
	S_n := \{ p \in \field{R}^d : \inf_{m \neq n} | E_m(p) - E_n(p) | \geq M \},
\end{equation}
so that for all $t \in [0,\infty), \, p(t) \in S_n$. For each fixed $p \in S_n$, by elliptic regularity, $\chi_n(z;p)$ is smooth in $z$ so that $\| \chi_n(z;p) \|_{L^\infty_z} < \infty$. Using compactness of the Brillouin zone $\mathcal{B}$ and smoothness of $\chi_n(z;p)$ for $p \in S_n$ we have that:
\begin{equation} \label{eq:chi_on_brillouin_zone}
	\sup_{p \in \mathcal{B} \cap S_n} \| \chi_n(z;p) \|_{L^\infty_z} < \infty.
\end{equation}
Recall that for any reciprocal lattice vector $b \in \Lambda^*$, $\chi_n(z;p + b) = e^{- i b \cdot z} \chi_n(z;p)$. It then follows that:
\begin{equation} \label{eq:periodic_with_respect_lambda_star}
	\text{for all } b \in \Lambda^*,\quad \| \chi_n(z;p + b) \|_{L^\infty_z} = \| \chi_n(z;p) \|_{L^\infty_z}
\end{equation} 
It then follows from combining (\ref{eq:chi_on_brillouin_zone}) and (\ref{eq:periodic_with_respect_lambda_star}) that:
\begin{equation}
	\sup_{p \in S_n} \| \chi_n(z;p) \|_{L^\infty_z} < \infty.
\end{equation}
In Appendix \ref{app:uniform_bounds_on_z_dependence} we show how to bound all $z$-dependence in $r^\epsilon_3(x,t)$ (\ref{eq:expression_for_r}) uniformly in $p \in S_n$ in a similar way. 

We have therefore that (\ref{eq:term_to_be_bounded}):
\begin{equation} \label{eq:term_to_be_bounded_again}
	\mathcal{I}(t) \leq \epsilon^2 \frac{1}{4!} \left\| \sum_{|\alpha| = 4} | \de_x^\alpha W(x) | \right\|_{L^\infty(\field{R}^d)} \sup_{p \in S_n} \| \chi_n(z;p) \|_{L^\infty_z} \sum_{|\alpha| = 4} \left\| y^\alpha a(y,t) \right\|_{L^2_y} 
\end{equation}
We see that to complete the bound, we require a bound on the 4th moments of $a(y,t)$, which solves the Schr\"{o}dinger equation with time-dependent co-efficients:
\begin{equation} \label{eq:envelope_equation_again}
\begin{split}
	&i \de_t a(y,t) = \frac{1}{2} \de_{p_\alpha} \de_{p_\beta} E_n(p(t)) (- i \de_{y_\alpha}) (- i \de_{y_\beta}) a(y,t)  + \frac{1}{2} \de_{q_\alpha} \de_{q_\beta} W(q(t)) y_\alpha y_\beta a(y,t)	\\
	&a(y,0) = a_0(y)
\end{split}
\end{equation}
Following Carles and Sparber \cite{carles_sparber} we first define, for any $l \in \field{N}$, the spaces:
\begin{equation} \label{eq:sigma_spaces}
	\Sigma^l(\field{R}^d) := \left\{ f \in L^2(\field{R}^d) : \| f \|_{\Sigma^l} := \sum_{|\alpha| + |\beta| \leq l} \| y^\alpha (- i \de_y)^\beta f(y) \|_{L^2_y} < \infty, \right\}  
\end{equation}
We then require the following Lemma due to Kitada \cite{kitada}: 
\begin{lemma} [Existence of unitary solution operator for the envelope equation] \label{lem:kitada}
Let $u_0 \in L^2(\field{R}^d)$, and $\eta_{\alpha \beta}(t), \zeta_{\alpha \beta}(t)$ be real-valued, symmetric, continuous, and uniformly bounded in $t$. Then the equation: 
\begin{equation} \label{eq:time_dependent_harmonic_oscillator}
\begin{split}
	&i \de_t u = \frac{1}{2} \eta_{\alpha \beta}(t) (- i \de_{y_\alpha}) (- i \de_{y_\beta}) u + \frac{1}{2} \zeta_{\alpha \beta}(t) y_\alpha y_\beta u	\\
	&u(y,0) = u_0(y)
\end{split}
\end{equation}
has a unique solution $u \in C([0,\infty);L^2(\field{R}^d))$. It satisfies:
\begin{equation}
	\| u(\cdot,t) \|_{L^2(\field{R}^d)} = \| u_0(\cdot) \|_{L^2(\field{R}^d)}
\end{equation}
Moreover, if $u_0 \in \Sigma^l(\field{R}^d)$, then $u \in C([0,\infty);\Sigma^l(\field{R}^d))$.  
\end{lemma}
We seek quantitative bounds on $\| u(\cdot,t) \|_{\Sigma^l(\field{R}^d)}$ for $l \geq 1$. For simplicity, we consider in detail the case $l = 1$. Recall (\ref{eq:def_of_H}):
\begin{equation}
	\mathscr{H}(t) := \frac{1}{2} \eta_{\alpha \beta}(t) (- i \de_{y_\alpha})(- i \de_{y_\beta}) + \frac{1}{2} \zeta_{\alpha \beta}(t) y_\alpha y_\beta
\end{equation}
and let $u_0(y) \in \mathcal{S}(\field{R}^d)$ so that $\forall l \geq 0$, the solution of (\ref{eq:time_dependent_harmonic_oscillator}), $u(y,t) \in C([0,\infty);\Sigma^l(\field{R}^d))$. Then $(- i \de_{y_\alpha}) u(y,t) \in \mathcal{S}(\field{R})$ solves:
\begin{equation}
\begin{split}
	&i \de_t (- i \de_{y_\alpha}) u = \mathscr{H}(t) (- i \de_{y_\alpha}) u + [(- i \de_{y_\alpha}),\mathscr{H}(t)] u 	\\
	&(- i \de_{y_\alpha}) u(y,0) = (- i \de_{y_\alpha}) u_0(y)
\end{split}
\end{equation}
We can solve this equation using Duhamel's formula and the solution operator of equation (\ref{eq:time_dependent_harmonic_oscillator}). It follows that:
\begin{equation}
	\| (- i \de_{y_\alpha}) u(y,t) \|_{L^2_y} \leq \| (- i \de_{y_\alpha}) u(y,0) \|_{L^2_y} + \inty{0}{t}{ \| [(- i \de_{y_\alpha}),\mathscr{H}(s)] u(y,s) \|_{L^2_y} }{s} 
\end{equation} 
Since $\zeta_{\alpha \beta}(t)$ is symmetric, the commutator is given explicitly by:
\begin{equation}
	[(- i \de_{y_\alpha}),\mathscr{H}(s)] = (- i) \zeta_{\alpha \beta}(t) y_\beta 
\end{equation}
So that:
\begin{equation} \label{eq:first_inequality}
\begin{split}
	\| (- i\de_{y_\alpha}) u(y,t) \|_{L^2_y} \leq \| (- i \de_{y_\alpha}) u_0(y) \|_{L^2_y} + \inty{0}{t}{ | \zeta_{\alpha \beta}(t) | \| y_\beta u(y,s) \|_{L^2_y} }{s}  
\end{split}
\end{equation} 
By an identical reasoning we can derive a similar bound on $y_\alpha u(y,t)$: 
\begin{equation} \label{eq:second_inequality}
	\| y_\alpha u(y,t) \|_{L^2_y} \leq \| y_\alpha u_0(y) \|_{L^2_y} + \inty{0}{t}{ |\eta_{\alpha \beta}(s)| \| (- i \de_{y_\beta}) u(y,s) \|_{L^2_y} }{s} 
\end{equation}
Adding inequalities (\ref{eq:first_inequality}) and (\ref{eq:second_inequality}) gives:
\begin{equation}
	\| u(\cdot,t) \|_{\Sigma^1} \leq \| u_0(\cdot) \|_{\Sigma^1} + 2 \inty{0}{t}{ \max_{ \alpha, \beta \in \{1,...,d\} } \left\{ |\eta_{\alpha \beta}(s)|, |\zeta_{\alpha \beta}(s)| \right\} \| u(\cdot,s) \|_{\Sigma^1} }{s}
\end{equation}
Using the following version of Gronwall's inequality:
\begin{lemma} [Gronwall's inequality]
Let $v(t)$ satisfy the inequality:
\begin{equation}
	v(t) \leq a(t) + \inty{0}{t}{ b(s) v(s) }{s}
\end{equation}
where $b(t)$ is non-negative and $a(t)$ is non-decreasing. Then:
\begin{equation}
	v(t) \leq a(t) \exp\left( \inty{0}{t}{ b(s) }{s} \right)	
\end{equation} 
\end{lemma}
We have that:
\begin{equation}
	\| u(\cdot,t) \|_{\Sigma^1} \leq \| u_0(\cdot) \|_{\Sigma^1} e^{ 2 \inty{0}{t}{ \max_{\alpha,\beta \in \{1,...,d\}} \left\{ |\eta_{\alpha \beta}(s)|, |\zeta_{\alpha \beta}(s)| \right\} }{s} }
\end{equation}
More generally, we have for any $l \geq 0$ that there exists a constant $C_l > 0$ such that:
\begin{equation} \label{eq:sigma_l_bound}
	\| u(\cdot,t) \|_{\Sigma^l} \leq \| u_0(\cdot) \|_{\Sigma^l} e^{ C_l \inty{0}{t}{ \max_{\alpha,\beta \in \{1,...,d\}} \left\{ |\eta_{\alpha \beta}(s)|, |\zeta_{\alpha \beta}(s)| \right\} }{s} } 
\end{equation}
We have proved the following: 
\begin{lemma} [Bound on solutions of (\ref{eq:time_dependent_harmonic_oscillator}) in the spaces $\Sigma^l(\field{R}^d)$] \label{bounds_in_sigma_l}
Let the time-dependent co-efficients $\eta_{\alpha \beta}(t), \zeta_{\alpha \beta}(t)$ be real-valued, symmetric, continuous, and uniformly bounded in $t$. Let $u_0(y) \in \Sigma^l(\field{R}^d)$. Then, by Lemma \ref{lem:kitada}, there exists a unique solution $u(y,t) \in C([0,\infty);\Sigma^l(\field{R}^d))$. For each integer $l \geq 0$, there exists a constant $C_l > 0$ such that this solution satisfies: 
\begin{equation}
	\| u(\cdot,t) \|_{\Sigma^l(\field{R}^d)} \leq \| u_0(\cdot) \|_{\Sigma^l} e^{C_l \inty{0}{t}{ \max_{\alpha,\beta \in \{1,...,d\}} \left\{ |\eta_{\alpha \beta}(s)|, |\zeta_{\alpha \beta}(s)| \right\} }{s} }
\end{equation}
\end{lemma}
Since the map $p \mapsto E_n(p)$ is $\mathcal{B}$-periodic and smooth for all $p \in S_n$, we have that under Assumption \ref{isolated_band_assumption_nonseparable_case}, $\sup_{t \in [0,\infty)} \max_{\alpha, \beta \in \{1,...,d\}} | \de_{p_\alpha} \de_{p_\beta} E_n(p(t)) | < \infty$. Under Assumption \ref{W_assumption} we have that $\sup_{t \in [0,\infty)} \max_{\alpha, \beta \in \{1,...,d\}} | \de_{q_\alpha} \de_{q_\beta} W(q(t)) | < \infty$. Since $\de_{p_\alpha} \de_{p_\beta} E_n(p(t)), \de_{q_\alpha} \de_{q_\beta} W(q(t))$ are clearly real-valued, symmetric, and continuous in $t$ we have that Lemma \ref{bounds_in_sigma_l} applies to solutions of (\ref{eq:envelope_equation_again}). Since $a_0(y) \in \mathcal{S}(\field{R}^d)$ by assumption we have that for any integer $l \geq 0$:
\begin{equation}
\begin{split}
	&\| a(y,t) \|_{\Sigma^l(\field{R}^d)} \leq \| a_0(y) \|_{\Sigma^l(\field{R}^d)} 	\\
	&\cdot \exp\left(C_l \max_{\alpha,\beta \in \{1,...,d\}} \sup_{s \in [0,\infty)} \left\{ | \de_{p_\alpha} \de_{p_\beta} E_n(p(s)) |, \vphantom{\frac{1}{2}} | \de_{q_\alpha} \de_{q_\beta} W(q(s)) | \right\} t \right).
\end{split}
\end{equation}
\begin{remark}
Terms which depend on $b(y,t)$ rather than $a(y,t)$ may be dealt with similarly, by an application of Duhamel's formula and a Gronwall inequality.  
\end{remark}
We have therefore that:
\begin{equation}
\begin{split}
	&\epsilon^2 \left\| \inty{0}{1}{ \frac{(\tau - 1)^4}{4!} \de_{q_\alpha} \de_{q_\beta} \de_{q_\gamma} \de_{q_\delta} W(q(t) + \tau \epsilon^{1/2} y) }{\tau} y_\alpha y_\beta y_\gamma y_\delta a(y,t) \chi_n(z;p(t)) \right\|_{L^\infty_z, L^2_y}		\\
	&\leq c_1 \epsilon^2 e^{c_2 t}
\end{split}
\end{equation}
where:
\begin{equation}
\begin{split}
	&c_1 := \frac{1}{4!} \left\| \sum_{|\alpha| = 4} | \de^\alpha_x W(x) | \right\|_{L_x^\infty(\field{R}^d)} \sup_{p \in S_n} \| \chi_n(z;p) \|_{L^\infty_z} \| a_0(y) \|_{\Sigma^4_y(\field{R}^d)} \\
	&c_2 := C_l \max_{\alpha,\beta \in \{1,...,d\}} \sup_{s \in [0,\infty)} \left\{ \vphantom{\frac{1}{2}} | \de_{p_\alpha}\de_{p_\beta} E_n(p(s)) |, | \de_{q_\alpha} \de_{q_\beta} W(q(s)) | \vphantom{\frac{1}{2}} \right\} 
\end{split}
\end{equation}
are constants independent of $t, \epsilon$. 

We conclude that there exist constants $C_1, C_2, C_3 > 0$, independent of $t, \epsilon$ such that:
\begin{equation}
\begin{split}
	&\| \eta_{3,0}^\epsilon(x) \|_{L^2_x} \leq C_1 \epsilon	\\
	&\| r_3^\epsilon(x,t) \|_{L^2_x} \leq C_2 e^{C_3 t} \epsilon^{2} 
\end{split}
\end{equation}
The bound (\ref{eq:bound_on_eta}) then follows from the basic a priori $L^2$ bound for solutions of the linear time-dependent Schr\"{o}dinger equation:
\begin{lemma} \label{lem:L2_bound}
Let $\psi(x,t)$ be the unique solution of:
\begin{equation}
\begin{split}
	&i \de_t \psi = H \psi + f	\\
	&\psi(x,0) = \psi_0(x)
\end{split}
\end{equation}
where $H$ is a self-adjoint operator. Then:
\begin{equation}
	\| \psi(\cdot,t) \|_{L^2(\field{R}^d)} \leq \| \psi_0(\cdot) \|_{L^2(\field{R}^d)} + \inty{0}{t}{ \| f(\cdot,t') \|_{L^2(\field{R}^d)} }{t'}
\end{equation} 
when $f = 0$, we have:
\begin{equation}
	\| \psi(\cdot,t) \|_{L^2(\field{R}^d)} = \| \psi_0(\cdot) \|_{L^2(\field{R}^d)} 
\end{equation}
\end{lemma} 
Applying Lemma \ref{lem:L2_bound} to equation (\ref{eq:equation_for_eta_3}) then gives the bound on $\eta^\epsilon_3(x,t)$:
\begin{equation}
\begin{split}
	&\| \eta_3^\epsilon(x,t) \|_{L^2_x} \leq \| \eta_3^\epsilon(x,0) \|_{L^2_x} + \frac{1}{\epsilon} \inty{0}{t}{ \| r_3^\epsilon(x,t') \|_{L^2_x} }{t'}  \\
	&\leq C_1 \epsilon + \frac{1}{\epsilon} \inty{0}{t}{ C_2 e^{C_3 t'} \epsilon^{2} }{t'}  \\
	&\leq C e^{C t} \epsilon 
\end{split}
\end{equation}
where $C$ is a constant independent of $\epsilon, t$. This completes the proof of Theorem \ref{th:periodic_background_nonseparable_isolated_band_theorem}.  

\section{Proof of Theorem \ref{prop:center_of_mass_proposition} on dynamics of physical observables} \label{sec:center_of_mass_calculation}
Let $\psi^\epsilon(x,t)$ be the solution of (\ref{eq:equation_and_initial_data}). By Theorem \ref{th:periodic_background_nonseparable_isolated_band_theorem} we have that this solution has the form:
\begin{equation} \label{eq:form_of_psi}
	\psi^\epsilon(x,t) = \left. \tilde{\psi}^\epsilon(y,z,t) \right|_{y = \frac{x - q(t)}{\epsilon^{1/2}},z = \frac{x}{\epsilon}} + \eta^\epsilon(x,t)
\end{equation}
where:
\begin{equation} \label{eq:tilde_psi_again}
\begin{split}
	&\tilde{\psi}^\epsilon(y,z,t) := \epsilon^{-d/4} e^{i S(t) / \epsilon} e^{ i p(t) \cdot y / \epsilon^{1/2} } e^{i \phi_B(t)} \left\{ a(y,t) \chi_n(z;p(t)) \vphantom{\epsilon^{1/2}} \right. \\
	&\left. + \epsilon^{1/2} \left[ \vphantom{\epsilon^{1/2}}  (- i \nabla_{y}) a(y,t) \cdot \nabla_{p} \chi_n(z;p(t)) + b(y,t) \chi_n(z;p(t)) \right] \right\}	
\end{split}
\end{equation}
In this section we compute the dynamics of the physical observables:
\begin{equation} \label{eq:bloch_observables_again}
\begin{split}
	&\mathcal{Q}^\epsilon(t) := \frac{1}{\mathcal{N}^\epsilon(t)} \inty{\field{R}^d}{}{ x \left| \tilde{\psi}^\epsilon(y,z,t) \right|^2_{z = \frac{x}{\epsilon},y = \frac{x - q(t)}{\epsilon^{1/2}}} }{x} 	\\
	&\mathcal{P}^\epsilon(t) := \frac{1}{\mathcal{N}^\epsilon(t)} \inty{\field{R}^d}{}{ \overline{ \tilde{\psi}^\epsilon(y,z,t) } \left( - i \epsilon^{1/2} \nabla_y \right) \left. \tilde{\psi}^\epsilon(y,z,t) \right|_{z = \frac{x}{\epsilon}, y = \frac{x - q(t)}{\epsilon^{1/2}}} }{x}  
\end{split}
\end{equation}
where:
\begin{equation} \label{eq:def_of_N}
	\mathcal{N}^\epsilon(t) = \inty{\field{R}^d}{}{ \left| \tilde{\psi}^\epsilon(y,z,t) \right|^2_{z = \frac{x}{\epsilon},y = \frac{x - q(t)}{\epsilon^{1/2}}} }{x} 
\end{equation}
\begin{remark}
Throughout this section we will employ a short-hand notation: 
\begin{equation}
\begin{split}
	&f^\epsilon(x,t) = O(\epsilon^K e^{c t}) \iff \exists c > 0, C > 0 \text{ independent of $t, \epsilon$ such that } \| f^\epsilon(x,t) \|_{L^2_x} \leq C \epsilon^K e^{c t} 	\\
	&g^\epsilon(t) = O(\epsilon^K e^{c t}) \iff \exists c > 0, C > 0 \text{ independent of $t, \epsilon$ such that } | g^\epsilon(t) | \leq C \epsilon^K e^{c t} 	\\
\end{split}
\end{equation} 
\end{remark}
We will use the following Lemma which is a mild generalization of that found in \cite{ilan_weinstein} (as Lemma 4.2):
\begin{lemma} \label{lem:basic_homogenization_lemma}
Let $f \in \mathcal{S}(\field{R}^d)$, $g$ smooth and periodic with respect to the lattice $\Lambda$, $s \in \field{R}$ a constant, and $\delta > 0$ an arbitrary positive parameter. Then for any positive integer $N > 0$:
\begin{equation} \label{eq:statement_of_theorem}
	\inty{\field{R}^d}{}{ f\left(x\right) g\left(\frac{x}{\delta} + \frac{s}{\delta^2}\right) }{x} = \left( \inty{\field{R}^d}{}{ f(x) }{x} \right) \left( \inty{\Omega}{}{ g(z) }{z} \right) + O(\delta^N). 
\end{equation}
\end{lemma}
For the proof, see Appendix \ref{app:proof_of_lemma}.
\subsection{Asymptotic expansion and dynamics of $\mathcal{N}^\epsilon(t)$}
By changing variables in the integral (\ref{eq:def_of_N}), we have that:
\begin{equation} \label{eq:N}
	\mathcal{N}^\epsilon(t) = \epsilon^{d/2} \inty{\field{R}^d}{}{ \left| \tilde{\psi}^\epsilon(y,z,t) \right|^2_{z = \frac{y}{\epsilon^{1/2}} + \frac{q(t)}{\epsilon} } }{y}. 
\end{equation}
Substituting (\ref{eq:tilde_psi_again}) into (\ref{eq:N}) gives:
\begin{equation} 
\begin{split}	
	&= \int_{\field{R}^d} \left\{ \overline{ a(y,t) \chi_n(z;p(t)) \vphantom{\epsilon^{1/2}} + \epsilon^{1/2} \left[ \vphantom{\epsilon^{1/2}}  (- i \nabla_{y}) a(y,t) \cdot \nabla_{p} \chi_n(z;p(t)) + b(y,t) \chi_n(z;p(t)) \right] } \right\}  \\
	&\left. \left\{ a(y,t) \chi_n(z;p(t)) \vphantom{\epsilon^{1/2}} + \epsilon^{1/2} \left[ \vphantom{\epsilon^{1/2}}  (- i \nabla_{y}) a(y,t) \cdot \nabla_{p} \chi_n(z;p(t)) + b(y,t) \chi_n(z;p(t)) \right] \right\} \right|_{z = \frac{y}{\epsilon^{1/2}} + \frac{q(t)}{\epsilon}} \text{d}y.    \\	
\end{split}
\end{equation}
We expand the product in the integral and apply Lemma \ref{lem:basic_homogenization_lemma} term by term with $s = q(t), \delta = \epsilon^{1/2}$. Since the $\chi_n$ are assumed normalized: for all $t \in [0,\infty)$ $\| \chi_n(\cdot;p(t)) \|_{L^2(\Omega)} = 1$, we have:
\begin{equation} \label{eq:N_expanded}
\begin{split}
	&\mathcal{N}^\epsilon(t) = \| a(y,t) \|^2_{L^2_y(\field{R}^d)} \\
	&+ \epsilon^{1/2} \left[ \ipy{(- i \nabla_y) a(y,t)}{a(y,t)} \cdot \ipz{\nabla_p \chi_n(z;p(t))}{\chi_n(z;p(t))} \right. \\
	&+ \ipy{a(y,t)}{(- i \nabla_y) a(y,t)} \cdot \ipz{\chi_n(z;p(t))}{\nabla_p \chi_n(z;p(t))} 	\\
	&\left. + \ipy{b(y,t)}{a(y,t)} + \ipy{a(y,t)}{b(y,t)} \right] + O(\epsilon e^{c t}).
\end{split}
\end{equation}
\begin{remark} \label{rem:remark_on_errors}
In (\ref{eq:N_expanded}) we have made explicit all terms through order $\epsilon^{1/2}$. To justify the error bound, consider that the remaining terms may be bounded by $\mathcal{C}(t) \epsilon$ where $\mathcal{C}(t)$ depends on $\Sigma^{l_1}_y$-norms of $a(y,t), b(y,t)$ and $L^2_z$-norms of $\de_{p}^{l_2} \chi_n(z;p(t))$ where $l_1, l_2$ are positive integers. By an identical reasoning to that given in Section \ref{sec:estimation_of_residual} we have that $\mathcal{C}(t)$ may be bounded by $C e^{c t}$ where $c > 0, C > 0$ are constants independent of $\epsilon, t$. Error terms of this type will arise throughout the following discussion and will be treated similarly. 
\end{remark}
 
Under Assumption \ref{isolated_band_assumption_nonseparable_case}, in a neighborhood of the curve $p(t) \in \mathcal{B}$, the mapping $p \mapsto \chi_n(z;p)$ is smooth. Hence, we may differentiate the normalization condition: $\| \chi_n(\cdot;p) \|^2_{L^2(\Omega)} = 1$ with respect to $p$ and evaluate along the curve $p(t)$ to obtain the identity: 
\begin{equation} \label{eq:de_p_chi_identity}
	\ipz{\chi_n(z;p(t))}{\nabla_p \chi_n(z;p(t))} + \ipz{\nabla_p \chi_n(z;p(t))}{\chi_n(z;p(t))} = 0	
\end{equation}
It follows from this, and the fact that $(- i \nabla_y)$ is symmetric with respect to the $L^2_y$-inner product, that:
\begin{equation}
\begin{split}
	&\ipy{(- i \nabla_y) a(y,t)}{a(y,t)} \cdot \ipz{\nabla_p \chi_n(z;p(t))}{\chi_n(z;p(t))}  \\
	&+ \ipy{a(y,t)}{(- i \nabla_y) a(y,t)} \cdot \ipz{\chi_n(z;p(t))}{\nabla_p \chi_n(z;p(t))} = 0
\end{split}
\end{equation}
so that (\ref{eq:N_expanded}) reduces to:
\begin{equation} \label{eq:expansion_of_N}
	\mathcal{N}^\epsilon(t) = \| a(y,t) \|^2_{L^2_y(\field{R}^d)} + \epsilon^{1/2} \left[ \ipy{b(y,t)}{a(y,t)} + \ipy{a(y,t)}{b(y,t)} \right] + O(\epsilon e^{c t}).
\end{equation}
From $L^2$-norm conservation for solutions of (\ref{eq:envelope_equation}), we have that $\| a(y,t) \|_{L^2_y} = \| a_0(y) \|_{L^2_y}$. In Appendix \ref{app:computation_of_dynamics_of_physical_observables} we calculate (\ref{eq:de_b_a_1}): 
\begin{equation}
	\fdf{t} \left[ \ipy{b(y,t)}{a(y,t)} + \ipy{a(y,t)}{b(y,t)} \right] = 0,
\end{equation}
so that: 
\begin{equation} \label{eq:N_conserved}
	\dot{ \mathcal{N} }^\epsilon(t) = O(\epsilon e^{C t}) 
\end{equation}
Integrating in time then gives:
\begin{equation} \label{eq:N_solution}
\begin{split}
	\mathcal{N}^\epsilon(t) &= \mathcal{N}^\epsilon(0) + O(\epsilon e^{c t})	\\
	&= \| a_0(y) \|^2_{L^2_y(\field{R}^d)} + \epsilon^{1/2} \left[ \ipy{b_0(y)}{a_0(y)} + \ipy{a_0(y)}{b_0(y)} \right] + O(\epsilon e^{c t}).
\end{split}
\end{equation}

\subsection{Asymptotic expansion of $\mathcal{Q}^\epsilon(t), \mathcal{P}^\epsilon(t)$; proof of assertion (1) of Theorem \ref{prop:center_of_mass_proposition}}
\noindent Changing variables in the integrals (\ref{eq:bloch_observables_again}) and using the identity:
\begin{equation}
	x = \left. q(t) + \epsilon^{1/2} y \right|_{y = \frac{x - q(t)}{\epsilon^{1/2}} }
\end{equation}
we have: 
\begin{equation} \label{eq:Q_P}
\begin{split}
	&\mathcal{Q}^\epsilon(t) = q(t) + \epsilon^{1/2 + d/2} \frac{1}{\mathcal{N}^\epsilon(t)} \inty{\field{R}^d}{}{ y \left| \tilde{\psi}^\epsilon(y,z,t) \right|^2_{z = \frac{q(t)}{\epsilon} + \frac{y}{\epsilon^{1/2}}} }{y}  	\\
	&\mathcal{P}^\epsilon(t) = \epsilon^{1/2 + d/2} \frac{1}{\mathcal{N}^\epsilon(t)} \inty{\field{R}^d}{}{ \overline{ \tilde{\psi}^\epsilon(y,z,t) } \left( - i \nabla_y \right) \left. \tilde{\psi}^\epsilon(y,z,t) \right|_{z = \frac{q(t)}{\epsilon} + \frac{y}{\epsilon^{1/2}}} }{y}.	\\
\end{split}
\end{equation}
Substituting (\ref{eq:tilde_psi_again}) into (\ref{eq:Q_P}) we have, for each $\alpha \in \{1,...,d\}$: 
\begin{equation} \label{eq:Q_P_again}
\begin{split}
	&\mathcal{Q}_\alpha^\epsilon(t) = q_\alpha(t) + \epsilon^{1/2} \frac{1}{\mathcal{N}^\epsilon(t)} \int_{\field{R}^d} y_\alpha \left| a(y,t) \chi_n(z;p(t)) \vphantom{\frac{1}{2}} \right. \\
	&\left. + \epsilon^{1/2} \left[ (- i \de_{y_\beta}) a(y,t) \de_{p_\beta} \chi_n(z;p(t)) + b(y,t) \chi_n(z;p(t)) \right] \right|^2_{z = \frac{q(t)}{\epsilon} + \frac{y}{\epsilon^{1/2}}} \text{d}y \\
	&\mathcal{P}_\alpha^\epsilon(t) = p_\alpha(t) + \epsilon^{1/2} \frac{1}{\mathcal{N}^\epsilon(t)} \int_{\field{R}^d} \overline{ a(y,t) \chi_n(z;p(t)) + \epsilon^{1/2} \left[ (- i \de_{y_\beta}) a(y,t) \de_{p_\beta} \chi_n(z;p(t)) + b(y,t) \chi_n(z;p(t)) \right] } \\
	&\left( - i \de_{y_\alpha} \right) \left( \left. a(y,t) \chi_n(z;p(t)) + \epsilon^{1/2} \left[ (- i \de_{y_\beta}) a(y,t) \de_{p_\beta} \chi_n(z;p(t)) + b(y,t) \chi_n(z;p(t)) \right] \right) \right|_{z = \frac{q(t)}{\epsilon} + \frac{y}{\epsilon^{1/2}}} \text{d}y \\
\end{split}
\end{equation}
Expanding all products and applying Lemma \ref{lem:basic_homogenization_lemma} term by term in (\ref{eq:Q_P_again}) we obtain: 
\begin{equation} 
\begin{split}
	&\mathcal{Q}_\alpha^\epsilon(t) = {q}_\alpha(t) + \epsilon^{1/2} \frac{1}{\mathcal{N}^\epsilon(t)} \ipy{a(y,t)}{y_\alpha a(y,t)} \\
	&+ \epsilon \frac{1}{\mathcal{N}^\epsilon(t)} \left[ \ipy{b(y,t)}{y_\alpha a(y,t)} + \ipy{a(y,t)}{y_\alpha b(y,t)} \right.	\\
	&+ \ipy{(- i \de_{y_\beta}) a(y,t)}{y_\alpha a(y,t)} \ipz{\de_{p_\beta} \chi_n(z;p(t))}{\chi_n(z;p(t))}  \\
	&\left. + \ipy{y_\alpha a(y,t)}{(- i \de_{y_\beta}) a(y,t)} \ipz{\chi_n(z;p(t))}{\de_{p_\beta} \chi_n(z;p(t))} \right]	\\
	&+ O(\epsilon^{3/2} e^{c t})	\\
	&\mathcal{P}_\alpha^\epsilon(t) = {p}_\alpha(t) + \epsilon^{1/2} \frac{1}{\mathcal{N}^\epsilon(t)} \ipy{a(y,t)}{ (- i \de_{y_\alpha}) a(y,t) } 	\\
	&+ \epsilon \frac{1}{\mathcal{N}^\epsilon(t)} \left[ \ipy{b(y,t)}{(- i \de_{y_\alpha}) a(y,t)} + \ipy{a(y,t)}{(- i \de_{y_\alpha}) b(y,t)} \right] \\
	&+ O(\epsilon^{3/2} e^{c t}). 
\end{split}
\end{equation}
Here, terms of higher order than $\epsilon$ are bounded by a similar reasoning to that given in Section \ref{sec:estimation_of_residual} (Remark \ref{rem:remark_on_errors}). Using the identity (\ref{eq:de_p_chi_identity}) and the fact that $(- i \nabla_y)$ is self-adjoint we have that:
\begin{equation}
\begin{split}
	&\ipy{(- i \de_{y_\beta}) a(y,t)}{y_\alpha a(y,t)} \ipz{\de_{p_\beta} \chi_n(z;p(t))}{\chi_n(z;p(t))}  \\
	&+ \ipy{y_\alpha a(y,t)}{(- i \de_{y_\beta}) a(y,t)} \ipz{\chi_n(z;p(t))}{\de_{p_\beta} \chi_n(z;p(t))} 	\\
	&= \ipy{a(y,t)}{[y_\alpha,(- i \de_{y_\beta})] a(y,t)} \ipz{\chi_n(z;p(t))}{\de_{p_\beta} \chi_n(z;p(t))} 
\end{split}
\end{equation}
where $[y_\alpha,(- i \de_{y_\beta})] := y_\alpha (- i \de_{y_\beta}) - (- i \de_{y_\beta}) y_\alpha$ is the commutator. Since $[y_\alpha,(- i \de_{y_\beta})] = i \delta_{\alpha \beta}$ We have that: 
\begin{equation}
\begin{split}
	&\ipy{a(y,t)}{[y_\alpha,(- i \de_{y_\beta})] a(y,t)} \ipz{\chi_n(z;p(t))}{\de_{p_\beta} \chi_n(z;p(t))} \\
	&= i \| a(y,t) \|^2_{L^2_y(\field{R}^d)} \ipz{\chi_n(z;p(t))}{\de_{p_\alpha} \chi_n(z;p(t))}.
\end{split}
\end{equation} 
Using $L^2$-norm conservation for solutions of (\ref{eq:envelope_equation}), we have that for all $t \geq 0$, $\| a(y,t) \|^2_{L^2_y(\field{R}^d)} = \| a_0(y) \|^2_{L^2_y(\field{R}^d)}$. Using (\ref{eq:N_solution}), we have that $\| a_0(y) \|^2_{L^2_y(\field{R}^d)} = \mathcal{N}^\epsilon(t) + O(\epsilon^{1/2} e^{c t})$ (\ref{eq:expansion_of_N}). We have proved that:
\begin{equation}
\begin{split}
	&i \| a(y,t) \|^2_{L^2_y(\field{R}^d)} \ipz{\chi_n(z;p(t))}{\nabla_p \chi_n(z;p(t))} \\
	&= \mathcal{N}^\epsilon(t) \mathcal{A}_n(p(t)) + O(\epsilon^{1/2} e^{c t})
\end{split}
\end{equation}
where the last equality holds by the definition of the $n$-th band Berry connection (\ref{eq:berry_connection}). We have proved that:
\begin{equation} \label{eq:QP_Rd}
\begin{split}
	&\mathcal{Q}^\epsilon(t) = {q}(t) + \epsilon^{1/2} \frac{1}{\mathcal{N}^\epsilon(t)}\ipy{a(y,t)}{y a(y,t)} \\
	&+ \epsilon \frac{1}{\mathcal{N}^\epsilon(t)} \left[ \ipy{b(y,t)}{y a(y,t)} + \ipy{a(y,t)}{y b(y,t)} \right]	\\
	&+ \epsilon \mathcal{A}_n(p(t))	+ O(\epsilon^{3/2} e^{c t})	\\
	&\mathcal{P}^\epsilon(t) = {p}(t) + \epsilon^{1/2} \frac{1}{\mathcal{N}^\epsilon(t)} \ipy{a(y,t)}{ (- i \nabla_{y}) a(y,t) } 	\\
	&+ \epsilon \frac{1}{\mathcal{N}^\epsilon(t)} \left[ \ipy{b(y,t)}{(- i \nabla_{y}) a(y,t)} + \ipy{a(y,t)}{(- i \nabla_{y}) b(y,t)} \right] \\
	&+ O(\epsilon^{3/2} e^{c t}). 
\end{split}
\end{equation}

\subsection{Computation of dynamics of $\mathcal{Q}^\epsilon(t), \mathcal{P}^\epsilon(t)$; proof of assertion (2) of Theorem \ref{prop:center_of_mass_proposition}}
Differentiating both sides of (\ref{eq:QP_Rd}) with respect to time and using $\dot{\mathcal{N}}^\epsilon(t) = O(\epsilon e^{c t})$ (\ref{eq:N_conserved}) gives:
\begin{equation} \label{eq:differentiated_QP_equations}
\begin{split}
	&\dot{\mathcal{Q}}^\epsilon_\alpha(t) = \dot{q}_\alpha(t) + \epsilon^{1/2} \frac{1}{\mathcal{N}^\epsilon(t)} \fdf{t} \ipy{a(y,t)}{y_\alpha a(y,t)} \\
	&+ \epsilon \frac{1}{\mathcal{N}^\epsilon(t)} \fdf{t} \left[ \ipy{b(y,t)}{y_\alpha a(y,t)} + \ipy{a(y,t)}{y_\alpha b(y,t)} \right] \\
	&+ \epsilon \dot{p}_\beta(t) \de_{p_\beta} { \mathcal{A}_{n,\alpha} }(p(t)) + O(\epsilon^{3/2} e^{c t})	\\
	&\dot{\mathcal{P}}^\epsilon_\alpha(t) = \dot{p}_\alpha(t) + \epsilon^{1/2} \frac{1}{\mathcal{N}^\epsilon(t)} \fdf{t} \ipy{a(y,t)}{(- i \de_{y_\alpha}) a(y,t)}  \\
	&+ \epsilon \frac{1}{\mathcal{N}^\epsilon(t)} \fdf{t} \left[ \ipy{b(y,t)}{ (- i \de_{y_\alpha}) a(y,t) } + \ipy{a(y,t)}{(- i \de_{y_\alpha}) b(y,t)} \right] \\
	&+ O(\epsilon^{3/2} e^{c t})	\\
\end{split}
\end{equation}
Recall that $(q(t),p(t))$ satisfy the classical system (\ref{eq:classical_system}). In Appendix \ref{app:computation_of_dynamics_of_physical_observables} we calculate (\ref{eq:Q_dot_P_dot_a}): 
\begin{equation}
\begin{split}
	&\fdf{t} \ipy{a(y,t)}{y_\alpha a(y,t)} = \de_{p_\alpha} \de_{p_\beta} E_n(p(t)) \ipy{a(y,t)}{ (- i \de_{y_\beta}) a(y,t)} \\
	&\fdf{t} \ipy{a(y,t)}{(- i \de_{y_\alpha}) a(y,t)} = - \de_{q_\alpha} \de_{q_\beta} W(q(t)) \ipy{a(y,t)}{y_\beta a(y,t)}   	\\
\end{split}
\end{equation}
and (\ref{eq:a_b_more}):
\begin{equation}
\begin{split}
	&\fdf{t} \left[ \ipy{b({y},t)}{y_\alpha a({y},t) } + \ipy{ a(y,t) }{ y_\alpha b({y},t) } \right]	\\
	&= \de_{p_\alpha} \de_{p_\beta} E_n(p(t)) \left[ \ipy{b(y,t)}{(- i \de_{y_\beta}) a(y,t)} + \ipy{a(y,t)}{(- i \de_{y_\beta}) a(y,t)} \right] \\ 
	&+ \frac{1}{2} \de_{p_\alpha} \de_{p_\beta} \de_{p_\gamma} E_n(p(t)) \ipy{a(y,t)}{(- i \de_{y_\beta})(- i \de_{y_\gamma}) a(y,t)} + \de_{q_\beta} W(q(t)) \de_{p_\alpha} \mathcal{A}_{n,\beta}(p(t)) \| a(y,t) \|^2_{L^2_y(\field{R}^d)}  \\
	&\fdf{t} \left[ \ipy{b(y,t)}{(- i \de_{y_\alpha}) a(y,t)} + \ipy{a(y,t)}{(- i \de_{y_\alpha}) b(y,t)} \right] 	\\
	&= - \de_{q_\alpha} \de_{q_\beta} W(q(t)) \left[ \ipy{b(y,t)}{y_\beta a(y,t)} + \ipy{a(y,t)}{y_\beta a(y,t)} \right] \\ 
	&- \frac{1}{2} \de_{q_\alpha} \de_{q_\beta} \de_{q_\gamma} W(q(t)) \ipy{a(y,t)}{y_\beta y_\gamma a(y,t)} - \de_{q_\alpha} \de_{q_\beta} W(q(t)) \mathcal{A}_{n,\beta}(p(t)) \| a(y,t) \|^2_{L^2_y(\field{R}^d)}   \\
\end{split}
\end{equation}
Substituting these expressions into (\ref{eq:differentiated_QP_equations}) and using $\| a(y,t) \|^2_{L^2_y(\field{R}^d)} = \mathcal{N}^\epsilon(t) + O(\epsilon^{1/2} e^{c t})$ (\ref{eq:N_expanded}) we have: 
\begin{equation} \label{eq:QdotPdot_Rd}
\begin{split}
	&\dot{\mathcal{Q}}^\epsilon_\alpha(t) = \de_{{p_\alpha}} E_n({p}(t)) + \epsilon^{1/2} \frac{1}{\mathcal{N}^\epsilon(t)} \de_{p_\alpha} \de_{p_\beta} E_n(p(t)) \ipy{a(y,t)}{(- i \de_{y_\beta})a(y,t)} \\
	&+ \epsilon  \frac{1}{\mathcal{N}^\epsilon(t)} \de_{p_\alpha} \de_{p_\beta} E_n(p(t)) \left[ \ipy{b(y,t)}{(- i \de_{y_\beta})a(y,t)} + \ipy{a(y,t)}{(- i \de_{y_\beta}) a(y,t)} \right] \\ 
	&+ \epsilon  \frac{1}{2} \frac{1}{\mathcal{N}^\epsilon(t)}  \de_{p_\alpha} \de_{p_\beta} \de_{p_\gamma} E_n(p(t)) \ipy{a(y,t)}{(- i \de_{y_\beta}) (- i \de_{y_\gamma}) a(y,t)}  	\\
	&+ \epsilon \de_{q_\beta} W(q(t)) \de_{p_\alpha} \mathcal{A}_{n,\beta}(p(t)) - \epsilon \de_{q_\beta} W(q(t)) \de_{{p_\beta}} { \mathcal{A}_{n,\alpha} }(p(t)) + O(\epsilon^{3/2} e^{c t})	\\ 
	&\dot{\mathcal{P}}^\epsilon_\alpha(t) = - \de_{{q_\alpha}} W(q(t)) - \epsilon^{1/2} \frac{1}{\mathcal{N}^\epsilon(t)}  \de_{q_\alpha} \de_{q_\beta} W(q(t)) \ipy{a(y,t)}{y_\beta a(y,t)}   	\\
	&- \epsilon \frac{1}{\mathcal{N}^\epsilon(t)}  \de_{q_\alpha} \de_{q_\beta} W(q(t)) \left[ \ipy{b(y,t)}{y_\beta a(y,t)} + \ipy{b(y,t)}{y_\beta a(y,t)} \right] \\ 
	&- \epsilon \frac{1}{2} \frac{1}{\mathcal{N}^\epsilon(t)}  \de_{q_\alpha} \de_{q_\beta} \de_{q_\gamma} W(q(t)) \ipy{a(y,t)}{ y_\beta y_\gamma a(y,t)} 	\\
	&- \epsilon \de_{q_\alpha} \de_{q_\beta} W(q(t)) \mathcal{A}_{n,\beta}(p(t)) + O(\epsilon^{3/2} e^{c t})  \\
\end{split}
\end{equation}
Equation (\ref{eq:QP_Rd}) gives expressions for $q(t),p(t)$ in terms of $\mathcal{Q}^\epsilon(t), \mathcal{P}^\epsilon(t)$:
\begin{equation}
\begin{split}
	&{q_\alpha}(t) = \mathcal{Q}^\epsilon_\alpha(t) - \epsilon^{1/2} \frac{1}{\mathcal{N}^\epsilon(t)} \ipy{a(y,t)}{y_\alpha a(y,t)} \\
	&- \epsilon \frac{1}{\mathcal{N}^\epsilon(t)} \left[ \ipy{b(y,t)}{y_\alpha a(y,t)} + \ipy{a(y,t)}{y_\alpha b(y,t)} \right] - \epsilon \mathcal{A}_{n,\alpha}(p(t)) + O(\epsilon^{3/2} e^{c t})	\\
	&p_\alpha(t) = \mathcal{P}^\epsilon_\alpha(t) - \epsilon^{1/2} \frac{1}{\mathcal{N}^\epsilon(t)} \ipy{a(y,t)}{(- i \de_{y_\alpha}) a(y,t)}  	\\
	&- \epsilon \frac{1}{\mathcal{N}^\epsilon(t)} \left[ \ipy{b(y,t)}{(- i \de_{y_\alpha}) a(y,t)} + \ipy{a(y,t)}{(- i \de_{y_\alpha}) b(y,t)} \right] + O(\epsilon^{3/2} e^{c t})	\\
\end{split}
\end{equation}
Substituting these expressions into (\ref{eq:QdotPdot_Rd}), Taylor-expanding in $\epsilon^{1/2}$, and again using $\mathcal{N}^\epsilon(t) = \| a_0(y) \|^2_{L^2_y(\field{R}^d)} + O(\epsilon^{1/2} e^{c t})$ (\ref{eq:N_expanded}) then gives: 
\begin{equation} \label{eq:full_equations_for_QP}
\begin{split}
	&\dot{\mathcal{Q}}^\epsilon_\alpha(t) = \de_{\mathcal{P}^\epsilon_\alpha} E_n(\mathcal{P}^\epsilon(t)) + \epsilon \frac{1}{2} \de_{\mathcal{P}^\epsilon_\alpha} \de_{\mathcal{P}^\epsilon_\beta} \de_{\mathcal{P}^\epsilon_\gamma} E_n(\mathcal{P}^\epsilon(t)) \left[ \frac{1}{ \| a_0(y) \|^2_{L^2_y(\field{R}^d)} } \ipy{a(y,t)}{(- i \de_{y_\beta})(- i \de_{y_\gamma}) a(y,t) } \right. \\
	&\left. - \frac{1}{ \| a_0(y) \|^4_{L^2_y(\field{R}^d)}} \ipy{ a(y,t) }{  (- i \de_{y_\beta}) a(y,t) } \ipy{a(y,t)}{ (- i \de_{y_\gamma}) a(y,t) }   \right]   \\
	&+ \epsilon \de_{\mathcal{Q}^\epsilon_\gamma} W(\mathcal{Q}^\epsilon(t)) \mathcal{F}_{n,\alpha \gamma}(\mathcal{P}^\epsilon(t)) + O(\epsilon^{3/2} e^{c t})		\\
	&\dot{\mathcal{P}}^\epsilon_\alpha(t) = - \de_{\mathcal{Q}^\epsilon_\alpha} W(\mathcal{Q}^\epsilon(t)) - \epsilon \frac{1}{2} \de_{\mathcal{Q}^\epsilon_\alpha} \de_{\mathcal{Q}^\epsilon_\beta} \de_{\mathcal{Q}^\epsilon_\gamma} W(\mathcal{Q}^\epsilon(t)) \left[ \frac{1}{ \| a_0(y) \|^2_{L^2_y(\field{R}^d)} } \ipy{a(y,t)}{ y_\beta y_\gamma a(y,t) } 	\right.	\\
	&\left. - \frac{1}{ \| a_0(y) \|^4_{L^2_y(\field{R}^d)} } \ipy{ a(y,t) }{ y_\beta a(y,t) } \ipy{ a(y,t) }{ y_\gamma a(y,t) } \right] + O(\epsilon^{3/2} e^{c t}).
\end{split}
\end{equation}
where $\mathcal{F}_{n,\alpha \gamma}(\mathcal{P}^\epsilon) := \de_{\mathcal{P}^\epsilon_\alpha} \mathcal{A}_{n,\gamma}(\mathcal{P}^\epsilon) - \de_{\mathcal{P}^\epsilon_\beta} \mathcal{A}_{n,\alpha}(\mathcal{P}^\epsilon)$ is the $n$th band Berry curvature (\ref{eq:berry_curvature}). Note that the system (\ref{eq:full_equations_for_QP}) is not closed: $a(y,t)$ satisfies an equation parametrically forced by $q(t),p(t)$ (\ref{eq:envelope_equation}). Recall the definition of $a^\epsilon(y,t)$ (\ref{eq:a_epsilon_equation}) as the solution of (\ref{eq:envelope_equation}) with co-efficients evaluated at $\mathcal{Q}^\epsilon(t), \mathcal{P}^\epsilon(t)$:
\begin{equation} \label{eq:envelope_equation_new_PQ}
\begin{split}
	&i \de_t a^\epsilon(y,t) = \frac{1}{2} \de_{\mathcal{P}^\epsilon_\alpha} \de_{\mathcal{P}^\epsilon_\beta} E_n(\mathcal{P}^\epsilon(t)) (- i \de_{y_\alpha})( - i \de_{y_\beta}) a^\epsilon(y,t) + \frac{1}{2} \de_{\mathcal{Q}^\epsilon_\alpha} \de_{\mathcal{Q}^\epsilon_\beta} W(\mathcal{Q}^\epsilon(t)) y_\alpha y_\beta a^\epsilon(y,t)	\\
	&a^\epsilon(y,0) = a_0(y),
\end{split}
\end{equation}
Recall the definition of the $\Sigma^l$ norms (\ref{eq:sigma_spaces}). If we can show that $\| a^\epsilon(y,t) - a(y,t) \|_{\Sigma^l_y(\field{R}^d)} = O(\epsilon^{1/2} e^{c t})$ for each positive integer $l$, then we may replace $a(y,t)$ by $a^\epsilon(y,t)$ everywhere in (\ref{eq:full_equations_for_QP}) and, after dropping error terms, we will have obtained a closed system for $\mathcal{Q}^\epsilon(t), \mathcal{P}^\epsilon(t), a^\epsilon(y,t)$. Let:
\begin{equation}
\begin{split}
	&\mathscr{H}^\epsilon(t) := \frac{1}{2} \de_{\mathcal{P}^\epsilon_\alpha} \de_{\mathcal{P}^\epsilon_\beta} E_n(\mathcal{P}^\epsilon(t)) (- i \de_{y_\alpha})(- i \de_{y_\beta}) + \frac{1}{2} \de_{\mathcal{Q}^\epsilon_\alpha} \de_{\mathcal{Q}^\epsilon_\beta} W(q(t)) y_\alpha y_\beta,	\\
	&\mathscr{H}(t) := \frac{1}{2} \de_{p_\alpha} \de_{p_\beta} E_n(p(t)) (- i \de_{y_\alpha})(- i \de_{y_\beta}) + \frac{1}{2} \de_{q_\alpha} \de_{q_\beta} W(q(t)) y_\alpha y_\beta	\\
\end{split}
\end{equation} 
then $a^\epsilon(y,t) - a(y,t)$ satisfies:
\begin{equation} \label{eq:eom_a_minus_a}
\begin{split}
	&i \de_t \left( a^\epsilon(y,t) - a(y,t) \right) = {\mathscr{H}}^\epsilon(t) a^\epsilon(y,t) - {\mathscr{H}}(t) a(y,t)	\\
	&= {\mathscr{H}}^\epsilon(t) \left( a^\epsilon(y,t) - a(y,t) \right) + \left( {\mathscr{H}}^\epsilon(t) - {\mathscr{H}}(t) \right) a(y,t)	\\
	&a^\epsilon(y,0) - a(y,0) = 0
\end{split}
\end{equation} 
Using the fact that ${\mathscr{H}}^\epsilon(t)$ is self-adjoint on $L^2_y(\field{R}^d)$ for each $t$, it follows from (\ref{eq:eom_a_minus_a}) that:
\begin{equation}
\begin{split}
	&\fdf{t} \| a^\epsilon(y,t) - a(y,t) \|^2_{L^2_y(\field{R}^d)} \\
	&= i \ipy{ \left( {\mathscr{H}}^\epsilon(t) - {\mathscr{H}}(t) \right) a(y,t) }{ a^\epsilon(y,t) - a(y,t) } - i \ipy{ a^\epsilon(y,t) - a(y,t) }{ \left( {\mathscr{H}}^\epsilon(t) - {\mathscr{H}}(t) \right) a(y,t) }	\\
\end{split}
\end{equation}
By the Cauchy-Schwarz inequality:
\begin{equation}
	\leq 2 \| a^\epsilon(y,t) - a(y,t) \|_{L^2_y(\field{R}^d)} \| \left( {\mathscr{H}}^\epsilon(t) - {\mathscr{H}}(t) \right) a(y,t) \|_{L^2_y(\field{R}^d)}
\end{equation}
It then follows that:
\begin{equation}
	\| a^\epsilon(y,t) - a(y,t) \|_{L^2_y(\field{R}^d)} \leq \inty{0}{t}{ \| \left( {\mathscr{H}}^\epsilon(s) - {\mathscr{H}}(s) \right) a(y,s) \|_{L^2_y(\field{R}^d)} }{s}.
\end{equation}
Using the precise forms of ${\mathscr{H}}^\epsilon(t), {\mathscr{H}}(t)$ we have:
\begin{equation}
\begin{split}
	&\| a^\epsilon(y,t) - a(y,t) \|_{L^2_y(\field{R}^d)} \leq \int_{0}^{t}  \\
	&\sup_{s \in [0,\infty), \alpha, \beta \in \{1,...,d\}} \left( |\de_{\mathcal{P}^\epsilon_\alpha}\de_{\mathcal{P}^\epsilon_\beta} E_n(\mathcal{P}^\epsilon(s)) - \de_{p_\alpha} \de_{p_\beta} E_n(p(s)) | + | \de_{\mathcal{Q}^\epsilon_\alpha} \de_{\mathcal{Q}^\epsilon_\beta} W(\mathcal{Q}^\epsilon(s)) - \de_{q_\alpha} \de_{q_\beta} W(q(s)) | \right) 	\\
	&\| a(y,s) \|_{\Sigma^2_y(\field{R}^d)} \text{ d}s
\end{split}
\end{equation}
where $\Sigma^2_y(\field{R}^d)$ is the norm defined in (\ref{eq:sigma_spaces}). Recall that $| \mathcal{Q}^\epsilon(t) - q(t)| + | \mathcal{P}^\epsilon(t) - p(t) | = O(\epsilon^{1/2} e^{c t})$ (\ref{eq:QP_Rd}). It follows from compactness of the Brillouin zone and Assumptions \ref{isolated_band_assumption_nonseparable_case} and \ref{W_assumption} that there exists a uniform bound in $t$ on third derivatives of $E_n(p), W(q)$ for all $p$ along the line segments connecting $p(t)$ and $\mathcal{P}^\epsilon(t)$, and all $q$ along the line segments connecting $q(t)$ and $\mathcal{Q}^\epsilon(t)$. We may therefore conclude from the mean-value theorem that there exist constants $c > 0, C > 0$ independent of $\epsilon, t$ such that:
\begin{equation}
	\| a^\epsilon(y,t) - a(y,t) \|_{L^2_y(\field{R}^d)} \leq C \epsilon^{1/2} \inty{0}{t}{ e^{c s}  \| a(y,s) \|_{\Sigma^2_y(\field{R}^d)} }{s}.
\end{equation}
We now use the a priori bounds on the $\Sigma^l_y(\field{R}^d)$-norms of $a(y,t)$ for each $l \in \field{N}$ (Lemma \ref{bounds_in_sigma_l}) to see that: 
\begin{equation}
	\| a^\epsilon(y,t) - a(y,t) \|_{L^2_y(\field{R}^d)} \leq \epsilon^{1/2} C' e^{c' t}. 
\end{equation}
for some constants $c'>0, C'>0$ independent of $\epsilon, t$. By a similar argument, we see that for any integer $l \geq 0$ there exist a constants $c'_l > 0, C'_l > 0$ such that:
\begin{equation}
	\| a^\epsilon(y,t) - a(y,t) \|_{\Sigma^l_y(\field{R}^d)} \leq \epsilon^{1/2} C'_l e^{c'_l t}. 
\end{equation}
It then follows that we may replace $a(y,t)$ by $a^\epsilon(y,t)$ everywhere in (\ref{eq:full_equations_for_QP}), generating further errors which are $O(\epsilon^{3/2} e^{c t})$ to derive:
\begin{equation} \label{eq:full_Q_P_equations_again}
\begin{split}
	&\dot{\mathcal{Q}}^\epsilon_\alpha(t) = \de_{\mathcal{P}^\epsilon_\alpha} E_n(\mathcal{P}^\epsilon(t)) + \epsilon \frac{1}{2} \de_{\mathcal{P}^\epsilon_\alpha} \de_{\mathcal{P}^\epsilon_\beta} \de_{\mathcal{P}^\epsilon_\gamma} E_n(\mathcal{P}^\epsilon(t)) \left[ \frac{1}{\| a_0(y) \|^2_{L^2_y(\field{R}^d)}} \ipy{a^\epsilon(y,t)}{(- i \de_{y_\beta})(- i \de_{y_\gamma}) a^\epsilon(y,t) } \right. \\
	&\left. -  \frac{1}{\| a_0(y) \|^4_{L^2_y(\field{R}^d)}} \ipy{ a^\epsilon(y,t) }{  (- i \de_{y_\beta}) a^\epsilon(y,t) } \ipy{a^\epsilon(y,t)}{ (- i \de_{y_\gamma}) a^\epsilon(y,t) } \right]   \\
	&+ \epsilon \de_{\mathcal{Q}^\epsilon_\gamma} W(\mathcal{Q}^\epsilon(t)) \mathcal{F}_{n, \alpha \gamma}(\mathcal{P}^\epsilon(t)) + O(\epsilon^{3/2} e^{c t})	\\
	&\dot{\mathcal{P}}^\epsilon_\alpha(t) = - \de_{\mathcal{Q}^\epsilon_\alpha} W(\mathcal{Q}^\epsilon(t)) - \epsilon \frac{1}{2} \de_{\mathcal{Q}^\epsilon_\alpha} \de_{\mathcal{Q}^\epsilon_\beta} \de_{\mathcal{Q}^\epsilon_\gamma} W(\mathcal{Q}^\epsilon(t)) \left[  \frac{1}{\| a_0(y) \|^2_{L^2_y(\field{R}^d)}} \ipy{a^\epsilon(y,t)}{ y_\beta y_\gamma a^\epsilon(y,t) } 	\right.	\\
	&\left. - \frac{1}{\| a_0(y) \|^4_{L^2_y(\field{R}^d)}} \ipy{ a^\epsilon(y,t) }{ y_\beta a^\epsilon(y,t) } \ipy{ a^\epsilon(y,t) }{ y_\gamma a^\epsilon(y,t) } \right] + O(\epsilon^{3/2} e^{c t}).
\end{split}
\end{equation}

\subsection{Hamiltonian structure of dynamics of $\mathcal{Q}^\epsilon(t), \mathcal{P}^\epsilon(t)$; proof of assertion (3) of Theorem \ref{prop:center_of_mass_proposition}}
Following \cite{e_lu_yang}, we introduce the new variables (\ref{eq:change_to_canonicals}): 
\begin{equation} \label{eq:change_to_canonicals_again}
\begin{split}
	&\mathscr{Q}^\epsilon(t) := \mathcal{Q}^\epsilon(t) - \epsilon \mathcal{A}_n(\mathcal{P}^\epsilon(t))	\\
	&\mathscr{P}^\epsilon(t) := \mathcal{P}^\epsilon(t).
\end{split}
\end{equation}
Let $\mathfrak{a}^\epsilon(y,t)$ denote the solution of (\ref{eq:a_epsilon_equation}) with co-efficients evaluated at $\mathscr{Q}^\epsilon(t), \mathscr{P}^\epsilon(t)$ rather than $\mathcal{Q}^\epsilon(t), \mathcal{P}^\epsilon(t)$, with initial data normalized in $L^2_y(\field{R}^d)$:
\begin{equation} \label{eq:mathfrak_a_equation}
\begin{split}
	&i \de_t \mathfrak{a}^\epsilon(y,t) = \frac{1}{2} \de_{\mathscr{P}^\epsilon_\alpha} \de_{\mathscr{P}^\epsilon_\beta} E_n(\mathscr{P}^\epsilon(t)) (- i \de_{y_\alpha}) (- i \de_{y_\beta}) \mathfrak{a}^\epsilon(y,t) + \frac{1}{2} \de_{\mathscr{Q}^\epsilon_\alpha} \de_{\mathscr{Q}^\epsilon_\beta} W(\mathscr{Q}^\epsilon(t)) y_\alpha y_\beta \mathfrak{a}^\epsilon(y,t)	\\
	&\mathfrak{a}^\epsilon(y,0) = \frac{a_0(y)}{\| a_0(y) \|_{L^2_y(\field{R}^d)} }.
\end{split}
\end{equation}
Since $\mathscr{Q}^\epsilon(t) - \mathcal{Q}^\epsilon(t) = O(\epsilon e^{c t})$, $\mathscr{P}^\epsilon(t) = \mathcal{P}^\epsilon(t)$, by a similar argument to that given in the previous section we have that for each integer $l \geq 0$:
\begin{equation} \label{eq:can_replace_a}
	\left\| \frac{a^\epsilon(y,t)}{\|a_0(y)\|_{L^2_y(\field{R}^d)}} - \mathfrak{a}^\epsilon(y,t) \right\|_{\Sigma^l_y(\field{R}^d)} = O(\epsilon e^{c t}).
\end{equation}
Differentiating (\ref{eq:change_to_canonicals_again}), using equations (\ref{eq:full_Q_P_equations_again}) for $\dot{\mathcal{Q}}^\epsilon(t), \dot{\mathcal{P}}^\epsilon(t)$, and using (\ref{eq:can_replace_a}) to replace $\frac{a^\epsilon(y,t)}{\| a_0(y) \|_{L^2_y(\field{R}^d)}}$ with $\mathfrak{a}^\epsilon(y,t)$ everywhere we obtain: 
\begin{equation} \label{eq:Q_P_equation_canonicals}
\begin{split}
	&\dot{\mathscr{Q}}^\epsilon_\alpha(t) = \de_{\mathscr{P}^\epsilon_\alpha} E_n(\mathscr{P}^\epsilon(t)) + \epsilon \frac{1}{2} \de_{\mathscr{P}^\epsilon_\alpha} \de_{\mathscr{P}_\beta} \de_{\mathscr{P}^\epsilon_\gamma} E_n(\mathscr{P}^\epsilon(t)) \left[ \ipy{\mathfrak{a}^\epsilon(y,t)}{(- i \de_{y_\beta})(- i \de_{y_\gamma}) \mathfrak{a}^\epsilon(y,t) } \right. \\
	&\left. - \ipy{\mathfrak{a}^\epsilon(y,t)}{(- i \de_{y_\beta}) \mathfrak{a}^\epsilon(y,t)}\ipy{\mathfrak{a}^\epsilon(y,t)}{(- i \de_{y_\gamma}) \mathfrak{a}^\epsilon(y,t)} \right] \\
	&+ \de_{\mathscr{Q}^\epsilon_\beta} W(\mathscr{Q}^\epsilon(t)) \de_{\mathscr{P}^\epsilon_\beta} \mathcal{A}_{n,\alpha}(\mathscr{P}^\epsilon(t)) + O(\epsilon^{3/2} e^{c t})	\\
	&\dot{\mathscr{P}}^\epsilon_\alpha(t) = - \de_{\mathscr{Q}^\epsilon_\alpha} W(\mathscr{Q}^\epsilon(t)) - \epsilon \frac{1}{2} \de_{\mathscr{Q}^\epsilon_\alpha} \de_{\mathscr{Q}^\epsilon_\beta} \de_{\mathscr{Q}^\epsilon_\gamma} W(\mathscr{Q}^\epsilon(t)) \left[ \ipy{\mathfrak{a}^\epsilon(y,t)}{ y_\beta y_\gamma \mathfrak{a}^\epsilon(y,t) } \right.	\\
	&\left. - \ipy{\mathfrak{a}^\epsilon(y,t)}{y_\beta \mathfrak{a}^\epsilon(y,t)} \ipy{\mathfrak{a}^\epsilon(y,t)}{y_\gamma \mathfrak{a}^\epsilon(y,t)}\right] 	\\
	&- \de_{\mathscr{Q}^\epsilon_\alpha} \de_{\mathscr{Q}^\epsilon_\beta} W(\mathscr{Q}^\epsilon(t)) \mathcal{A}_{n,\beta}(\mathscr{P}^\epsilon(t)) + O(\epsilon^{3/2} e^{c t}).	\\
\end{split}
\end{equation}
Note that, up to error terms, equations (\ref{eq:Q_P_equation_canonicals}) (\ref{eq:mathfrak_a_equation}) constitute a closed system for $\mathscr{Q}^\epsilon(t),\mathscr{P}^\epsilon(t),\mathfrak{a}^\epsilon(y,t)$. 

We now show that this system may be derived from a Hamiltonian. Let:
\begin{equation} \label{eq:mu_lambda}
\begin{split}
	&\mu^\epsilon(t) := \ipy{\mathfrak{a}^\epsilon(y,t)}{y \mathfrak{a}^\epsilon(y,t)}	\\
	&\lambda^\epsilon(t) := \ipy{\mathfrak{a}^\epsilon(y,t)}{ (- i \nabla_y) \mathfrak{a}^\epsilon(y,t)}.
\end{split}
\end{equation}
Then, we may write (\ref{eq:Q_P_equation_canonicals}) as:
\begin{equation} \label{eq:final_Q_P_equation}
\begin{split}
	&\dot{\mathscr{Q}}^\epsilon_\alpha(t) = \de_{\mathscr{P}^\epsilon_\alpha} E_n(\mathscr{P}^\epsilon(t)) + \epsilon \frac{1}{2} \de_{\mathscr{P}^\epsilon_\alpha} \de_{\mathscr{P}_\beta} \de_{\mathscr{P}^\epsilon_\gamma} E_n(\mathscr{P}^\epsilon(t)) \left[ \ipy{\mathfrak{a}^\epsilon(y,t)}{(- i \de_{y_\beta})(- i \de_{y_\gamma}) \mathfrak{a}^\epsilon(y,t) } - \lambda^\epsilon_\beta(t) \lambda^\epsilon_\gamma(t) \right] \\
	&+ \epsilon \de_{\mathscr{Q}^\epsilon_\beta} W(\mathscr{Q}^\epsilon(t)) \de_{\mathscr{P}^\epsilon_\beta} \mathcal{A}_{n,\alpha}(\mathscr{P}^\epsilon(t)) + O(\epsilon^{3/2} e^{c t})	\\
	&\dot{\mathscr{P}}^\epsilon_\alpha(t) = - \de_{\mathscr{Q}^\epsilon_\alpha} W(\mathscr{Q}^\epsilon(t)) - \epsilon \frac{1}{2} \de_{\mathscr{Q}^\epsilon_\alpha} \de_{\mathscr{Q}^\epsilon_\beta} \de_{\mathscr{Q}^\epsilon_\gamma} W(\mathscr{Q}^\epsilon(t))\left[ \ipy{\mathfrak{a}^\epsilon(y,t)}{ y_\beta y_\gamma \mathfrak{a}^\epsilon(y,t) } - \mu^\epsilon_\alpha(t) \mu^\epsilon_\beta(t) \right] 	\\
	&- \epsilon \de_{\mathscr{Q}^\epsilon_\alpha} \de_{\mathscr{Q}^\epsilon_\beta} W(\mathscr{Q}^\epsilon(t)) \mathcal{A}_{n,\beta}(\mathscr{P}^\epsilon(t)) + O(\epsilon^{3/2} e^{c t})	\\
\end{split}
\end{equation}
By an identical calculation to that given in Appendix \ref{app:computation_of_dynamics_of_physical_observables} (\ref{eq:Q_dot_P_dot_a}) (\ref{eq:Q_dot_P_dot_a_corollary}), we have that:
\begin{equation} \label{eq:mu_lambda_equation}
\begin{split}
	&\dot{\mu}^\epsilon_\alpha(t) = \de_{\mathscr{P}^\epsilon_\alpha} \de_{\mathscr{P}^\epsilon_\beta} E_n(\mathscr{P}^\epsilon(t)) \lambda^\epsilon_{\beta}(t)  	\\
	&\dot{\lambda}^\epsilon_\alpha(t) = - \de_{\mathscr{Q}^\epsilon_\alpha} \de_{\mathscr{Q}^\epsilon_\beta} W(\mathscr{Q}^\epsilon(t)) \mu^\epsilon_{\beta}(t).
\end{split}
\end{equation}
Let:
\begin{equation} \label{eq:full_hamiltonian}
\begin{split}
	&\mathcal{H}^\epsilon(\mathscr{Q}^\epsilon,\mathscr{P}^\epsilon,\overline{\mathfrak{a}^\epsilon},\mathfrak{a}^\epsilon,\mu^\epsilon,\lambda^\epsilon) := E_n(\mathscr{P}^\epsilon) + \epsilon W(\mathscr{Q}^\epsilon) + \epsilon \nabla_{\mathcal{Q}^\epsilon} W(\mathscr{Q}^\epsilon) \cdot \mathcal{A}_n(\mathscr{P}^\epsilon)	\\
	&+ \epsilon \frac{1}{2} \de_{\mathscr{P}^\epsilon_\alpha} \de_{\mathscr{P}^\epsilon_{\beta}} E_n(\mathscr{P}^\epsilon)  \left[ \ipy{\de_{y_\alpha} a^\epsilon}{\de_{y_\beta} a^\epsilon} - \lambda^\epsilon_{\alpha} \lambda^\epsilon_{\beta} \right] \\
	&+ \epsilon \frac{1}{2} \de_{\mathscr{Q}^\epsilon_\alpha} \de_{\mathscr{Q}^\epsilon_\beta} W({\mathscr{Q}}^\epsilon) \left[ \ipy{y_\alpha a^\epsilon}{y_\beta a^\epsilon} - \mu_\alpha^\epsilon \mu_\beta^\epsilon \right]
\end{split}
\end{equation}
Then we may write the closed system (\ref{eq:mathfrak_a_equation}), (\ref{eq:mu_lambda_equation}), (\ref{eq:final_Q_P_equation}) as:
\begin{equation} \label{eq:full_hamiltonian_eom}
\begin{split}
	&\dot{\mathscr{Q}}^\epsilon = \nabla_{\mathscr{P}^\epsilon} \mathcal{H}^\epsilon(\mathscr{P}^\epsilon), \dot{\mathscr{P}}^\epsilon = - \nabla_{\mathscr{Q}^\epsilon} \mathcal{H}^\epsilon(\mathscr{Q}^\epsilon),	\\
	&i \de_t \mathfrak{a}^\epsilon = \frac{ \delta \mathcal{H}^\epsilon }{ \delta \overline{\mathfrak{a}^\epsilon} }, \dot{\mu}^\epsilon(t) = - \nabla_{ \lambda^\epsilon } \mathcal{H}^\epsilon, \dot{\lambda}^\epsilon(t) = \nabla_{ \mu^\epsilon } \mathcal{H}^\epsilon.
\end{split}
\end{equation}

The precise statements (1),(2),(3) of Theorem \ref{prop:center_of_mass_proposition} follow from the following observations. The errors in equations (\ref{eq:QP_Rd}), (\ref{eq:full_Q_P_equations_again}), (\ref{eq:final_Q_P_equation}) may each be bounded by $\epsilon^{3/2} C_1 e^{c_1 t}, \epsilon^{3/2} C_2 e^{c_2 t}, \epsilon^{3/2} C_3 e^{c_3 t}$ for positive constants $c_j, C_j, j \in \{1,2,3\}$. Define $c' := \max_{j \in \{1,2,3\}} c_j, C' := \max_{j \in \{1,2,3\}} C_j$. Then all of these errors may be bounded by $\epsilon^{3/2} C' e^{c' t}$. It follows that these terms are $o(\epsilon)$ for all $t \in [0,\tilde{C}' \ln 1/\epsilon]$ where $\tilde{C}'$ is any constant such that $\tilde{C}' < \frac{1}{2 c'}$. Next, in Appendix \ref{app:computation_of_dynamics_of_physical_observables} (\ref{eq:Q_dot_P_dot_a}) (\ref{eq:Q_dot_P_dot_a_corollary}) we show that $\ipy{a_0(y)}{y a_0(y)} = \ipy{a_0(y)}{(- i \nabla_y)a_0(y)} = 0$ implies that for all $t \geq 0$ $\ipy{a(y,t)}{y a(y,t)} = \ipy{a(y,t)}{(- i \nabla_y)a(y,t)} = 0$. Imposing the constraints (\ref{eq:well_prepared_envelope}), then, the simplified expressions (\ref{eq:observables_expanded}) (\ref{eq:eom_for_center_of_mass}) follow from (\ref{eq:QP_Rd}), (\ref{eq:full_Q_P_equations_again}) respectively. We are also justified in ignoring the $\lambda^\epsilon, \mu^\epsilon$ degrees of freedom in (\ref{eq:full_hamiltonian}) (\ref{eq:full_hamiltonian_eom}) since for all $t \geq 0$, $\lambda^\epsilon(t) = \mu^\epsilon(t) = 0$. In this way we obtain the simplified Hamiltonian system (\ref{eq:periodic_hamiltonian_system}) (\ref{eq:periodic_hamiltonian}). 

\appendix
\section{Useful identities involving $E_n, \chi_n$} \label{ch:useful_identities}
Let $E(p), \chi(z;p)$ satisfy the eigenvalue problem:
\begin{equation}
\begin{split}	
	&\left[ H(p) - E(p) \right] \chi(z;p) = 0  	\\
	&H(p) = \frac{1}{2}(p - i\nabla_z)^2 + V(z) 
\end{split}
\end{equation}
and assume that $E(p), \chi(z;p)$ are smooth functions of $p$. Taking the gradient with respect to $p$ gives: 
\begin{equation} \label{eq:first_derivative_identities} 
	\left[  (p - i \nabla_z) - \nabla_p E_n(p) \right] \chi_n(z;p) + \left[  H(p) - E_n(p) \right] \nabla_p \chi_n(z;p) = 0	
\end{equation}
Taking two derivatives with respect to $p$ of the equation gives: 
\begin{equation} \label{eq:kk_identity}
\begin{split}
	&\left[ \delta_{\alpha \beta} - \de_{p_\alpha} \de_{p_\beta} E_n(p) \right] \chi_n(z;p) + \left[  (p - i \de_z)_\alpha - \de_{p_\alpha} E_n(p) \right] \de_{p_\beta} \chi_n(z;p) \\
	&+ \left[  (p - i \de_z)_\beta - \de_{p_\beta} E_n(p) \right] \de_{p_\alpha} \chi_n(z;p) + \left[  H(p) - E_n(p) \right] \de_{p_\alpha} \de_{p_\beta} \chi_n(z;p) = 0	\\
\end{split}
\end{equation}
where $\delta_{\alpha \beta}$ is the Kronecker delta. Taking the derivative with respect to $p_\gamma$ of (\ref{eq:kk_identity}) gives:
\begin{equation} \label{eq:kkk_identity}
\begin{split}
	&[ - \de_{p_\alpha} \de_{p_\beta} \de_{p_\gamma} E_n(p) ] \chi_n(z;p) + [ \delta_{\alpha \beta} - \de_{p_\alpha} \de_{p_\beta} E_n(p) ] \de_{p_\gamma} \chi_n(z;p)	\\
	&+ [ \delta_{\alpha \gamma} - \de_{p_\alpha} \de_{p_\gamma} E_n(p) ] \de_{p_\beta} \chi_n(z;p) + [ \delta_{\beta \gamma} - \de_{p_\beta} \de_{p_\gamma} E_n(p) ] \de_{p_\alpha} \chi_n(z;p)		\\
	&+ [ ( p - i \de_z )_\alpha - \de_{p_\alpha} E_n(p) ] \de_{p_\beta} \de_{p_\gamma} \chi_n(z;p) + [ (p - i \de_z)_\beta - \de_{p_\beta} E_n(p) ] \de_{p_\alpha} \de_{p_\gamma} \chi_n(z;p) \\
	&+ [ (p - i \de_z)_\gamma - \de_{p_\gamma} E_n(p) ] \de_{p_\alpha} \de_{p_\beta} \chi_n(z;p) + [ H(p) - E_n(p) ] \de_{p_\alpha} \de_{p_\beta} \de_{p_\gamma} \chi_n(z;p) = 0
\end{split}
\end{equation}

\section{Derivation of leading-order envelope equation} \label{derivation_of_leading_order_envelope_equation}
Collecting terms of order $\epsilon$ in the expansion (\ref{eq:everything_expanded}), using equations (\ref{eq:action_integral}) for $\dot{S}(t)$ and (\ref{eq:classical_system}) for $\dot{q}(t), \dot{p}(t)$, and setting equal to zero gives the following inhomogeneous self-adjoint elliptic equation in $z$ for $f^2(y,z,t)$:
\begin{equation} \label{eq:equation_for_f_2}
\begin{split}
	&\left[ H(p(t)) - E_n(p(t)) \right] f^2(y,z,t) = \xi^2(y,z,t) 	\\
	&\text{for all } v \in \Lambda, f^2(y,z + v,t) = f^2(y,z,t); \; f^2(y,z,t) \in \Sigma^{R-2}_y(\field{R}^d)	\\
	&\xi^2 := - \left[ \frac{1}{2} (- i \nabla_y)^2 + \frac{1}{2} \de_{q_\alpha}\de_{q_\beta} W(q(t)) y_\alpha y_\beta - i \de_t \right] f^0(y,z,t)	\\
	&- \left[ \left( (p(t) - i \nabla_z) - \nabla_p E_n(p(t)) \right) \cdot (- i \nabla_y) \right] f^1(y,z,t).	
\end{split}
\end{equation} 
We follow the strategy outlined in Remark \ref{rem:general_strategy}. The proof of the following Lemma will be given at the end of this section:
\begin{lemma} \label{lem:xi_2}
$\xi^2(y,z,t)$, defined in (\ref{eq:equation_for_f_2}) satisfies:
\begin{equation} \label{eq:form_of_xi_2}
	\xi^2(y,z,t) = \tilde{\xi}^2(y,z,t) + \left[ H(p(t)) - E_n(p(t)) \right] u^2(y,z,t)
\end{equation}
where:
\begin{equation} \label{eq:xi_and_u_2}
\begin{split}
	&\tilde{\xi}^2(y,z,t) =	\left[ i \de_t a^0(y,t) - \frac{1}{2} \de_{p_\alpha} \de_{p_\beta} E_n(p(t)) (- i \de_{y_\alpha})(- i \de_{y_\beta}) a^0(y,t) - \frac{1}{2} \de_{q_\alpha} \de_{q_\beta} W(q(t)) y_\alpha y_\beta a^0(y,t) \right. \\
	&\left. \vphantom{\frac{1}{2}} - \nabla_q W(q(t)) \cdot \mathcal{A}_n(p(t)) \right] \chi_n(z;p(t)) 	\\
	&+ P^\perp_n(p(t)) \left[ - i a^0(y,t) \nabla_q W(q(t)) \cdot \nabla_p \chi_n(z;p(t)) \right] 	\\
	&u^2(y,z,t) = (- i \nabla_y) a^1(y,t) \cdot \nabla_p \chi_n(z;p(t)) + \frac{1}{2} (- i \de_{y_\alpha})(- i \de_{y_\beta}) a^0(y,t) \de_{p_\alpha} \de_{p_\beta} \chi_n(z;p(t)). 
\end{split}
\end{equation}
Here, $\mathcal{A}_n(p(t))$ is the Berry connection (\ref{eq:berry_connection}) and $P^\perp_n(p(t))$ is the orthogonal projection operator away from the subspace of $L^2_{per}$ spanned by $\chi_n(z;p(t))$ (\ref{eq:def_of_P_perp}). 
\end{lemma}
Imposing the solvability condition of equation (\ref{eq:equation_for_f_2}), given by (\ref{eq:general_solvability_condition}) with $j = 2$ and $\tilde{\xi}^2(y,z,t)$ given by (\ref{eq:xi_and_u_2}), gives the following evolution equation for $a^0(y,t)$:
\begin{equation} \label{eq:equation_for_a}
\begin{split}
	&i \de_t a^0(y,t) = \frac{1}{2} \de_{p_\alpha} \de_{p_\beta} E_n(p(t)) (- i \de_{y_\alpha})(- i \de_{y_\beta}) a^0(y,t) + \frac{1}{2} \de_{q_\alpha} \de_{q_\beta} W(q(t)) y_\alpha y_\beta a^0(y,t) 	\\
 	&+ \nabla_q W(q(t)) \cdot \mathcal{A}_n(p(t)) a^0(y,t) \\
\end{split}
\end{equation}
Taking $a^0(y,t) = a(y,t) e^{i \phi_B(t)}$ and matching with the initial data implies equations (\ref{eq:phase_equation}) and (\ref{eq:envelope_equation}). The general solution of (\ref{eq:equation_for_f_2}) is given by (\ref{eq:general_solution_f_j}) with $j = 2$: 
\begin{equation} \label{eq:second_order_solution}
\begin{split}
	&f^2(y,z,t) = a^2(y,t) \chi_n(z;p(t)) + (- i \nabla_y) a^1(y,t) \cdot \nabla_p \chi_n(z;p(t)) 	\\
	&+ \frac{1}{2} (- i \de_{y_\alpha})(- i \de_{y_\beta}) a^0(y,t) \de_{p_\alpha} \de_{p_\beta} \chi_n(z;p(t))	\\ 
	&[H(p(t)) - E_n(p(t))]^{-1} P^\perp_n(p(t)) \left[ - i \nabla_q W(q(t)) a^0(y,t) \cdot \nabla_p \chi_n(z;p(t)) \right] 
\end{split}
\end{equation}
where $a^2(y,t)$ is an arbitrary function in $\Sigma^{R-2}_y(\field{R}^d)$ to be fixed at higher order in the expansion. 
\begin{proof}[Proof of Lemma \ref{lem:xi_2}]
Adding and substracting terms, using smoothness of the band $E_n(p)$ in a neighborhood of $p(t)$ (Assumption \ref{isolated_band_assumption_nonseparable_case}) we can re-write $\xi^2$ (\ref{eq:equation_for_f_2}) as:
\begin{equation}
\begin{split}
	&\xi^2(y,z,t) = - \left[ \frac{1}{2} \de_{p_\alpha} \de_{p_\beta} E_n(p(t)) (- i \de_{y_\alpha})(- i \de_{y_\beta}) + \frac{1}{2} \de_{q_\alpha} \de_{q_\beta} W(q(t)) y_\alpha y_\beta - i \de_t \right] f^0(y,z,t) \\
	&- \left[ \frac{1}{2} \left( \delta_{\alpha \beta} - \de_{p_\alpha} \de_{p_\beta} E_n(p(t)) \right) (- i \de_{y_\alpha})(- i \de_{y_\beta}) \right] f^0(y,z,t)	\\ 
	&- \left[ \left( (p(t) - i \nabla_z) - \nabla_p E_n(p(t)) \right) \cdot (- i \nabla_y) \right] f^1(y,z,t).
\end{split}
\end{equation}
Substituting the forms of $f^0(y,z,t)$ (\ref{eq:leading_order_solution}) and $f^1(y,z,t)$ (\ref{eq:first_order_solution}) gives: 
\begin{equation} \label{eq:xi_2_to_simplify}
\begin{split}
	&\xi^2(y,z,t) = - \left[ \frac{1}{2} \de_{p_\alpha} \de_{p_\beta} E_n(p(t)) (- i \de_{y_\alpha})(- i \de_{y_\beta}) a^0(y,t) + \frac{1}{2} \de_{q_\alpha} \de_{q_\beta} W(q(t)) y_\alpha y_\beta a^0(y,t) \right. \\
	&\left. \vphantom{\frac{1}{2}} - i \de_t a^0(y,t) \right] \chi_n(z;p(t)) + i a^0(y,t) \dot{p}(t) \cdot \nabla_p \chi_n(z;p(t))  	\\
	&- (- i \de_{y_\alpha})(- i \de_{y_\beta}) a^0(y,t) \frac{1}{2} \left( \delta_{\alpha \beta} - \de_{p_\alpha} \de_{p_\beta} E_n(p(t)) \right) \chi_n(z;p(t))   \\
	&- (- i \de_{y_\alpha})(- i \de_{y_\beta}) a^0(y,t) \left( (p(t) - i \de_z)_\alpha - \de_{p_\alpha} E_n(p(t)) \right) \de_{p_\beta} \chi_n(z;p(t))  \\
	&- (- i \nabla_y) a^1(y,t) \cdot \left( (p(t) - i \nabla_z) - \nabla_p E_n(p(t)) \right) \chi_n(z;p(t)).
\end{split}
\end{equation}
Using (\ref{eq:first_derivative_identities}) we can simplify the term involving $a^1$:
\begin{equation} \label{eq:terms_involving_a_1}
\begin{split}
	&- (- i \nabla_y) a^1(y,t) \cdot \left( (p(t) - i \nabla_z) - \nabla_p E_n(p(t)) \right) \chi_n(z;p(t)) \\
	&= (- i \nabla_y) a^1(y,t) \cdot \left[ H(p(t)) - E_n(p(t)) \right] \nabla_p \chi_n(z;p(t)).
\end{split}
\end{equation}
Using (\ref{eq:kk_identity}), and the symmetry: $(- i \de_{y_\alpha})(- i \de_{y_\beta}) a^0(y,t) = (- i \de_{y_\beta})(- i \de_{y_\alpha}) a^0(y,t)$ we can simplify the terms:
\begin{equation} \label{eq:kk_terms}
\begin{split}
	&- (- i \de_{y_\alpha})(- i \de_{y_\beta}) a^0(y,t) \frac{1}{2} \left( \delta_{\alpha \beta} - \de_{p_\alpha} \de_{p_\beta} E_n(p(t)) \right) \chi_n(z;p(t))   \\
	&- (- i \de_{y_\alpha})(- i \de_{y_\beta}) a^0(y,t) \left( (p(t) - i \de_z)_\alpha - \de_{p_\alpha} E_n(p(t)) \right) \de_{p_\beta} \chi_n(z;p(t))  \\
	&= \frac{1}{2} (- i \de_{y_\alpha})(- i \de_{y_\beta}) a^0(y,t) \left[ H(p(t)) - E_n(p(t)) \right] \de_{p_\alpha} \de_{p_\beta} \chi_n(z;p(t)).
\end{split}
\end{equation}
Recall that $\dot{p}(t) = - \nabla_q W(q(t))$ (\ref{eq:classical_system}). Substituting this, (\ref{eq:terms_involving_a_1}), and (\ref{eq:kk_terms}) into (\ref{eq:equation_for_f_2}) and adding and subtracting $a^0(y,t) \nabla_q W(q(t)) \cdot \mathcal{A}_n(p(t)) \chi_n(z;p(t))$ gives:
\begin{equation} \label{eq:xi_2_simplified}
\begin{split}
	&\xi^2(y,z,t) = \left[ i \de_t a^0(y,t) - \frac{1}{2} \de_{p_\alpha} \de_{p_\beta} E_n(p(t)) (- i \de_{y_\alpha})(- i \de_{y_\beta}) a^0(y,t) - \frac{1}{2} \de_{q_\alpha} \de_{q_\beta} W(q(t)) y_\alpha y_\beta a^0(y,t) \right. \\
	&\left. \vphantom{\frac{1}{2}} - \nabla_q W(q(t)) \cdot \mathcal{A}_n(p(t)) \right] \chi_n(z;p(t)) + P^\perp_n(p(t)) \left[ - i a^0(y,t) \nabla_q W(q(t)) \cdot \nabla_p \chi_n(z;p(t)) \right] 	\\
	&+ \left[ H(p(t)) - E_n(p(t)) \right] \left[ \frac{1}{2} (- i \de_{y_\alpha})(- i \de_{y_\beta}) a^0(y,t) \de_{p_\alpha} \de_{p_\beta} \chi_n(z;p(t)) \right.	\\
	&\left.\vphantom{\frac{1}{2}}  + (- i \nabla_y) a^1(y,t) \cdot \nabla_p \chi_n(z;p(t)) \right]
\end{split}
\end{equation}
where $\mathcal{A}_n(p(t))$ is the Berry connection (\ref{eq:berry_connection}) and $P^\perp_n(p(t))$ is the orthogonal projection in $L^2_{per}$ away from the subspace spanned by $\chi_n(z;p(t))$. 
\end{proof}

\section{Derivation of first-order envelope equation} \label{derivation_of_first_order_envelope_equation}
Collecting terms of order $\epsilon^{3/2}$ in the expansion (\ref{eq:everything_expanded}), using equations (\ref{eq:action_integral}) for $\dot{S}(t)$ and (\ref{eq:classical_system}) for $\dot{q}(t), \dot{p}(t)$, and setting equal to zero gives the following inhomogeneous self-adjoint elliptic equation in $z$ for $f^3(y,z,t)$:
\begin{equation} \label{eq:equation_for_f_3}
\begin{split}
	&\left[ \vphantom{\frac{1}{2}} H(p(t)) - E_n(p(t)) \right] f^3(y,z,t) = \xi^3(y,z,t)	\\
	&\text{for all } v \in \Lambda, f^3(y,z + v,t) = f^3(y,z,t); \; f^3(y,z,t) \in \Sigma^{R-3}_y(\field{R}^d)	\\
	&\xi^3(y,z,t) := - \left[ \frac{1}{6} \de_{q_\alpha} \de_{q_\beta} \de_{q_\gamma} W(q(t)) y_\alpha y_\beta y_\gamma \right] f^0(y,z,t) \\
	&- \left[ \frac{1}{2} (- i \nabla_{y})^2 + \frac{1}{2} \de_{q_\alpha} \de_{q_\beta} W(q(t)) y_\alpha y_\beta - i \de_t \right] f^1(y,z,t) 	\\
	&- \left[ \vphantom{\frac{1}{2}} [ (p(t) - i \nabla_z) - \nabla_p E_n(p(t)) ] \cdot ( - i \nabla_{y}) \right] f^2(y,z,t) \\
\end{split}
\end{equation}
We claim the following lemmas, the proofs of which will be given at the end of this section:
\begin{lemma} \label{lem:xi_3}
$\xi^3(y,z,t)$, as defined in (\ref{eq:equation_for_f_3}), satisfies:
\begin{equation} \label{eq:form_of_xi_3}
	\xi^3(y,z,t) = \tilde{\xi}^3(y,z,t) + \left[ H(p(t)) - E_n(p(t)) \right] u^3(y,z,t)	
\end{equation} 
where $\tilde{\xi}^3$ is given explicitly by (\ref{eq:xi_3}) and:
\begin{equation} \label{eq:u_3}
\begin{split}
	&u^3(y,z,t) := (- i \nabla_y) a^2(y,t) \cdot \nabla_p \chi_n(z;p(t)) + \frac{1}{2} (- i \de_{y_\alpha})(- i \de_{y_\beta}) a^1(y,t) \de_{p_\alpha} \de_{p_\beta} \chi_n(z;p(t))  \\
	&+ \frac{1}{6} (- i \de_{y_\alpha})(- i \de_{y_\beta})(- i \de_{y_\gamma}) a^0(y,t) \de_{p_\alpha} \de_{p_\beta} \de_{p_\gamma} \chi_n(z;p(t)).
\end{split}
\end{equation}
\end{lemma}
\begin{lemma} \label{lem:a_1}
The solvability condition for (\ref{eq:equation_for_f_3}), given by (\ref{eq:general_solvability_condition}) with $j = 3$ and $\tilde{\xi}^3(y,z,t)$ given by (\ref{eq:xi_3}), is equivalent to the following evolution equation for $a^1(y,t)$:
\begin{equation} \label{eq:equation_for_a_1}
\begin{split}
	&i \de_t a^1(y,t) = \frac{1}{2} \de_{p_\alpha} \de_{p_\beta} E_n(p(t)) (- i \de_{y_\alpha})(- i \de_{y_\beta}) a^1(y,t) + \frac{1}{2} \de_{q_\alpha} \de_{q_\beta} W(q(t)) y_\alpha y_\beta a^1(y,t)  	\\
	&+ \nabla_q W(q(t)) \cdot \mathcal{A}_n(p(t)) a^1(y,t)	\\
	&+ \frac{1}{6} \de_{p_\alpha} \de_{p_\beta} \de_{p_\gamma} E_n(p(t)) (- i \de_{y_\alpha})(- i \de_{y_\beta})(- i \de_{y_\gamma}) a^0(y,t) + \frac{1}{6} \de_{q_\alpha} \de_{q_\beta} \de_{q_\gamma} W(q(t)) y_\alpha y_\beta y_\gamma a^0(y,t) 	\\
	&+ \de_{q_\beta} W(q(t)) \de_{p_\gamma} \mathcal{A}_{n,\beta}(p(t)) (- i \de_{y_\gamma}) a^0(y,t) + \de_{q_\beta} \de_{q_\gamma} W(q(t)) \mathcal{A}_{n,\beta}(p(t)) y_\gamma a^0(y,t).
\end{split}
\end{equation}
Here, $\mathcal{A}_n(p(t))$ is the Berry connection (\ref{eq:berry_connection}). 
\end{lemma}
Taking $a^1(y,t) = b(y,t) e^{i \phi_B(t)}$ and matching with the initial data implies equation (\ref{eq:first_order_envelope_equation}) for $b(y,t)$. The solution of (\ref{eq:equation_for_f_3}) is then given by (\ref{eq:general_solution_f_j}):
\begin{equation} \label{eq:f_3}
	f^3(y,z,t) = a^3(y,t) \chi_n(z;p(t)) + u^3(y,z,t) + \left[ H(p(t)) - E_n(p(t)) \right]^{-1} P^\perp_n(p(t)) \tilde{\xi}^3(y,z,t) 
\end{equation}
where $\tilde{\xi}^3(y,z,t)$ is given by (\ref{eq:xi_3}) and $u^3(y,z,t)$ by (\ref{eq:u_3}). $a^3(y,t)$ is an arbitrary function in $\Sigma^{R-3}_y(\field{R}^d)$ to be fixed at higher order in the expansion. Note that all manipulations so far are valid as long as $R \geq 3$.  

\begin{proof}[Proof of Lemma \ref{lem:xi_3}]
Adding and subtracting terms using smoothness of the band $E_n(p)$ in a neighborhood of $p(t)$ (Assumption \ref{isolated_band_assumption_nonseparable_case}) we can re-write $\xi^3(y,z,t)$ (\ref{eq:equation_for_f_3}) as:
\begin{equation} \label{eq:xi_3_to_simplify}
\begin{split}
	&\xi^3(y,z,t) = - \left[ \frac{1}{6} \de_{p_\alpha} \de_{p_\beta} \de_{p_\gamma} E_n(p(t)) (- i \de_{y_\alpha})(- i \de_{y_\beta})(- i \de_{y_\gamma}) + \frac{1}{6} \de_{q_\alpha} \de_{q_\beta} \de_{q_\gamma} W(q(t)) y_\alpha y_\beta y_\gamma \right] f^0(y,z,t)  \\
	&- \left[ - \frac{1}{6} \de_{p_\alpha} \de_{p_\beta} \de_{p_\gamma} E_n(p(t)) (- i \de_{y_\alpha})(- i \de_{y_\beta})(- i \de_{y_\gamma}) \right] f^0(y,z,t) \\
	&- \left[ \frac{1}{2} \de_{p_\alpha} \de_{p_\beta} E_n(p(t)) (- i \de_{y_\alpha})(- i \de_{y_\beta}) + \frac{1}{2} \de_{q_\alpha} \de_{q_\beta} W(q(t)) y_\alpha y_\beta - i \de_t \right] f^1(y,z,t) \\
	&- \left[ \frac{1}{2} \left( \delta_{\alpha \beta} - \de_{p_\alpha} \de_{p_\beta} E_n(p(t)) \right) (- i \de_{y_\alpha})(- i \de_{y_\beta}) \right] f^1(y,z,t)  \\
	&- \left[ \vphantom{\frac{1}{2}} \left( (p(t) - i \nabla_z) - \nabla_p E_n(p(t)) \right) \cdot (- i \nabla_y) \right] f^2(y,z,t) \\
\end{split}
\end{equation}
Substituting the forms of $f^0(y,z,t)$ (\ref{eq:leading_order_solution}), $f^1(y,z,t)$ (\ref{eq:first_order_solution}), and $f^2(y,z,t)$ (\ref{eq:second_order_solution}) gives a very long expression on the right-hand side. We simplify this expression by treating terms which depend on $a^2(y,t), a^1(y,t), a^0(y,t)$ in turn. 
\newline \underline{Contributions to (\ref{eq:xi_3_to_simplify}) depending on $a^2(y,t)$.} There is one term which depends on $a^2(y,t)$: 
\begin{equation}
	- (- i \nabla_y) a^2(y,t) \cdot \left( (p(t) - i \nabla_z) - \nabla_p E_n(p(t)) \right) \chi_n(z;p(t)) \\
\end{equation}
which can be simplified using (\ref{eq:first_derivative_identities}):
\begin{equation} \label{eq:a_2_term}
\begin{split}
	- (- i \nabla_y) a^2(y,t) \cdot \left( (p(t) - i \nabla_z) - \nabla_p E_n(p(t)) \right) \chi_n(z;p(t)) \\
	= \left[ H(p(t)) - E_n(p(t)) \right] \left[ (- i \nabla_y) a^2(y,t) \cdot \nabla_p \chi_n(z;p(t)) \right].
\end{split}
\end{equation}
\underline{Contributions to (\ref{eq:xi_3_to_simplify}) depending on $a^1(y,t)$.} The terms which depend on $a^1(y,t)$ are as follows:
\begin{equation}
\begin{split}
	&- \left[ \frac{1}{2} \de_{p_\alpha} \de_{p_\beta} E_n(p(t)) (- i \de_{y_\alpha})(- i \de_{y_\beta}) a^1(y,t) + \frac{1}{2} \de_{q_\alpha} \de_{q_\beta} W(q(t)) y_\alpha y_\beta a^1(y,t) \right. \\
	&\left. \vphantom{\frac{1}{2}} - i \de_t a^1(y,t) \right] \chi_n(z;p(t)) + i \dot{p}(t) \cdot \nabla_p \chi_n(z;p(t)) a^1(y,t)   	\\
	&- (- i \de_{y_\alpha})(- i \de_{y_\beta}) a^1(y,t) \frac{1}{2} \left( \delta_{\alpha \beta} - \de_{p_\alpha} \de_{p_\beta} E_n(p(t)) \right) \chi_n(z;p(t))   \\
	&- (- i \de_{y_\alpha})(- i \de_{y_\beta}) a^1(y,t) \left( (p(t) - i \de_z)_\alpha - \de_p E_n(p(t))_\alpha \right) \de_{p_\beta} \chi_n(z;p(t)).
\end{split}
\end{equation} 
Note that these terms have an identical form to the terms depending on $a^0(y,t)$ in expression (\ref{eq:xi_2_to_simplify}) for $\xi^2(y,z,t)$ which were simplified to the form (\ref{eq:xi_2_simplified}). We may therefore manipulate these terms in an identical way (specifically, using (\ref{eq:classical_system}), (\ref{eq:kk_identity})) into the form: 
\begin{equation} \label{eq:a_1_terms}
\begin{split}
	&= \left[ i \de_t a^1(y,t) - \frac{1}{2} \de_{p_\alpha} \de_{p_\beta} E_n(p(t)) (- i \de_{y_\alpha})(- i \de_{y_\beta}) a^1(y,t) - \frac{1}{2} \de_{q_\alpha} \de_{q_\beta} W(q(t)) y_\alpha y_\beta a^1(y,t) \right. \\
	&\left. \vphantom{\frac{1}{2}} - \nabla_q W(q(t)) \cdot \mathcal{A}_n(p(t)) \right] \chi_n(z;p(t)) + P^\perp_n(p(t)) \left[ - i a^1(y,t) \nabla_q W(q(t)) \cdot \nabla_p \chi_n(z;p(t)) \right] 	\\
	&+ \left[ H(p(t)) - E_n(p(t)) \right] \left[ \frac{1}{2} (- i \de_{y_\alpha})(- i \de_{y_\beta}) a^1(y,t) \de_{p_\alpha} \de_{p_\beta} \chi_n(z;p(t)) \right]	
\end{split}
\end{equation}
\underline{Contributions to (\ref{eq:xi_3_to_simplify}) depending on $a^0(y,t)$.} The terms which depend on $a^0(y,t)$ may be written as ${T}_1 + {T}_2 + {T}_3 + T_4$ where: 
\begin{equation} \label{eq:T_1}
\begin{split}
	&{T}_1 := - \left[ \frac{1}{6} \de_{p_\alpha} \de_{p_\beta} \de_{p_\gamma} E_n(p(t)) (- i \de_{y_\alpha})(- i \de_{y_\beta})(- i \de_{y_\gamma}) a^0(y,t) \right. \\
	&\left. + \frac{1}{6} \de_{q_\alpha} \de_{q_\beta} \de_{q_\gamma} W(q(t)) y_\alpha y_\beta y_\gamma a^0(y,t) \right] \chi_n(z;p(t))  
\end{split}
\end{equation}
\begin{equation}
\begin{split}
	&{T}_2 := - (- i \de_{y_\alpha})(- i \de_{y_\beta})(- i \de_{y_\gamma}) a^0(y,t) \left[ - \frac{1}{6} \de_{p_\alpha} \de_{p_\beta} \de_{p_\gamma} E_n(p(t)) \chi_n(z;p(t)) \right]	\\
	&- (- i \de_{y_\alpha})(- i \de_{y_\beta})(- i \de_{y_\gamma}) a^0(y,t) \left[ \frac{1}{2} \left( \delta_{\alpha \beta} - \de_{p_\alpha} \de_{p_\beta} E_n(p(t)) \right) \de_{p_\gamma} \chi_n(z;p(t))  \right]	\\
	&- (- i \de_{y_\alpha})(- i \de_{y_\beta})(- i \de_{y_\gamma}) a^0(y,t) \left[ \frac{1}{2} \left( (p(t) - i \de_z)_{\alpha} - \de_{p_\alpha} E_n(p(t)) \right) \de_{p_\beta} \de_{p_\gamma} \chi_n(z;p(t)) \right]		
\end{split}
\end{equation}
\begin{equation}
	{T}_3 := - \left[ \frac{1}{2} \de_{p_\alpha} \de_{p_\beta} E_n(p(t)) (- i \de_{y_\alpha})(- i \de_{y_\beta}) + \frac{1}{2} \de_{q_\alpha} \de_{q_\beta} W(q(t)) y_\alpha y_\beta - i \de_t \right] (- i \de_{y_\gamma}) a^0(y,t) \de_{p_\gamma} \chi_n(z;p(t)) 	
\end{equation}
\begin{equation} \label{eq:T_4}
	T_4 := i \left( (p(t) - i \de_z)_\beta - \de_{p_\beta} E_n(p(t)) \right) \de_{q_\gamma} W(q(t)) (- i \de_{y_\beta}) a^0(y,t) [H(p(t)) - E_n(p(t))]^{-1} P^\perp_n(p(t)) \de_{p_\gamma} \chi_n(z;p(t)) 
\end{equation}
Using (\ref{eq:kkk_identity}) and the equality of mixed partial derivatives, we can simplify $T_2$: 
\begin{equation} \label{eq:T_2}
	T_2 = \left[ H(p(t)) - E_n(p(t)) \right] \left[ \frac{1}{6} (- i \de_{y_\alpha})(- i \de_{y_\beta})(- i \de_{y_\gamma}) a^0(y,t) \de_{p_\alpha} \de_{p_\beta} \de_{p_\gamma} \chi_n(z;p(t)) \right].
\end{equation}
We can simplify $T_3$ using the evolution equation for $a^0(y,t)$ (\ref{eq:equation_for_a}): 
\begin{equation} \label{eq:terms_simplified_a_bit}
\begin{split}
	&T_3 = - \left[ \frac{1}{2} \de_{p_\alpha} \de_{p_\beta} E_n(p(t)) (- i \de_{y_\alpha})(- i \de_{y_\beta}) + \frac{1}{2} \de_{q_\alpha} \de_{q_\beta} W(q(t)) y_\alpha y_\beta \right] (- i \de_{y_\gamma}) a^0(y,t) \de_{p_\gamma} \chi_n(z;p(t)) 		\\
	&+ (- i \de_{y_\gamma}) \left[ i \de_t a^0(y,t) \right] \de_{p_\gamma} \chi_n(z;p(t)) + (- i \de_{y_\gamma}) a^0(y,t) i \dot{p}_\beta(t) \de_{p_\beta} \de_{p_\gamma} \chi_n(z;p(t))	\\
	&= - \frac{1}{2} \de_{q_\alpha} \de_{q_\beta} W(q(t)) y_\alpha y_\beta (- i \de_{y_\gamma}) a^0(y,t) \de_{p_\gamma} \chi_n(z;p(t)) + \frac{1}{2} \de_{q_\alpha} \de_{q_\beta} W(q(t)) (- i \de_{y_\gamma}) y_\alpha y_\beta a^0(y,t) \de_{p_\gamma} \chi_n(z;p(t)) 	\\
	&+ \de_{q_\beta} W(q(t)) \mathcal{A}_{n,\beta}(p(t)) (- i \de_{y_\gamma}) a^0(y,t) \de_{p_\gamma} \chi_n(z;p(t)) - i \de_{q_\beta} W(q(t)) (- i \de_{y_\gamma}) a^0(y,t) \de_{p_\beta} \de_{p_\gamma} \chi_n(z;p(t)).
\end{split}
\end{equation}
We now write $T_3 = T_{3,1} + T_{3,2}$ where:
\begin{equation} \label{eq:T_3_1_to_simplify}
	T_{3,1} := - \frac{1}{2} \de_{q_\alpha} \de_{q_\beta} W(q(t)) y_\alpha y_\beta (- i \de_{y_\gamma}) a^0(y,t) \de_{p_\gamma} \chi_n(z;p(t)) + \frac{1}{2} \de_{q_\alpha} \de_{q_\beta} W(q(t)) (- i \de_{y_\gamma}) y_\alpha y_\beta a^0(y,t) \de_{p_\gamma} \chi_n(z;p(t)) 	
\end{equation}
\begin{equation} \label{eq:T_3_2}
	T_{3,2} := \de_{q_\beta} W(q(t)) \mathcal{A}_{n,\beta}(p(t)) (- i \de_{y_\gamma}) a^0(y,t) \de_{p_\gamma} \chi_n(z;p(t)) - i \de_{q_\beta} W(q(t)) (- i \de_{y_\gamma}) a^0(y,t) \de_{p_\beta} \de_{p_\gamma} \chi_n(z;p(t)).
\end{equation}
We can simplify $T_{3,1}$ as follows. We first re-arrange (\ref{eq:T_3_1_to_simplify}): 
\begin{equation}
	T_{3,1} = \left( (- i \de_{y_\gamma}) y_\alpha y_\beta - y_\alpha y_\beta (- i \de_{y_\gamma})  \right) \frac{1}{2} \de_{q_\alpha} \de_{q_\beta} W(q(t)) a^0(y,t) \de_{p_\gamma} \chi_n(z;p(t)).
\end{equation}
Using the identity: $(- i \de_{y_\alpha}) y_\beta - y_\beta (- i \de_{y_\alpha}) = - i \delta_{\alpha \beta}$ twice we have that:
\begin{equation}
	(- i \de_{y_\gamma}) y_\alpha y_\beta - y_\alpha y_\beta (- i \de_{y_\gamma}) = (- i \delta_{\alpha \gamma}) y_\beta + (- i \delta_{\beta \gamma}) y_\alpha. 
\end{equation}
Using the symmetry $\de_{q_{\alpha}} \de_{q_\beta} W(q(t)) = \de_{q_\beta} \de_{q_\alpha} W(q(t))$ we have that:
\begin{equation} \label{eq:T_3_1}
	T_{3,1} = - i \de_{q_\alpha} \de_{q_\beta} W(q(t)) y_\alpha a^0(y,t) \de_{p_\beta} \chi_n(z;p(t)). 
\end{equation}
Summing $T_1 + T_2 + T_{3,1} + T_{3,2} + T_4$ (\ref{eq:T_1}) (\ref{eq:T_2}) (\ref{eq:T_3_1}) (\ref{eq:T_3_2}) (\ref{eq:T_4}), we have that the terms which depend on $a^0(y,t)$ in (\ref{eq:xi_3_to_simplify}) are equal to:
\begin{equation} \label{eq:a_0_terms}
\begin{split}
	&- \left[ \frac{1}{6} \de_{p_\alpha} \de_{p_\beta} \de_{p_\gamma} E_n(p(t)) (- i \de_{y_\alpha})(- i \de_{y_\beta})(- i \de_{y_\gamma}) a^0(y,t) + \frac{1}{6} \de_{q_\alpha} \de_{q_\beta} \de_{q_\gamma} W(q(t)) y_\alpha y_\beta y_\gamma a^0(y,t) \right] \chi_n(z;p(t))  \\
	&+ \left[ H(p(t)) - E_n(p(t)) \right] \left[\frac{1}{6} (- i \de_{y_\alpha})(- i \de_{y_\beta})(- i \de_{y_\gamma}) a^0(y,t) \de_{p_\alpha} \de_{p_\beta} \de_{p_\gamma} \chi_n(z;p(t)) \right]	\\
	&- i \de_{q_\alpha} \de_{q_\beta} W(q(t)) y_\alpha a^0(y,t) \de_{p_\beta} \chi_n(z;p(t))	\\
	&+ \de_{q_\beta} W(q(t)) \mathcal{A}_{n,\beta}(p(t)) (- i \de_{y_\gamma}) a^0(y,t) \de_{p_\gamma} \chi_n(z;p(t)) - i \de_{q_\beta} W(q(t)) (- i \de_{y_\gamma}) a^0(y,t) \de_{p_\beta} \de_{p_\gamma} \chi_n(z;p(t))	\\
	&+ i \left( (p(t) - i \de_z)_\beta - \de_{p_\beta} E_n(p(t)) \right) \de_{q_\gamma} W(q(t)) (- i \de_{y_\beta}) a^0(y,t) [H(p(t)) - E_n(p(t))]^{-1} P^\perp_n(p(t)) \de_{p_\gamma} \chi_n(z;p(t)).
\end{split}
\end{equation}
By adding and subtracting terms and using the definition of $\mathcal{A}_n(p(t))$ (\ref{eq:berry_connection}) we can put (\ref{eq:a_0_terms}) into the form:
\begin{equation}
\begin{split}
	&\left[  i \de_{q_\beta} W(q(t)) \left( \vphantom{\frac{1}{2}} \ipz{\chi_n(z;p(t))}{\de_{p_\beta} \chi_n(z;p(t))} \ipz{\chi_n(z;p(t))}{\de_{p_\gamma} \chi_n(z;p(t))} \right. \right.	\\
	&+ \ipz{\chi_n(z;p(t))}{\left( (p(t) - i \de_z)_\gamma- \de_{p_\gamma} E_n(p(t)) \right) [H(p(t)) - E_n(p(t))]^{-1} P^\perp_n(p(t)) \de_{p_\beta} \chi_n(z;p(t))}  	\\
	&\left. \vphantom{\frac{1}{2}} - \ipz{\chi_n(z;p(t))}{\de_{p_\beta} \de_{p_\gamma} \chi_n(z;p(t))} \right)(- i \de_{y_\gamma}) a^0(y,t) - \de_{q_\alpha} \de_{q_\beta} W(q(t)) \mathcal{A}_{n,\beta}(p(t)) y_\alpha a^0(y,t)	\\
	&\left. - \frac{1}{6} \de_{p_\alpha} \de_{p_\beta} \de_{p_\gamma} E_n(p(t)) (- i \de_{y_\alpha})(- i \de_{y_\beta})(- i \de_{y_\gamma}) a^0(y,t) - \frac{1}{6} \de_{q_\alpha} \de_{q_\beta} \de_{q_\gamma} W(q(t)) y_\alpha y_\beta y_\gamma a^0(y,t) \right] \chi_n(z;p(t))   			\\
	&+ P^\perp_n(p(t)) \left[ \vphantom{\frac{1}{2}} - i \de_{q_\alpha} \de_{q_\beta} W(q(t)) y_\alpha a^0(y,t) \de_{p_\beta} \chi_n(z;p(t)) \right.	\\
	&+ \de_{q_\beta} W(q(t)) \mathcal{A}_{n,\beta}(p(t)) (- i \de_{y_\gamma}) a^0(y,t) \de_{p_\gamma} \chi_n(z;p(t)) - i \de_{q_\beta} W(q(t)) (- i \de_{y_\gamma}) a^0(y,t) \de_{p_\beta} \de_{p_\gamma} \chi_n(z;p(t))	\\
	&\left. \vphantom{\frac{1}{2}} + i \left( (p(t) - i \de_z)_\beta - \de_{p_\beta} E_n(p(t)) \right) \de_{q_\gamma} W(q(t)) (- i \de_{y_\beta}) a^0(y,t) [H(p(t)) - E_n(p(t))]^{-1} P^\perp_n(p(t)) \de_{p_\gamma} \chi_n(z;p(t)) \right]	\\
	&+ \left[ H(p(t)) - E_n(p(t)) \right] \left[\frac{1}{6} (- i \de_{y_\alpha})(- i \de_{y_\beta})(- i \de_{y_\gamma}) a^0(y,t) \de_{p_\alpha} \de_{p_\beta} \de_{p_\gamma} \chi_n(z;p(t)) \right]	
\end{split}
\end{equation}
Adding (\ref{eq:a_2_term}), (\ref{eq:a_1_terms}) and (\ref{eq:a_0_terms}) we have that $\xi^3(y,z,t)$ can be decomposed as in (\ref{eq:form_of_xi_3}) where $u^3(y,z,t)$ is given by (\ref{eq:u_3}) and $\tilde{\xi}^3(y,z,t)$ is equal to:
\begin{equation} \label{eq:xi_3}
\begin{split}
	&\tilde{\xi}^3(y,z,t) = \left[ i \de_t a^1(y,t) - \frac{1}{2} \de_{p_\alpha} \de_{p_\beta} E_n(p(t)) (- i \de_{y_\alpha})(- i \de_{y_\beta}) a^1(y,t) - \frac{1}{2} \de_{q_\alpha} \de_{q_\beta} W(q(t)) y_\alpha y_\beta a^1(y,t) \right. \\
	&\vphantom{\frac{1}{2}} - \nabla_q W(q(t)) \cdot \mathcal{A}_n(p(t)) + i \de_{q_\beta} W(q(t)) \left( \vphantom{\frac{1}{2}} \ipz{\chi_n(z;p(t))}{\de_{p_\beta} \chi_n(z;p(t))} \ipz{\chi_n(z;p(t))}{\de_{p_\gamma} \chi_n(z;p(t))} \right.	\\
	&+ \ipz{\chi_n(z;p(t))}{\left( (p(t) - i \de_z)_\gamma- \de_{p_\gamma} E_n(p(t)) \right) [H(p(t)) - E_n(p(t))]^{-1} P^\perp_n(p(t)) \de_{p_\beta} \chi_n(z;p(t))}  	\\
	&\left. \vphantom{\frac{1}{2}} - \ipz{\chi_n(z;p(t))}{\de_{p_\beta} \de_{p_\gamma} \chi_n(z;p(t))} \right)(- i \de_{y_\gamma}) a^0(y,t) - \de_{q_\alpha} \de_{q_\beta} W(q(t)) \mathcal{A}_{n,\beta}(p(t)) y_\alpha a^0(y,t)	\\
	&\left. - \frac{1}{6} \de_{p_\alpha} \de_{p_\beta} \de_{p_\gamma} E_n(p(t)) (- i \de_{y_\alpha})(- i \de_{y_\beta})(- i \de_{y_\gamma}) a^0(y,t) - \frac{1}{6} \de_{q_\alpha} \de_{q_\beta} \de_{q_\gamma} W(q(t)) y_\alpha y_\beta y_\gamma a^0(y,t) \right] \chi_n(z;p(t))   			\\
	&+ P^\perp_n(p(t)) \left[ \vphantom{\frac{1}{2}} - i a^1(y,t) \nabla_q W(q(t)) \cdot \nabla_p \chi_n(z;p(t)) - i \de_{q_\alpha} \de_{q_\beta} W(q(t)) y_\alpha a^0(y,t) \de_{p_\beta} \chi_n(z;p(t)) \right.	\\
	&+ \de_{q_\beta} W(q(t)) \mathcal{A}_{n,\beta}(p(t)) (- i \de_{y_\gamma}) a^0(y,t) \de_{p_\gamma} \chi_n(z;p(t)) - i \de_{q_\beta} W(q(t)) (- i \de_{y_\gamma}) a^0(y,t) \de_{p_\beta} \de_{p_\gamma} \chi_n(z;p(t))	\\
	&\left. \vphantom{\frac{1}{2}} + i \left( (p(t) - i \de_z)_\beta - \de_{p_\beta} E_n(p(t)) \right) \de_{q_\gamma} W(q(t)) (- i \de_{y_\beta}) a^0(y,t) [H(p(t)) - E_n(p(t))]^{-1} P^\perp_n(p(t)) \de_{p_\gamma} \chi_n(z;p(t)) \right]	\\
\end{split}
\end{equation}
\end{proof}
\begin{proof}[Proof of Lemma \ref{lem:a_1}]
Imposing the orthogonality condition (\ref{eq:general_solvability_condition}) with $j = 3$ on $\tilde{\xi}^3(y,z,t)$ given by (\ref{eq:xi_3}) we obtain: 
\begin{equation} 
\begin{split}
	&i \de_t a^1(y,t) = \frac{1}{2} \de_{p_\alpha} \de_{p_\beta} E_n(p(t)) (- i \de_{y_\alpha})(- i \de_{y_\beta}) a^1(y,t) + \frac{1}{2} \de_{q_\alpha} \de_{q_\beta} W(q(t)) y_\alpha y_\beta a^1(y,t)  	\\
	&+ \nabla_q W(q(t)) \cdot \mathcal{A}_n(p(t)) a^1(y,t)	\\
	&+ \frac{1}{6} \de_{p_\alpha} \de_{p_\beta} \de_{p_\gamma} E_n(p(t)) (- i \de_{y_\alpha})(- i \de_{y_\beta})(- i \de_{y_\gamma}) a^0(y,t) + \frac{1}{6} \de_{q_\alpha} \de_{q_\beta} \de_{q_\gamma} W(q(t)) y_\alpha y_\beta y_\gamma a^0(y,t) 	\\
	&+ \kappa_\gamma(t) (- i \de_{y_\gamma}) a^0(y,t) + \de_{q_\beta} \de_{q_\gamma} W(q(t)) \mathcal{A}_{n,\beta}(p(t)) y_\gamma a^0(y,t)
\end{split}
\end{equation}
which is precisely (\ref{eq:equation_for_a_1}) with the coefficient multiplying $(- i \de_{y_\gamma}) a^0(y,t)$ replaced by:
\begin{equation} \label{eq:coefficient}
\begin{split}
	&\kappa_\gamma(t) := i \de_{q_\beta} W(q(t)) \left( \ipz{\chi_n(z;p(t))}{ \de_{p_\beta} \de_{p_\gamma} \chi_n(z;p(t)) } \vphantom{\frac{1}{2}} \right. \\
	&- \ipz{\chi_n(z;p(t))}{ \de_{p_\gamma} \chi_n(z;p(t)) } \ipz{\chi_n(z;p(t))}{ \de_{p_\beta} \chi_n(z;p(t)) } 	\\
	&\left.\vphantom{\frac{1}{2}}  - i \ipz{\chi_n(z;p(t))}{ \left( (p(t) - i \de_z)_\gamma- \de_{p_\gamma} E_n(p(t)) \right) [H(p(t)) - E_n(p(t))]^{-1} P^\perp_n(p(t)) \de_{p_\beta} \chi_n(z;p(t)) } \right).
\end{split}
\end{equation}
We claim that:
\begin{equation} \label{eq:claim}
	\kappa_\gamma(t) = \de_{q_\beta} W(q(t)) \de_{p_\gamma} \mathcal{A}_{n,\beta}(p(t)).
\end{equation}
Adding and subtracting $i \de_{q_\beta} W(q(t)) \ipz{\de_{p_\gamma} \chi_n(z;p(t))}{\de_{p_\beta} \chi_n(z;p(t))}$ in (\ref{eq:coefficient}), we have that:
\begin{equation}
	\kappa_\gamma(t) = \de_{q_\beta} W(q(t)) \de_{p_\gamma} \mathcal{A}_{n,\beta}(p(t)) + \tilde{\kappa}_\gamma(t)		
\end{equation}
where:
\begin{equation} \label{eq:tilde_kappa}
\begin{split}
	&\tilde{\kappa}_\gamma(t) := i \de_{q_\beta} W(q(t)) \left( - \ipz{\de_{p_\gamma} \chi_n(z;p(t))}{\de_{p_\beta} \chi_n(z;p(t))} \vphantom{\frac{1}{2}} \right. \\
	&- \ipz{\chi_n(z;p(t))}{ \de_{p_\gamma} \chi_n(z;p(t)) } \ipz{\chi_n(z;p(t))}{ \de_{p_\beta} \chi_n(z;p(t)) }  	\\
	&\left.\vphantom{\frac{1}{2}}  - i \ipz{\chi_n(z;p(t))}{ \left( (p(t) - i \de_z)_\gamma- \de_{p_\gamma} E_n(p(t)) \right) [H(p(t)) - E_n(p(t))]^{-1} P^\perp_n(p(t)) \de_{p_\beta} \chi_n(z;p(t)) } \right).
\end{split}	
\end{equation} 
Using self-adjointness of the operators:
\begin{equation}
	(p(t) - i \de_z)_\gamma - \de_{p_\gamma} E_n(p(t)), [H(p(t)) - E_n(p(t))]^{-1} P^\perp_n(p(t)) 
\end{equation}
on $L^2_{per}$ for each $t \geq 0$, and then identity (\ref{eq:first_derivative_identities}) we have that the last term in (\ref{eq:tilde_kappa}) is equal to:
\begin{equation} \label{eq:long_inner_product_term}
\begin{split}
	&- \ipz{ [H(p(t)) - E_n(p(t))]^{-1} P^\perp_n(p(t)) \left( (p(t) - i \de_z)_\gamma- \de_{p_\gamma} E_n(p(t)) \right) \chi_n(z;p(t))}{ \de_{p_\beta} \chi_n(z;p(t)) }	\\
	&= \ipz{ [H(p(t)) - E_n(p(t))]^{-1} P^\perp_n(p(t)) [H(p(t)) - E_n(p(t))] \de_{p_\gamma} \chi_n(z;p(t))}{ \de_{p_\beta} \chi_n(z;p(t)) }. 
\end{split}
\end{equation}
It is clear that the operators $P^\perp_n(p(t)), [H(p(t)) - E_n(p(t))]$ commute on $L^2_{per}$ for any $t \geq 0$. We have therefore that this term: 
\begin{equation} \label{eq:long_inner_product_term_again}
	= \ipz{ P^\perp_n(p(t)) \de_{p_\gamma} \chi_n(z;p(t))}{ \de_{p_\beta} \chi_n(z;p(t)) }.
\end{equation}
Substituting (\ref{eq:long_inner_product_term_again}) into (\ref{eq:tilde_kappa}) we obtain:
\begin{equation} \label{eq:kappa_again_again}
\begin{split}
	&\tilde{\kappa}_\gamma(t) = i \de_{q_\beta} W(q(t)) \left[ - \ipz{\de_{p_\gamma} \chi_n(z;p(t))}{ \de_{p_\beta} \chi_n(z;p(t)) } \vphantom{\frac{1}{2}} \right. \\
	&- \ipz{\chi_n(z;p(t))}{ \de_{p_\gamma} \chi_n(z;p(t)) } \ipz{\chi_n(z;p(t))}{\de_{p_\beta} \chi_n(z;p(t))}  	\\
	&\left. + \ipz{ P^\perp_n(p(t)) \de_{p_\gamma} \chi_n(z;p(t))}{ \de_{p_\beta} \chi_n(z;p(t)) } \right].
\end{split}
\end{equation}
Recall that $\chi_n(z;p(t))$ is assumed normalized: $\ipz{\chi_n(z;p(t))}{\chi_n(z;p(t))} = 1$. Differentiating this relation with respect to $p$ and using the definition of the $L^2$-inner product we obtain:
\begin{equation} \label{eq:an_ident}
	\overline{ \ipz{\chi_n(z;p(t))}{\de_{p_\gamma} \chi_n(z;p(t))} } = - \ipz{\chi_n(z;p(t))}{\de_{p_\gamma} \chi_n(z;p(t))}. 
\end{equation}
We now use conjugate linearity of the $L^2$-inner product in its first argument and the identity (\ref{eq:an_ident}) to re-write the expression inside the square brackets in (\ref{eq:kappa_again_again}) as:
\begin{equation}
\begin{split}
	&- \ipz{\de_{p_\gamma} \chi_n(z;p(t))}{ \de_{p_\beta} \chi_n(z;p(t)) }  \\
	&+ \ipz{ \ipz{\chi_n(z;p(t))}{ \de_{p_\gamma} \chi_n(z;p(t)) } \chi_n(z;p(t))}{\de_{p_\beta} \chi_n(z;p(t))}  	\\
	&+ \ipz{ P^\perp_n(p(t)) \de_{p_\gamma} \chi_n(z;p(t))}{ \de_{p_\beta} \chi_n(z;p(t)) } 
\end{split}
\end{equation}
which is clearly zero by definition of the orthogonal projection operator $P^\perp_n(p(t))$ (\ref{eq:def_of_P_perp}). (\ref{eq:tilde_kappa}) is therefore zero, and the claim (\ref{eq:claim}) holds. 
\end{proof}

\section{Proof of $L^\infty$ bounds on $z$-dependence of residual, uniform in $p \in S_n$} \label{app:uniform_bounds_on_z_dependence}
In this Appendix we provide details on how to bound the $z$-dependence of terms which appear in the residual (\ref{eq:r_uniform_and_L_2_norms}) in $L^\infty_z$, uniformly in $p \in S_n$, where:
\begin{equation} \label{eq:S_n_again}
	S_n := \{ p \in \field{R}^d : \inf_{m \neq n} | E_m(p) - E_n(p) | \geq M \}, \text{ and $M > 0$}. 
\end{equation}
We consider the problem of bounding a representative term: 
\begin{equation}
\begin{split}
	&\mathcal{J}_{\alpha \beta}(p) := \| g_{\alpha \beta}(z;p) \|_{L^\infty_z(\Omega)}	\\
	&g_{\alpha \beta}(z;p) := \left[ (p_\alpha - i \de_{z_\alpha}) - \de_{p_\alpha} E_n(p) \right] [H(p) - E_n(p)]^{-1} P^\perp_n(p) \de_{p_\beta} \chi_n(z;p). 
\end{split}
\end{equation} 
uniformly in $p \in S_n$. Note that although the maps $p \mapsto E_n(p)$ are periodic with respect to the lattice $\Lambda$, the map $p \mapsto g(z;p)$ is not. We claim that:
\begin{equation} \label{eq:claim_of_appendix}
	\sup_{p \in S_n} \mathcal{J}_{\alpha \beta}(p) = \| g_{\alpha \beta}(z;p) \|_{L^\infty_z(\Omega)} < \infty. 
\end{equation}

First, we define the `shifted' Sobolev norms for any vector $p \in \field{R}^d$ to be: 
\begin{equation}
	\| f(z) \|_{H^s_{z,p}(\Omega)} := \sum_{|j| \leq s} \| (p - i \de_z)^j f(z) \|_{L^2_z(\Omega)}. 
\end{equation}
For any fixed $p$, $H^s_{z,p}$ is equivalent to the standard norm $H^s_z$. Using Sobolev embedding we have that for any integer $s > \frac{d}{2}$:
\begin{equation}
	\| g_{\alpha \beta}(z;p) \|_{L^\infty_z} = \| e^{i p z} g_{\alpha \beta}(z;p) \|_{L^\infty_z} \leq C_{s,d} \| e^{i p z} g_{\alpha \beta}(z;p) \|_{H^s_z} = C_{s,d} \| g_{\alpha \beta}(z;p) \|_{H^s_{z,p}}. 
\end{equation}
where the constant $C_{s,d} > 0$ depends on $s$ and $d$ but is independent of $p$. We are therefore done if we can show that $\| g_{\alpha \beta}(z;p) \|_{H^s_{z,p}}$ can be bounded uniformly in $p \in S_n$ for some integer $s > d/2$. 

From the definition of the $H^s_{z,p}$-norms and periodicity of $\de_{p_\alpha} E_n(p)$ we have that for any integer $s \geq 1$, $p \in S_n$:
\begin{equation} \label{eq:norm_of_p_minus_i_de_z}
	\| (p_\alpha - i \de_{z_\alpha}) - \de_{p_\alpha} E_n(p) \|_{ H^s_{z,p} \rightarrow H^{s-1}_{z,p} } \leq C', 
\end{equation}
where $C' := 1 + \sup_{p \in S_n \cap \mathcal{B}} |E_n(p)|$. By elliptic regularity, we have that for any integer $s \geq 0$, all $p \in S_n$:
\begin{equation} \label{eq:norm_of_resolvent}
	\| [H(p) - E_n(p)]^{-1} P^\perp_n(p) \|_{ H^s_{z,p} \rightarrow H^{s+2}_{z,p} } \lesssim \frac{1}{M}. 
\end{equation}
Differentiating the eigenvalue equation (\ref{eq:eigenvalue_problem}) for $\chi_n(z;p)$ with respect to $p$, we have that $P^\perp_n(p) \de_{p_\beta} \chi_n(z;p)$ satisfies:
\begin{equation}
	[H(p) - E_n(p)] P^\perp_n(p) \de_{p_\beta} \chi_n(z;p) = - P^\perp_n(p) \left[ (p_\beta - i \de_{z_\beta}) - \de_{p_\beta} E_n(p) \right] \chi_n(z;p).
\end{equation} 
Again, by elliptic regularity, for all $p \in S_n$:
\begin{equation}
	\| P^\perp_n(p) \de_{p_\beta} \chi_n(z;p) \|_{H^{s+2}_p} \lesssim \frac{1}{M} \| \left[ (p_\beta - i \de_{z_\beta}) - \de_{p_\beta} E_n(p) \right] \chi_n(z;p) \|_{H^{s}_{z,p}}.
\end{equation} 
Using (\ref{eq:norm_of_p_minus_i_de_z}) we then have for all $p \in S_n$:
\begin{equation} \label{eq:bound_on_de_p_chi}
	\| P^\perp_n(p) \de_{p_\beta} \chi_n(z;p) \|_{H^{s+2}_p} \leq \frac{C'}{M} \| \chi_n(z;p) \|_{H^{s+1}_{z,p}}. 
\end{equation}
Combining (\ref{eq:norm_of_p_minus_i_de_z}), (\ref{eq:norm_of_resolvent}), and (\ref{eq:bound_on_de_p_chi}) we have: 
\begin{align} \label{eq:big_bound}
	\| g_{\alpha \beta}(z;p) \|_{H^s_{z,p}} &\leq C' \| [H(p) - E_n(p)]^{-1} P^\perp_n(p) \de_{p_\beta} \chi_n(z;p) \|_{H^{s + 1}_{z,p}} &\text{ (using (\ref{eq:norm_of_p_minus_i_de_z}))} \nonumber	 	\\
	&\lesssim \frac{C'}{M} \| P^\perp_n(p) \de_{p_\beta} \chi_n(z;p) \|_{H^{s - 1}_{z,p}} &\text{ (using (\ref{eq:norm_of_resolvent}))} \nonumber	\\
	&\lesssim \left( \frac{C'}{M} \right)^2 \| \chi_n(z;p) \|_{H^{s-2}_{z,p}} \quad . &\text{ (using (\ref{eq:bound_on_de_p_chi}))}
\end{align}

We now claim that for any integer $s$:
\begin{equation} \label{eq:periodicity_of_function_and_norm}
	\sup_{p \in S_n} \| \chi_n(z;p) \|_{H^{s}_{z,p}} = \sup_{p \in S_n \cap \mathcal{B}} \| \chi_n(z;p) \|_{H^{s}_{z,p}} < \infty. 
\end{equation}
By elliptic regularity it is clear that for any fixed $p \in \field{R}^d$ and fixed positive integer $s$ that: 
\begin{equation}
	\| \chi_n(z;p) \|_{H^{s}_{z,p}} < \infty. 
\end{equation}
Using smoothness of the map $p \mapsto \chi_n(z;p)$ in $S_n$ and compactness of the Brillouin zone $\mathcal{B}$ we have that:
\begin{equation}
	\sup_{p \in S_n \cap \mathcal{B}} \| \chi_n(z;p) \|_{H^{s}_{z,p}} < \infty. 
\label{eq:d14}\end{equation}
Since for any reciprocal lattice vector $b \in \Lambda^*$ we have that $\chi_n(z;p + b) = e^{- i b \cdot z} \chi_n(z;p)$, we then have that: 
\begin{equation}
	\text{for any $b \in \Lambda^*$, }\| \chi_n(z;p + b) \|_{H^{s}_{z,p + b}} = \| e^{- i b \cdot z} \chi_n(z;p) \|_{H^{s}_{z,p + b}} = \| \chi_n(z;p) \|_{H^{s}_{z,p}}.
\end{equation} 
The bound (\ref{eq:periodicity_of_function_and_norm}) follows. 

We now turn to completing the proof of (\ref{eq:claim_of_appendix}). Fix $\sigma$, a positive integer such that $\sigma > \max \{ \frac{d}{2}, 2 \}$. Then: 
\begin{align*}
	\sup_{p \in S_n} \mathcal{J}(p) &= \sup_{p \in S_n} \| g_{\alpha \beta}(z;p) \|_{L^\infty_z} &\text{ (by definition)}  \\
	&\leq C_{s,d} \sup_{p \in S_n} \| g_{\alpha \beta}(z;p) \|_{H^{\sigma}_{z,p}} &\text{ (by Sobolev embedding, since $\sigma > d/2$)}  	\\
	&\leq C_{s,d} \left( \frac{C'}{M} \right)^2 \sup_{p \in S_n} \| \chi_n(z;p) \|_{H^{\sigma - 2}_{z,p}}  &\text{ (by (\ref{eq:big_bound}), with $s = \sigma$)} 	\\
	&= C_{s,d} \left( \frac{C'}{M} \right)^2 \sup_{p \in \mathcal{B} \cap S_n} \| \chi_n(z;p) \|_{H^{\sigma - 2}_{z,p}}  &\text{ (using (\ref{eq:periodicity_of_function_and_norm}) with $s = \sigma - 2$)} \\
	&< \infty. & \text{(by (\ref{eq:d14}) with $s = \sigma - 2$)}
\end{align*}
All other $z$-dependence in expression (\ref{eq:r_uniform_and_L_2_norms}) for the residual may be bounded in $L^\infty_z$ uniformly in $p \in S_n$ by similar arguments. 

\section{Proof of Lemma \ref{lem:basic_homogenization_lemma}} \label{app:proof_of_lemma}
First, it is clear from changing variables in the integral that:
\begin{equation}
	\inty{\field{R}^d}{}{ f\left(x\right) g\left(\frac{x}{\delta} + \frac{c}{\delta^2}\right) }{x} = \left( \inty{\field{R}^d}{}{ f(x) }{x} \right) \left( \inty{\Omega}{}{ g(z) }{z} \right) + O(\delta^N)
\end{equation}
is equivalent to: 
\begin{equation}
	\inty{\field{R}^d}{}{ f\left(x - \frac{c}{\delta}\right) g\left(\frac{x}{\delta}\right) }{x} = \left( \inty{\field{R}^d}{}{ f(x) }{x} \right) \left( \inty{\Omega}{}{ g(z) }{z} \right) + O(\delta^N). 
\end{equation}
Let $v_j, j \in \{1,...,d\}$ denote generators of the lattice $\Lambda$ so that if $v \in \Lambda$, there exist unique integers $n_j, j \in \{1,...,d\}$ such that:
\begin{equation}
	v = n_1 v_1 + n_2 v_2 + ... + n_d v_d. 
\end{equation}
And let $b_j, j \in \{1,...,d\}$ denote generators of the dual lattice $\Lambda^*$ such that if $b \in \Lambda^*$, there exist unique integers $m_j, j \in \{1,...,d\}$ such that:
\begin{equation}
	b = m_1 b_1 + m_2 b_2 + ... + m_d b_d,
\end{equation}
and furthermore, for all $i, j \in \{1,...,d\}$, $b_i \cdot v_j = 2 \pi \delta_{ij}$. 

Since $g(z)$ is smooth and periodic with respect to the lattice $\Lambda$, it has a uniformly convergent Fourier series:
\begin{equation} \label{eq:fourier_series_expansion}
\begin{split}
	&g(z) = \sum_{(m_1,...,m_d) \in \field{Z}^d} g_{m_1,...,m_d} e^{i \left[ m_1 b_1 \cdot z + m_2 b_2 \cdot z + ... m_d b_d \cdot z \right]} 	\\
	&g_{m_1,...,m_d} = \inty{\field{R}^d / \Lambda}{}{ e^{- i \left[ m_1 b_1 \cdot z + m_2 b_2 \cdot z + ... m_d b_d \cdot z \right]} g(z) }{z}.
\end{split}
\end{equation}
We have therefore that:
\begin{equation} \label{eq:integral_in_terms_of_fourier_series}
	\inty{\field{R}^d}{}{ f\left(x - \frac{c}{\delta}\right) g\left(\frac{x}{\delta}\right) }{x} = \sum_{(m_1,...,m_d) \in \field{Z}^d} g_{m_1,...,m_d} \inty{\field{R}^d}{}{ e^{i \left[ m_1 b_1 \cdot x / \delta + m_2 b_2 \cdot x / \delta + ... m_d b_d \cdot x / \delta \right]} f(x - c/\delta)  }{x}	
\end{equation} 
where it is valid to change the order of summation of the series with the integration by uniform convergence of the series. We now write the right-hand side of (\ref{eq:integral_in_terms_of_fourier_series}) as:
\begin{equation} \label{eq:terms_in_series_split_up}
\begin{split}
	&= g_{0,...,0} \inty{\field{R}^d}{}{ f(x - c/\delta)  }{x}	\\
	&+ \sum_{(m_1,...,m_d) \in \field{Z}^d,(m_1,...,m_d) \neq (0,...0)} g_{m_1,...,m_d} \inty{\field{R}^d}{}{ e^{i \left[ m_1 b_1 \cdot x / \delta + m_2 b_2 \cdot x / \delta + ... m_d b_d \cdot x / \delta \right]} f(x - c/\delta) }{x}	
\end{split}
\end{equation}
By the definition of $g_{0,...,0}$ (\ref{eq:fourier_series_expansion}) and by a trivial change of variables we have that:
\begin{equation}
	g_{0,...,0} \inty{\field{R}^d}{}{ f(x - c/\delta)  }{x}	= \left( \inty{\field{R}^d}{}{ f(x)  }{x} \right)\left( \inty{\Omega}{}{ g(z) }{z} \right)
\end{equation}
To see that the second term in (\ref{eq:terms_in_series_split_up}) is of $O(\delta^N)$ for arbitrary $N \in \field{N}$, consider a representative term in the series where $(m_1,...,m_d) = (1,0,...,0)$:
\begin{equation}
	\left|	g_{1,0,...0} \inty{\field{R}^d}{}{ e^{i b_1 \cdot x / \delta } f(x - c/\delta) }{x} \right| = \left| g_{1,0,...0} \inty{\field{R}^d}{}{ \left[ \left( \frac{- i \delta v_1 \cdot \nabla_x}{2 \pi} \right)^N e^{i b_1 \cdot x / \delta } \right] f(x - c/\delta) }{x} \right|	
\end{equation}
Integrating by parts gives:
\begin{equation}
\begin{split}
	&= \left| g_{1,0,...0} \inty{\field{R}^d}{}{ e^{i b_1 \cdot x / \delta } \left[ \left( \frac{i \delta v_1 \cdot \nabla_x}{2 \pi} \right)^N  f(x - c/\delta) \right] }{x} \right|	\\
	&\leq \delta^N \frac{1}{(2\pi)^N} | g_{1,0,...0} | \inty{\field{R}^d}{}{ \left| \left( v_1 \cdot \nabla_x \right)^N  f(x - c/\delta) \right| }{x} 	\\
\end{split}
\end{equation}
Using the definition of $g_{1,0,...0}$ (\ref{eq:fourier_series_expansion}) and another change of variables, we have:
\begin{equation}
	\leq \delta^N \frac{1}{(2\pi)^N} \inty{\field{R}/\Lambda}{}{ | g(z) | }{z} \inty{\field{R}^d}{}{ \left| \left( v_1 \cdot \nabla_x \right)^N  f(x) \right| }{x} 	
\end{equation}
since $f \in \mathcal{S}(\field{R}^d)$, we are done:
\begin{equation}
	\left| \inty{\field{R}^d}{}{ f\left(x + \frac{c}{\delta}\right) g\left(\frac{x}{\delta}\right) }{x} - \left( \inty{\field{R}^d}{}{ f(x)  }{x} \right)\left( \inty{\Omega}{}{ g(z) }{z} \right)	\right| \leq C_{N,f,g} \delta^N
\end{equation}
where $C_N > 0$ is a positive constant which depends on $N, f, g$ but not $\delta$. 

\section{Computation of dynamics of physical observables} \label{app:computation_of_dynamics_of_physical_observables}
In this Appendix we compute: 
\begin{equation}
\begin{split}
	&\fdf{t} \left[ \ipy{b(y,t)}{a(y,t)} + \ipy{a(y,t)}{b(y,t)} \right]	\\
	&\fdf{t} \left[ \ipy{a(y,t)}{y a(y,t)} \right], \fdf{t} \left[ \ipy{a(y,t)}{(- i \nabla_y) a(y,t)} \right]	\\
	&\fdf{t} \left[ \ipy{b(y,t)}{y a(y,t) } + \ipy{a(y,t)}{y b(y,t)} \right]	\\
	&\fdf{t} \left[ \ipy{b(y,t)}{(- i \nabla_y) a(y,t)} + \ipy{a(y,t)}{ (- i \nabla_y) b(y,t)} \right] 	\\
\end{split}
\end{equation}
We will make use of the following simple lemmas which are each elementary to prove:
\begin{lemma} \label{lem:basic_lemma}
Let $a(y,t)$ satisfy:
\begin{equation} \label{eq:equation_for_psi}
	i \de_t a = \mathscr{H}(t) a
\end{equation}
where $\mathscr{H}(t)$ is self-adjoint for every $t$. Let $\mathcal{G}(t)$ denote a physical observable defined by: 
\begin{equation}
	\mathcal{G}(t) := \inty{\field{R}^d}{}{ \overline{a(y,t)} {G} a(y,t) }{y} 
\end{equation} 
where ${G}$ is self-adjoint. Then: 
\begin{equation}
	\dot{\mathcal{G}}(t) = i \inty{\field{R}^d}{}{ \overline{a(y,t)} [\mathscr{H}(t), {G}] a(y,t) }{y} 
\end{equation}
\end{lemma} 
\begin{lemma} \label{lem:less_basic_lemma}
Let $a(y,t)$ satisfy (\ref{eq:equation_for_psi}), and $b(y,t)$ satisfy:
\begin{equation}
	i \de_t b = \mathscr{H}(t) b + \mathscr{I}(t) a
\end{equation}
where $\mathscr{I}(t)$ is self-adjoint for every $t$. Define:
\begin{equation}
	\mathcal{K}(t) := \inty{\field{R}^d}{}{ \overline{ b(y,t) } {K} a(y,t) }{y} + \inty{\field{R}^d}{}{ \overline{ a(y,t) } {K} b(y,t) }{y}
\end{equation}
where ${K}$ is self-adjoint. Then:
\begin{equation}
\begin{split}
	&\dot{\mathcal{K}}(t) = i \inty{\field{R}^d}{}{ \overline{b(y,t)} [\mathscr{H}(t),{K}] a(y,t) }{y} + i \inty{\field{R}^d}{}{ \overline{a(y,t)} [\mathscr{H}(t),{K}] b(y,t)}{y} \\  
	&+ i \inty{\field{R}^d}{}{ \overline{a(y,t)} [\mathscr{I}(t),{K}] a(y,t) }{y}
\end{split}
\end{equation}
\end{lemma}
\begin{lemma}
Let ${G}_1, {G}_2, {G}_3$ be operators. Then:
\begin{equation}
\begin{split}
	&[{G}_1 {G}_2,{G}_3] = {G}_1 [{G}_2,{G}_3] + [{G}_1,{G}_3] {G}_2	\\
	&[{G}_1,{G}_2 {G}_3] = [{G}_1,{G}_2] {G}_3 + {G}_2 [{G}_1,{G}_3]	
\end{split}
\end{equation} 
\end{lemma}
\begin{lemma}
\begin{equation}
	[(- i \de_{y_\alpha}),y_\beta] = - i \delta_{\alpha \beta}, [y_\alpha,(- i \de_{y_\beta})] = i \delta_{\alpha \beta}. 
\end{equation}
\end{lemma}
We will apply the Lemmas with:
\begin{equation}
	\mathscr{H}(t) = \frac{1}{2} \de_{p_\alpha} \de_{p_\beta} E_n(p(t)) (- i \de_{y_\alpha}) (- i \de_{y_\beta}) + \frac{1}{2} \de_{q_\alpha} \de_{q_\beta} W(q(t)) y_\alpha y_\beta \\
\end{equation}
For any operator ${G}$, we have: 
\begin{equation} \label{eq:commuting_G_with_H}
\begin{split}
	&i [ \mathscr{H}(t) , {G} ] = \frac{1}{2} i \de_{p_\alpha} \de_{p_\beta} E_n(p(t)) (- i \de_{y_\alpha}) [(- i \de_{y_\beta}),{G}] + \frac{1}{2} i \de_{p_\alpha} \de_{p_\beta} E_n(p(t)) [(- i \de_{y_\alpha}),{G}] (- i \de_{y_\beta}) \\
	&+ \frac{1}{2} i \de_{q_\alpha} \de_{q_\beta} W(q(t))(t) y_\alpha [ y_\beta , {G}] + \frac{1}{2} i \de_{q_\alpha} \de_{q_\beta} W(q(t))(t) [y_\alpha ,{G}] y_\beta		\\
\end{split}
\end{equation}
From (\ref{eq:commuting_G_with_H}), we have: 
\begin{equation}
\begin{split}
	&i [ \mathscr{H}(t) , y_\alpha ] = \de_{p_\alpha}\de_{p_\beta} E_n(p(t)) (- i \de_{y_\beta}) 	\\
	&i [ \mathscr{H}(t) , (- i \de_{y_\alpha}) ] = - \de_{q_\alpha}\de_{q_\beta} W(q(t)) y_\beta 
\end{split}
\end{equation}
So that, using Lemma \ref{lem:basic_lemma}:
\begin{equation} \label{eq:Q_dot_P_dot_a}
\begin{split}
	&\fdf{t} \ipy{a(y,t)}{y_\alpha a(y,t)} = \de_{p_\alpha} \de_{p_\beta} E_n(p(t)) \ipy{a(y,t)}{(- i \de_{y_\beta}) a(y,t)}	 \\
	&\fdf{t} \ipy{a(y,t)}{(- i \de_{y_\alpha}) a(y,t)} = - \de_{q_\alpha} \de_{q_\beta} W(q(t)) \ipy{a(y,t)}{y_\beta a(y,t)}.
\end{split}
\end{equation}
Note that it follows from (\ref{eq:Q_dot_P_dot_a}) that:
\begin{equation} \label{eq:Q_dot_P_dot_a_corollary}
\begin{split}
	&\ipy{a(y,0)}{y a(y,0)} = \ipy{a(y,0)}{(- i \nabla_y) a(y,0)} = 0	\\
	&\implies \text{for all } t \geq 0, \ipy{a(y,t)}{y a(y,t)} = \ipy{a(y,t)}{(- i \nabla_y) a(y,t)} = 0.
\end{split}
\end{equation}
We will then apply Lemma \ref{lem:less_basic_lemma} with:
\begin{equation}
\begin{split}
	&\mathscr{I}(t) = \frac{1}{6} \de_{p_\alpha} \de_{p_\beta} \de_{p_\gamma} E_n(p(t)) (- i \de_{y_\alpha})(- i \de_{y_\beta})( - i \de_{y_\gamma}) + \frac{1}{6} \de_{q_\alpha} \de_{q_\beta} \de_{q_\gamma} W(q(t)) y_\alpha y_\beta y_\gamma 	\\ 
	&+ \de_{p_\beta} \left[ \nabla_{q} W(q(t)) \cdot \mathcal{A}_{n}(p(t)) \right] (- i \de_{y_\beta}) + \de_{q_\beta} \left[ \nabla_{q} W(q(t)) \cdot \mathcal{A}_{n}(p(t)) \right] y_\beta \\
\end{split}
\end{equation}
Calculating the commutators:
\begin{equation}
\begin{split}
	&i [\mathscr{I}(t),1] = 0	\\
	&i [\mathscr{I}(t),y_\alpha] = \frac{1}{2} \de_{p_\alpha} \de_{p_\beta} \de_{p_\gamma} E_n(p(t)) (- i \de_{y_\beta}) (- i \de_{y_\gamma}) + \de_{p_\alpha} \left[ \nabla_{q} W(q(t)) \cdot \mathcal{A}_{n}(p(t)) \right] \\
	&i [\mathscr{I}(t),(- i \de_{y_\alpha})] = - \frac{1}{2} \de_{q_\alpha} \de_{q_\beta} \de_{q_\gamma} W(q(t)) y_\beta y_\gamma - \de_{q_\alpha} \left[ \nabla_{q} W(q(t)) \cdot \mathcal{A}_{n}(p(t))  \right] \\
\end{split}
\end{equation}
We have, by Lemma \ref{lem:less_basic_lemma}:
\begin{equation} \label{eq:de_b_a_1}
	\fdf{t} \left[ \ipy{b(y,t)}{a(y,t)} + \ipy{a(y,t)}{b(y,t)} \right] = 0	
\end{equation}
\begin{equation} \label{eq:a_b_more}
\begin{split}
	&\fdf{t} \left[ \ipy{b(y,t)}{y_\alpha a(y,t)} + \ipy{a(y,t)}{y_\alpha b(y,t)} \right]	\\
	&= \de_{p_\alpha} \de_{p_\beta} E_n(p(t)) \left[ \ipy{b(y,t)}{(- i \de_{y_\alpha}) a(y,t)} + \ipy{a(y,t)}{(- i \de_{y_\alpha}) b(y,t)} \right] \\ 
	&+ \frac{1}{2} \de_{p_\alpha} \de_{p_\beta} \de_{p_\gamma} E_n(p(t)) \ipy{ a(y,t) }{ (- i \de_{y_\beta}) (- i \de_{y_\gamma}) a(y,t) } \\
	&+ \de_{p_\alpha} \left[ \nabla_{q} W(q(t)) \cdot \mathcal{A}_{n}(p(t)) \right] \| a(y,t) \|^2_{L^2_y(\field{R}^d)} \\
	&\fdf{t} \left[ \ipy{b(y,t)}{(- i \de_{y_\alpha}) a(y,t)} + \ipy{a(y,t)}{(- i \de_{y_\alpha}) b(y,t)} \right] 	\\
	&= - \de_{q_\alpha} \de_{q_\beta} W(q(t)) \left[ \ipy{ b(y,t) }{ y_\beta a(y,t) } + \ipy{ a(y,t) }{ y_\beta b(y,t) } \right] \\ 
	&- \frac{1}{2} \de_{q_\alpha} \de_{q_\beta} \de_{q_\gamma} W(q(t)) \ipy{a(y,t)}{y_\beta y_\gamma a(y,t)} \\
	&- \de_{q_\alpha} \left[ \nabla_{q} W(q(t)) \cdot \mathcal{A}_{n}(p(t)) \right] \| a(y,t) \|^2_{L^2_y(\field{R}^d)}
\end{split}
\end{equation}

\printbibliography

\end{document}